\documentclass[11pt,letterpaper]{article}
\usepackage{amsmath,amssymb,amsfonts,amsthm}
\usepackage[mathscr]{euscript}
\usepackage{tikz}
\usepackage{tikz-cd}
\usepackage{float}
\usepackage{dsfont}
\usepackage{enumitem}
\usepackage{bbm}
\usepackage{xcolor}
\usepackage[ruled,vlined]{algorithm2e}
\usepackage[margin=1in]{geometry}
\usepackage[colorlinks, citecolor = blue, linkcolor = magenta]{hyperref}
\usepackage{cleveref}
\numberwithin{equation}{section}

\allowdisplaybreaks

\newtheorem{theorem}{Theorem}[section]
\newtheorem{lemma}[theorem]{Lemma}
\newtheorem{corollary}[theorem]{Corollary}
\newtheorem{conjecture}[theorem]{Conjecture}

\theoremstyle{definition}
\newtheorem{remark}{Remark}
\newtheorem{proposition}[theorem]{Proposition}
\newtheorem{definition}[theorem]{Definition}
\newtheorem{problem}[theorem]{Problem}

\renewcommand{\leq}{\leqslant}
\renewcommand{\le}{\leqslant}
\renewcommand{\geq}{\geqslant}
\renewcommand{\ge}{\geqslant}

\newcommand{\cA}{\ensuremath{\mathcal{A}}}

\newcommand{\cD}{\ensuremath{\mathcal{D}}}
\newcommand{\cE}{\ensuremath{\mathcal{E}}}

\newcommand{\cG}{\ensuremath{\mathcal{G}}}
\newcommand{\cH}{\ensuremath{\mathcal{H}}}

\newcommand{\cJ}{\ensuremath{\mathcal{J}}}

\newcommand{\cP}{\ensuremath{\mathcal{P}}}

\newcommand{\cS}{\ensuremath{\mathcal{S}}}

\newcommand{\cW}{\ensuremath{\mathcal{W}}}

\newcommand{\bE}{\ensuremath{\mathbb{E}}}
\newcommand{\bF}{\ensuremath{\mathbb{F}}}
\newcommand{\bN}{\ensuremath{\mathbb{N}}}
\newcommand{\bP}{\ensuremath{\mathbb{P}}}
\newcommand{\bR}{\ensuremath{\mathbb{R}}}

\newcommand{\bcJ}{\boldsymbol{\cJ}}
\newcommand{\bcP}{\boldsymbol{\cP}}
\newcommand{\bcS}{\boldsymbol{\cS}}
\newcommand{\bcW}{\boldsymbol{\cW}}

\newcommand{\supp}{\mathrm{supp}}
\newcommand{\yes}{\mathrm{yes}}
\newcommand{\no}{\mathrm{no}}

\newcommand{\ex}{\mathrm{ex}}

\newcommand{\frakP}{\mathfrak{P}}

\DeclareMathOperator*{\argmax}{argmax}

\newcommand{\nil}{\textup{\texttt{nil}}}

\newcommand{\dsam}{\mathsf{dsam}}
\newcommand{\sam}{\mathsf{sam}}

\newcommand{\Wedge}{\mathsf{Wedge}}

\newcommand{\Square}{\mathsf{Square}}
\newcommand{\Walk}{\mathsf{Walk}}
\newcommand{\Hop}{\mathsf{Hop}}
\newcommand{\HopD}{\mathsf{HopD}}

\newcommand{\Var}[1]{\mathrm{Var} \left[ #1 \right]}

\newcommand{\Ex}[1]{\bE \left[ #1 \right]}

\newcommand{\Exs}[2]{\bE_{#1}\left[ #2 \right]}
\renewcommand{\Pr}[1]{\bP \left[ #1 \right]} 
\newcommand{\Pru}[2]{\underset{ #1 }\bP \left[ #2 \right]}
\newcommand{\Prs}[2]{\bP_{ #1 }\left[ #2 \right]}

\newcommand{\ind}[1]{\mathds{1} \left[ #1 \right] }

\title{Testing Properties of Edge Distributions}
\author{Yumou Fei\thanks{Department of EECS, Massachusetts Institute of Technology.}}
\date{}

\begin{document}

\maketitle 

\begin{abstract}
We initiate the study of distribution testing for probability distributions over the edges of a graph, motivated by the closely related question of ``edge-distribution-free'' graph property testing. The main results of this paper are nearly-tight bounds on testing bipartiteness, triangle-freeness and square-freeness of edge distributions, whose sample complexities are shown to scale as $\Theta(n)$, $n^{4/3\pm o(1)}$ and $n^{9/8\pm o(1)}$, respectively.

The technical core of our paper lies in the proof of the upper bound for testing square-freeness, wherein we develop new techniques based on certain birthday-paradox-type lemmas that may be of independent interest. We will discuss how our techniques fit into the general framework of distribution-free property testing. We will also discuss how our results are conceptually connected with Tur\'an problems and subgraph removal lemmas in extremal combinatorics.
\end{abstract}

\newpage 

\tableofcontents

\newpage 

\section{Introduction}
Suppose $\Lambda$ is a finite set, and $\cP$ is a class of probability distributions on $\Lambda$. In the standard model of distribution testing \cite{goldreich1998property,batu2000testing}, given sampling access to an unknown distribution $\mu$ over $\Lambda$, an algorithm should accept with probability at least $2/3$ if $\mu$ belongs to the class $\cP$, and reject if $\mu$ has total variation distance at least $\varepsilon$ to any distribution in $\cP$. 

In many central problems such as uniformity testing and identity testing (see e.g. the survey \cite{canonne2022topics} and references therein), the domain $\Lambda$ is a general \emph{unstructured} set, i.e. there are no relations among the elements of $\Lambda$. However, there are also important examples of problems where the domain $\Lambda$ is endowed with a certain structure. For example, in monotonicity testing of distributions (introduced by \cite{batu2004sublinear}), the domain is assumed to be a partially ordered set. Some concrete domain structures, such as the hypercube $\Lambda=\{0,1\}^{n}$, have also been studied in the literature for various problems.

In this paper, we initiate the study of distribution testing with domain
\[
\Lambda=\binom{[n]}{2}=\{\text{two-element subsets of }[n]\}.
\]
A distribution over such a domain can be viewed as an edge-weighted graph on $n$ vertices, and a random sample from the distribution is a random edge from the graph generated with probabilities proportional to the edge weights.

We focus the present study on properties of distributions that are characterized solely by the \emph{support} of the unknown distribution $\mu$, i.e.
\[
\supp(\mu):=\big\{x\in \Lambda\,\big|\, \mu(\{x\})>0\big\}.
\]
If the domain $\Lambda$ is a general unstructured set, then the only (symmetric) information about the support is its cardinality. Indeed, the problem of estimating the support size (or testing whether the support size is at most some value) of distributions has been extensively studied (see e.g. \cite{valiant2017estimating,ferreira2025testing}). Since in our context the support of a distribution is the edge set of a graph, instead of a bare subset of an unstructured domain, one can study much richer properties of the support such as bipartiteness and subgraph-freeness.

For the sake of convenience, we say an edge distribution $\mu$ over $\binom{[n]}{2}$ satisfies a certain graph property (such as triangle-freeness) if the support of $\mu$ satisfies that property. The main results of this paper are nearly-tight bounds on the sample complexities of testing bipartiteness, triangle-freeness and square-freeness of edge distributions:

\begin{theorem}[Informal]\label{thm:main_informal}
The sample complexities of testing bipartiteness, triangle-freeness and square-freeness of edge distributions on $n$ vertices are $\Theta(n)$, $n^{4/3\pm o(1)}$ and $n^{9/8\pm o(1)}$, respectively.
\end{theorem}

\subsection{Distribution-Free Testing of Functions}

In this subsection, we show how our results relate to \emph{distribution-free property testing (of functions)}. We first present the following standard formalization of the distribution testing model used in Theorem~\ref{thm:main_informal}.

\begin{definition}\label{def:support_testing}
Suppose $\cH$ is a nonempty (downward-closed) family of subsets of a finite domain $\Lambda$. For any parameter $\varepsilon\in (0,1)$, we define $\dsam(\cH,\varepsilon)$ to be the minimum possible value of positive integer $m$ such that the following holds: there exists an algorithm that for any distribution $\mu$ over $\Lambda$, takes $m$ independent samples from $\mu$ and
\begin{enumerate}[label=(\arabic*)]
\item accepts with probability at least $2/3$ if $\supp(\mu)\in \cH$;
\item rejects with probability at least $2/3$ if $\left\|\mu-\nu\right\|_{\mathrm{TV}}\geq \varepsilon$ for any distribution $\nu$ over $\Lambda$ with $\supp(\nu)\in \cH$.
\end{enumerate}
\end{definition}

Distribution-free property testing (of functions) was first introduced by Goldreich, Goldwasser and Ron \cite{goldreich1998property} and has been studied extensively. The (sample-based) distribution-free property testing model is defined as follows.

\begin{definition}[{\cite[Definition 2.1]{goldreich1998property}}]\label{def:distribution_free_testing}
Suppose $\cH$ is a nonempty family of Boolean-valued functions on a finite domain $\Lambda$.\footnote{By identifying a subset of $\Lambda$ with its indicator function, we view a family of Boolean-valued functions interchangeably as a family of subsets of the domain.} For any parameter $\varepsilon\in (0,1)$, we define $\sam(\cH,\varepsilon)$ to be the minimum possible value of positive integer $m$ such that the following holds: there exists an algorithm that for any distribution $\mu$ over $\Lambda$ and any function $f:\Lambda\rightarrow\{0,1\}$, takes $m$ independent $f$-labeled samples $\big(x^{(1)},f(x^{(1)})\big),\dots,\big(x^{(m)},f(x^{(m)})\big)$, where each $x^{(i)}$ is drawn independently from $\mu$, and
\begin{enumerate}[label=(\arabic*)]
\item accepts with probability at least $2/3$ if $f\in \cH$;
\item rejects with probability at least $2/3$ if $\Prs{x\sim \mu}{f(x)\neq g(x)}\geq \varepsilon$ for any function $g\in \cH$. 
\end{enumerate}
\end{definition}

It is easy to observe the following relation between Definitions~\ref{def:support_testing} and~\ref{def:distribution_free_testing} (see Section~\ref{sec:proof_of_easy_equivalence} for a proof).

\begin{proposition}\label{prop:downward_closed_equivalence}
Suppose $\cH$ is a downward-closed family of subsets of a finite domain $\Lambda$. For any parameter $\varepsilon\in (0,1)$, we have
\[
\dsam(\cH,\varepsilon)\leq \sam(\cH,\varepsilon)\leq \frac{20}{\varepsilon}\cdot \big(\dsam(\cH,\varepsilon)+1\big).
\]
\end{proposition}

Therefore, the distribution testing problem in Definition~\ref{def:support_testing} can basically be viewed as the special case of (sample-based) distribution-free property testing of Boolean-valued functions where the property is downward-closed. In the rest of the paper, we will mostly work with Definition~\ref{def:distribution_free_testing} instead of Definition~\ref{def:support_testing}.

Our main theorem (Theorem~\ref{thm:main_informal}) can then be formalized as follows.

\begin{theorem}[Formal version of Theorem~\ref{thm:main_informal}]\label{thm:main}
For positive integers $n$, let $\cG^{\mathrm{bip}}_{n}$, $\cG^{\mathrm{tri}}_{n}$ and $\cG^{\mathrm{squ}}_{n}$ be the collection of bipartite, triangle-free and square-free subsets of $\binom{[n]}{2}$, respectively. For any $\varepsilon\in (0,\frac{1}{10})$, we have\footnote{The lower bound $\sam\big(\cG^{\text{bip}}_{n},1/9\big)\geq \Omega(n)$ was already proven by Goldreich and Ron \cite[Theorem 4.6]{goldreich2016sample} even in the case where the unknown distribution over $\binom{[n]}{2}$ is uniform. They also proved an $O(n/\varepsilon)$ upper bound in the uniform distribution case; our result $\sam\big(\cG^{\text{bip}}_{n},\varepsilon\big)\leq O(n/\varepsilon)$ extends their upper bound to the distribution-free setting.}
\begin{align}
\Omega(n)&\leq \sam\big(\cG^{\mathrm{bip}}_{n},\varepsilon\big)\leq O(n/\varepsilon),\label{eq:bipartiteness_result}\\
n^{4/3}\exp\left(-O\left(\sqrt{\log n}\right)\right)&\leq\sam\big(\cG^{\mathrm{tri}}_{n},\varepsilon\big)\leq O(n^{4/3}/\varepsilon),\text{ and}\label{eq:triangle_result}\\
n^{9/8}\exp\left(-O\left(\sqrt{\log n}\right)\right)&\leq\sam\big(\cG^{\mathrm{squ}}_{n},\varepsilon\big)\leq O(n^{9/8}/\varepsilon).\label{eq:square_result}
\end{align}
\end{theorem}

\begin{remark}\label{rem:tester_with_query}
In addition to the $f$-labeled sampling access to $\mu$ as described in Definition~\ref{def:distribution_free_testing}, one can also allow the distribution-free property tester to \emph{query} the function $f$ on any input $x\in \Lambda$ and receive the value $f(x)$. Viewing an $f$-labeled sample as also containing a ``query,'' the total query complexity of such an algorithm is the sum of the number of samples taken and the number of additional queries made. In this paper, unless otherwise stated, we assume the property testers to be \emph{sample-based}, i.e. they do not have the power to make oracle queries for function values.
\end{remark}

\subsection{Additional Results}

In this subsection, we present a few additional results complementing Theorem~\ref{thm:main}. The first additional result is the following generalization of the bipartiteness-testing result \eqref{eq:bipartiteness_result}: the sample complexity of testing any \emph{graph-homomorphism property} is $\Theta(n)$.

\begin{theorem}\label{thm:graph_homomorphism}
Let $H$ be a fixed simple graph with at least one edge. For positive integers $n$, let $\cG_{n}^{H\textup{-hom}}$ be the collection of edge sets $E\subseteq \binom{[n]}{2}$ such that there is a graph homomorphism from the graph $([n],E)$ to $H$. Then there exists a constant $\varepsilon_{0}\in (0,1)$ depending only on $H$ such that for any $\varepsilon\in (0,\varepsilon_{0}]$ we have
\[
\Omega(n)\leq \sam\big(\cG^{H\textup{-hom}}_{n},\varepsilon\big)\leq O(n/\varepsilon).
\]
\end{theorem}

A generalization of graph-homomorphism properties is the class of \emph{semi-homogeneous graph partition properties} \cite{fiat2021efficient}, but the $\Theta(n)$ sample complexity in Theorem~\ref{thm:graph_homomorphism} does not generalize to this class of properties. In fact, the property of being a clique belongs to this class, and our next result shows that testing it requires only $\Theta(n^{2/3})$ samples.\footnote{The property of being a clique is in fact a \emph{homogeneous} graph partition property, as defined in \cite{fiat2021efficient}.}
\begin{theorem}\label{thm:clique}
For positive integers $n$, let $\cG^{\textup{cliq}}_{n}$ be the collection of subsets of $\binom{[n]}{2}$ that correspond to cliques. For any $\varepsilon\in (0,\frac{1}{10})$, we have
\[
\Omega(n^{2/3})\leq \sam\big(\cG^{\textup{cliq}}_{n},\varepsilon\big)\leq O(n^{2/3}/\varepsilon).
\]
\end{theorem}
Note that since $\cG_{n}^{\textup{cliq}}$ is not downward-closed for $n\geq 3$, Theorem~\ref{thm:clique} cannot be formulated as a ``distribution testing'' result via Proposition~\ref{prop:downward_closed_equivalence} (while Theorem~\ref{thm:graph_homomorphism} can).

Turning to subgraph-freeness properties, however, we are unable to determine the sample complexity of testing $H$-freeness for every fixed graph $H$. A relatively simple special case that we are able to solve is when $H$ is a tree:

\begin{theorem}\label{thm:tree-freeness}
For any simple graph $H$ with at least one edge and any positive integer $n$, let $\cG^{H\textup{-free}}_{n}$ be the collection of $H$-free subsets of $\binom{[n]}{2}$. If $H$ is a fixed tree with $t$ edges, there exists a constant $\varepsilon_{0}$ depending only on $t$ such that for any $\varepsilon\in (0,\varepsilon_{0}]$ we have
\[
\Omega(n^{(t-1)/t})\leq \sam\big(\cG^{H\textup{-free}}_{n},\varepsilon\big)\leq O(n^{(t-1)/t}/\varepsilon).
\]
\end{theorem}

\subsection{Related Work}\label{subsec:related_work}

Possibly due to the close connection with PAC learning, many studies on distribution-free property testing focused on functions on the hypercube $\{0,1\}^{n}$ (e.g. \cite{glasner2009distribution,dolev2011distribution,chen2016tight,blais2021vc,chen2022distribution,chen2024distribution}). Nevertheless, distribution-free models have also been considered in the context of graph property testing, as we discuss below.

A few papers \cite{goldreich2019testing,gishboliner2019testing} studied a model called ``vertex-distribution-free'' graph property testing. In that model, the unknown distribution is over the vertices of a graph, and (more critically,) distances between graphs are measured with respect to the vertex distribution: suppose $E_{1},E_{2}\subseteq \binom{[n]}{2}$ are two edge sets, then the distance between $E_{1}$ and $E_{2}$ with respect to a distribution $\mu$ over $[n]$ is
\[
\sum_{\{u,v\}\in E_{1}\triangle E_{2}}\mu(\{u\})\cdot\mu(\{v\}),\quad\text{where }E_{1}\triangle E_{2}\text{ is the symmetric difference between }E_{1}\text{ and }E_{2}.
\]
This roughly corresponds to taking $\Lambda=\binom{[n]}{2}$ in Definition~\ref{def:distribution_free_testing} and restricting the unknown distribution $\mu$ over $\binom{[n]}{2}$ to be a ``product distribution'' (see e.g. \cite[Section 10.1.3]{goldreich1998property}). It turns out that such product distributions behave not too differently from the uniform distribution in many important aspects. In particular, subgraph removal lemmas (that are well-known in the uniform distribution setting; see e.g. \cite[Section 4.4]{shapira2022local}) still hold~\cite{goldreich2019testing}, implying that any subgraph-freeness property can be tested in constant queries\footnote{By ``constant queries,'' we mean the query complexity depends only on the proximity parameter $\varepsilon$ but not on the number of vertices of the graph.} in the vertex-distribution-free model.

In contrast, the questions considered in the present work correspond to taking $\Lambda=\binom{[n]}{2}$ in Definition~\ref{def:distribution_free_testing} but not imposing any restriction on the unknown distribution $\mu$. This setting perhaps should be called the ``edge-distribution-free'' model. In this model, (especially since distances between graphs are no longer measured with respect to a product distribution) subgraph removal lemmas no longer make sense, and many basic properties such as triangle-freeness cannot be tested in constant queries, as already observed in \cite[Section 10.1.4]{goldreich1998property}. Retreating from the unattainable constant-query regime, it is still natural to ask whether some properties have query/sample complexities that grow with the size parameter $n$ more slowly than other properties --- which is exactly what Theorems~\ref{thm:main} to~\ref{thm:tree-freeness} attempt to answer.\footnote{Indeed, it was asked in \cite{sublinear_open_99} whether one can define, motivate, and prove non-trivial results in ``an edge-distribution-free model'' for graph property testing.} 

Another type of distribution-free models that has been considered for graphs features unknown distributions over $[n]\times [d]$ (see e.g. \cite{halevy2008distribution}), where $[n]$ is the vertex set and $d$ is an upper bound on the vertex degrees. This is arguably closer in spirit to the setting of unknown vertex distributions than to the one of unknown edge distributions.

\subsection{Further Motivation for Edge-Distribution-Free Testing}\label{subsec:further_motivation}

As mentioned in Section~\ref{subsec:related_work}, it has been more popular to study the case $\Lambda=\{0,1\}^{n}$ in Definition~\ref{def:distribution_free_testing} than the case $\Lambda=\binom{[n]}{2}$. However, sometimes questions about the latter domain naturally arise when studying questions about the former. In particular, in the paper \cite{chen2024distribution} on distribution-free testing of decision lists (a class of Boolean functions on the hypercube), it turns out that a hard case for a decision list tester is (roughly speaking) when the unknown distribution over $\{0,1\}^{n}$ is actually supported on the ``weight-2 slice''
\[
\big\{x\in \{0,1\}^{n}\,\big|\,\text{the Hamming weight of }x\text{ is }2\big\},
\]
which is clearly equivalent to the domain $\binom{[n]}{2}$. The class of functions $\{0,1\}^{n}\rightarrow\{0,1\}$ that are decision lists, when restricted to the weight-2 slice, becomes a class of functions $\binom{[n]}{2}\rightarrow\{0,1\}$ or equivalently a class of graphs known as \emph{threshold graphs}. In order to show that the property of being a decision list on $\{0,1\}^{n}$ can be tested in $O(n^{11/12})$ queries, the authors of \cite{chen2024distribution} had to (roughly speaking) first show the following:

\begin{theorem}[Implicit in \cite{chen2024distribution}] \label{thm:threshold_graph}
The property of being a threshold graph on $n$ vertices can be tested in the edge-distribution-free model (with queries; see Remark~\ref{rem:tester_with_query}) using at most $O(n^{2/3})$ queries (samples and queries combined; see Remark~\ref{rem:tester_with_query}).
\end{theorem}

It seems likely that in order to obtain a query-optimal distribution-free decision list tester, one has to first optimize the query complexity in Theorem~\ref{thm:threshold_graph}. We note that a query lower bound of $\Omega(\sqrt{n})$ for testing decision lists and (implicitly) for testing threshold graphs was shown in \cite{chen2024distribution}.

\section{Technical Overview}

In this section, we provide an overview of our proof techniques. 

\subsection{General Framework}\label{subsec:general_framework}

We start by describing a general framework for distribution-free property testing.

\begin{definition}\label{def:violation_hypergraph}
Suppose $\cH$ is a nonempty family of Boolean-valued functions on a finite domain $\Lambda$. Fix a function $f:\Lambda\rightarrow \{0,1\}$. A subset $S\subseteq \Lambda$ is said to be an \emph{$f$-violation} of the property $\cH$ if there does not exist $h\in \cH$ such that $f$ agrees with $h$ on $S$ (i.e. $f(x)=h(x)$ for all $x\in S$). A subset $S\subseteq \Lambda$ is said to be a \emph{minimal $f$-violation} of $\cH$ if $S$ is an $f$-violation of $\cH$ but no proper subset of $S$ is an $f$-violation of $\cH$. The collection of minimal $f$-violations of $\cH$ is the edge set of a hypergraph on the vertex set $\Lambda$ that we call the \emph{violation hypergraph} of $f$ against $\cH$.
\end{definition}

The notion of violation hypergraph was formally introduced in \cite{dolev2011distribution},\footnote{In \cite{dolev2011distribution} the edge set of the violation hypergraph is the collection of $f$-violations, instead of minimal $f$-violations.} and is inherently important for property testing because of the following observation:

\begin{proposition}
For any distribution $\mu$ over $\Lambda$, if $S$ is a vertex cover of the violation hypergraph of $f$ against $\cH$ with minimum possible measure under $\mu$, then
\begin{equation}\label{eq:vertex_cover_equivalence}
\mu(S)=\min_{h\in \cH}\Pru{x\sim \mu}{f(x)\neq h(x)}.
\end{equation}
\end{proposition}
\begin{proof}
For any $h\in \cH$, the set $\{x\in \Lambda\mid f(x)\neq h(x)\}$ is a vertex cover of the violation hypergraph of $f$ against $\cH$, so its measure under $\mu$ is at least $\mu(S)$. Conversely, we claim that there exists some $h\in \cH$ such that
\[
\{x\in \Lambda\mid f(x)\neq h(x)\}\subseteq S,
\]
which implies that the right-hand side of \eqref{eq:vertex_cover_equivalence} is at most $\mu(S)$. Assume on the contrary that for every $h\in \cH$ there exists some $x\in \Lambda\setminus S$ such that $f(x)\neq h(x)$. Then $\Lambda\setminus S$ is an $f$-violation of $\cH$, and hence there exists a minimal $f$-violation of $\cH$ that is contained in $\Lambda\setminus S$. This contradicts the assumption that $S$ is a vertex cover of the violation hypergraph.  
\end{proof}

Note that a one-sided-error tester for the property $\cH$ can reject a function $f$ if and only if it has sampled (or queried) all elements of an $f$-violation of $\cH$. Therefore, the question of analyzing the one-sided-error sample complexity of testing $\cH$ is equivalent to: given that the minimum-weight vertex cover of the violation hypergraph of $f$ has weight at least $\varepsilon$ under $\mu$, how many samples from $\mu$ does one need to get a full edge of the $f$-violation hypergraph with high probability?

It turns out that one can prove a fairly general birthday-paradox-type lemma in response to this question (see Lemma~\ref{lem:hypergraph_birthday} for a formal version of the following lemma):

\begin{lemma}[{\cite[Lemma 2.2]{chen2024distribution}}, informal]\label{lem:hypergraph_birthday_informal}
For any $k$-uniform hypergraph on $n$ vertices with vertices weighted by $\mu$, if the minimum-weight vertex cover has weight at least $\varepsilon$, then $O(n^{(k-1)/k}/\varepsilon)$ samples from $\mu$ are sufficient to find a full edge with high probability. 
\end{lemma}

\begin{remark}
A proof of the $k=2$ case of Lemma~\ref{lem:hypergraph_birthday_informal} was implicit already in \cite{dolev2011distribution}; it was abstracted into the current form and generalized to $k\geq 3$ by \cite{chen2024distribution}. The proof of Lemma~\ref{lem:hypergraph_birthday_informal} in \cite{chen2024distribution} is different from the proof in \cite{dolev2011distribution}; see Section~\ref{subsec:overview_graph} for further discussion.
\end{remark}

As a simple application of Lemma~\ref{lem:hypergraph_birthday_informal}, consider the question of monotonicity testing over general posets. Suppose $\Lambda$ is a partially ordered set and $f:\Lambda\rightarrow\{0,1\}$ is a function, and we want to test the property that $f$ is monotone, i.e. $f(x)\leq f(y)$ for all $x\leq y$. It is easy to see that any minimal $f$-violation of monotonicity must be a pair $\{x,y\}\subseteq \Lambda$ such that $x< y$. Therefore, the violation hypergraphs are 2-uniform, and applying Lemma~\ref{lem:hypergraph_birthday_informal} immediately yields the following result of \cite{blais2021vc}:

\begin{theorem}[{\cite[Theorem 7.9]{blais2021vc}}]
For any partially ordered set $\Lambda$ with $n$ elements, if $\cH$ is the collection of monotone Boolean-valued functions on $\Lambda$, then $\sam(\cH,\varepsilon)\leq O(\sqrt{n}/\varepsilon)$.\footnote{In the uniform distribution case, monotonicity can be tested in $O\big(\sqrt{n/\varepsilon}\big)$ samples \cite{fischer2002monotonicity}.}
\end{theorem}

\begin{remark}\label{rem:canonical_tester_monotone}
As discussed earlier, every sample-based property testing problem admits a one-sided-error \emph{canonical tester}: the tester simply rejects if there is a violation of the property within the observed samples. When the property is \emph{downward-closed} (more commonly referred to as \emph{monotone} in the property testing literature), the canonical tester has an even simpler description.

Let $\cH$ be a downward-closed family of subsets of a finite domain $\Lambda$, and let $\mu$ be a probability distribution over $\Lambda$. Recall from Definition~\ref{def:distribution_free_testing} that in order to test whether an unknown set $E \subseteq \Lambda$ belongs to $\cH$, the algorithm receives samples $e_1,\dots,e_m$ drawn from $\mu$, together with the information of whether $e_i \in E$ for each $i\in[m]$. We call $e_i$ a \emph{positive sample} if $e_i \in E$. Since $\cH$ is downward-closed, it follows that the canonical tester for $\cH$ rejects if and only if the set formed by the positive samples does not belong to $\cH$.
\end{remark}

\subsection{Applying the Framework to Graph Problems}\label{subsec:overview_graph}

The framework in Section~\ref{subsec:general_framework} is especially suitable for analyzing subgraph-freeness properties. It is easy to see that for any graph $H$ and any $f:\binom{[n]}{2}\rightarrow \{0,1\}$, the minimal $f$-violations of $\cG^{H\textup{-free}}_{n}$ (defined in Theorem~\ref{thm:tree-freeness}) are the copies of $H$ in the edge set $f^{-1}(1)$. Therefore, for any graph $H$ with $t$ edges, any violation hypergraph against $H$-freeness is $t$-uniform. We may thus apply Lemma~\ref{lem:hypergraph_birthday_informal} and immediately get:
\begin{theorem}\label{thm:subgraph_freeness_upper_bound}
For any simple graph $H$ with $t\geq 1$ edges, we have
\[
\sam\big(\cG^{H\textup{-free}}_{n},\varepsilon\big)\leq O(1/\varepsilon)\cdot \binom{n}{2}^{(t-1)/t}=O(n^{2(t-1)/t}/\varepsilon).
\]
\end{theorem}

The upper bound part of \eqref{eq:triangle_result} follows as a special case of Theorem~\ref{thm:subgraph_freeness_upper_bound}, by taking $H$ to be a triangle. 

However, if we apply Theorem~\ref{thm:subgraph_freeness_upper_bound} to the square-freeness property, we can only get the upper bound $\sam\big(\cG^{\textup{squ}}_{n},\varepsilon\big)\leq O(n^{3/2}/\varepsilon)$, failing to reach the optimal bound $n^{9/8\pm o(1)}$ as stated in \eqref{eq:square_result}. In order to obtain the desired upper bound $O(n^{9/8}/\varepsilon)$, we have to develop new techniques based on Lemma~\ref{lem:hypergraph_birthday_informal}:

\begin{enumerate}
\item We first open up the proof of Lemma~\ref{lem:hypergraph_birthday_informal} (due to \cite{chen2024distribution}) as a white box. The main idea of the proof is to use linear programming duality to turn the \emph{universally-quantified} condition about vertex cover into an \emph{existence} of ``fractional matching'' in the violation hypergraph (Lemma~\ref{lem:LP_weak_duality}). In the context of testing square-freeness, the edges of the violation hypergraph are copies of squares in the input graph, so the ``fractional matching'' would translate to a family of ``weighted squares'' whose convex combination is dominated by the edge distribution (Definition~\ref{def:witness}).

\item The fractional matching effectively allows us to ``embed'' a classical-birthday-paradox structure into the edge distribution. One can use Carath\'eodory's theorem to limit the number of squares in the fractional matching, i.e. the number of ``birthday slots,'' to at most $O(n^{2})$. In $O(n^{3/2})$ samples, with high probability there are four people sharing a birthday, i.e. four edges forming a square. This is how Lemma~\ref{lem:hypergraph_birthday_informal} is proved in \cite{chen2024distribution} (and works perfectly well for testing triangle-freeness), but falls short of the optimal bound for testing square-freeness by a polynomial factor.

\item The main new idea is to \emph{delay} the application of Carath\'eodory's theorem, and to try to milk the ``fractional matching'' for more. It turns out that there are \emph{two} different sources of squares that we could hope to reveal by samples. On the one hand, there is the family of weighted squares ``planted'' into the edge distribution by the fractional matching, which has been our only source of squares. On the other hand, if these squares are planted in a sufficiently ``dilute'' manner, we argue that a huge number of \emph{unintended} squares will inevitably be created during the process, and these unintended ones will be our second source of squares.

\item We have to divide into two cases based on the ``fractional matching'' we obtained from linear programming duality. If the fractional matching is ``dilute,'' we argue that an unintended square will likely show up in as few as $\widetilde{O}(n)$ samples (Lemma~\ref{lem:dilute_case}). In the ``concentrated'' case, we argue that the number of ``birthday slots'' are effectively reduced to $O(n^{3/2})$, and hence there will likely be ``four people sharing a birthday'' within $O(n^{9/8})$ samples (Lemma~\ref{lem:concentrated_case}). We remark that a nontrivial amount of effort is required for finding a formalization of the ``diluteness'' notion that works smoothly in the proof (see Section~\ref{subsec:case_analysis}).

\end{enumerate}

The other sample complexity upper bounds proved in this paper (Theorems~\ref{thm:graph_homomorphism} to~\ref{thm:tree-freeness}) require different techniques, for which we choose not to provide overviews here.

\subsubsection{Subgraph-Removal for Sparse Graphs}\label{subsubsec:sparse_removal}

The discussion above is reminiscent of the celebrated \emph{removal lemmas} in graph theory (see e.g. the survey \cite{conlon2013graph}). Suppose $H$ is a fixed connected simple graph, and consider a graph on $n$ vertices whose edge set consists of $m$ edge-disjoint copies of $H$. How large can $m$ be if no ``unintended'' copy of $H$ is allowed, i.e. every edge is in exactly one copy of $H$? The subgraph removal lemma implies that for any $H$ with at least $t\geq 3$ vertices, we must have $m\leq o(n^{2})$ if there is no unintended copy of $H$. Furthermore, if $m=\Omega(n^{2})$ edge-disjoint copies of $H$ are planted, then there must be as many as $\Omega(n^{t})$ unintended copies of $H$. 

For convenience of further discussion, we use the following non-standard notation.

\begin{definition}
Given a simple graph $H$, let $\ex^{=1}(n,H)$ be the maximum number of edges in an $n$-vertex graph where every edge is contained in exactly one copy of $H$.
\end{definition}

For certain graphs $H$, one can pack into an $n$-vertex graph as many as $n^{2-o(1)}$ copies of $H$ without creating unintended copies. In the case where $H=C_{3}$ is a triangle, the celebrated Ruzsa-Szemer\'edi construction \cite{ruzsa1978triple} (based on Behrend's construction \cite{behrend1946sets} of integer sets without 3-term arithmetic progressions) shows that:

\begin{proposition}[\cite{behrend1946sets,ruzsa1978triple}]\label{prop:ruzsa_semeredi}
We have $\ex^{=1}(n,C_{3})\geq n^{2}\exp\bigl(-O\bigl(\sqrt{\log n}\bigr)\bigr)$.
\end{proposition}

We will use Proposition~\ref{prop:ruzsa_semeredi} to prove the sample complexity lower bound for testing triangle-freeness stated in \eqref{eq:triangle_result}. Indeed, Proposition~\ref{prop:ruzsa_semeredi} almost immediately implies an $n^{4/3-o(1)}$ sample lower bound for \emph{one-sided-error} triangle-freeness testers. To extend the lower bound to hold against two-sided-error testers, we use a standard constructional technique (see Section~\ref{subsec:triangle_lower_bound}) that has appeared in, for example, lower bounds for triangle-freeness testers in the ``general graph model'' \cite{alon2008testing}.

There are also some graphs $H$ for which the bound $\ex^{=1}(n,H)\leq  o(n^{2})$ provided by the removal lemma can be improved by a polynomial factor in $n$. Recall that the Tur\'an number of $H$, denoted by $\ex(n,H)$, is the maximum number of edges in an $n$-vertex graph with no subgraphs isomorphic to $H$. For any graph in which every edge is contained in exactly one copy of $H$, deleting one edge from every copy of $H$ yields an $H$-free graph, so we have:
\begin{proposition}
For a fixed simple graph $H$ with at least two edges, we have $\ex^{=1}(n,H)\leq 2\cdot \ex(n,H)$. 
\end{proposition}

The K\H{o}v\'ari-S\'os-Tur\'an theorem \cite{kHovari1954problem} shows for any fixed bipartite graph $H$ with $t$ vertices that $\ex(n,H)\leq O(n^{2-1/t})$, so we also have $\ex^{=1}(n,H)\leq O(n^{2-1/t})$. When $H=C_{4}$ is a square (i.e. 4-cycle), the resulting upper bound $\ex^{=1}(n,C_{4})\leq n^{3/2}$ is the key reason we are able to improve the upper bound on $\sam\big(\cG^{\textup{squ}}_{n},\varepsilon\big)$ from  $O(n^{3/2}/\varepsilon)$ to $O(n^{9/8}/\varepsilon)$, as discussed earlier. That being said, we are not able to use the bound $\ex^{=1}(n,C_{4})\leq O(n^{3/2})$ or $\ex(n,C_{4})\leq O(n^{3/2})$ as a black box to prove the $O(n^{9/8})$ sample complexity upper bound, and it seems that some careful case analysis (as described earlier) is necessary for proving the latter.

In terms of lower bounds, it was shown by \cite{brown1966graphs,erdHos1966problem} that $\ex(n,C_{4})=\Theta(n^{3/2})$, and the following lower bound on $\ex^{=1}(n,C_{4})$ is (implicitly) shown by Timmons and Verstra\"ete \cite{timmons2015counterexample}:

\begin{proposition}[\cite{timmons2015counterexample}]\label{prop:ruzas_square}
We have $\ex^{=1}(n,C_{4})\geq n^{3/2}\exp\bigl(-O\bigl(\sqrt{\log n}\bigr)\bigr)$.
\end{proposition}

As in the case of triangle-freeness, we will use Proposition~\ref{prop:ruzas_square} to prove the sample complexity lower bound for testing square-freeness stated in \eqref{eq:square_result}. For the sake of completeness, we will sketch the proof of Proposition~\ref{prop:ruzas_square} in Section~\ref{subsec:square_lower_bound}. 

\begin{remark} To the best of the author's knowledge, it is unknown whether $\ex^{=1}(n,C_{4})=o(n^{3/2})$,\footnote{Indeed, Solymosi \cite{solymosi2011c4} conjectured that $\ex^{=1}(n,C_{4})=o(n^{3/2})$, while Verstra\"ete \cite{verstraete2016extremal} conjectured that $\ex^{=1}(n,C_{4})=\Theta(n^{3/2})$.} and determining the asymptotics of $\ex^{=1}(n,C_{3})$ is a major open problem (see e.g. \cite{shapira2022local}).
\end{remark}

\section{Preliminaries}

\subsection{General Notations}\label{subsec:general_notations}

In this subsection we summarize general notational conventions used throughout this paper.

\paragraph{Sets.} For two subsets $E_{1},E_{2}$ of a domain $\Lambda$, we use $E_{1}\triangle E_{2}:=(E_{1}\setminus E_{2})\cup (E_{2}\setminus E_{1})$ to denote the symmetric difference between $E_{1}$ and $E_{2}$.

\paragraph{Probability.} For a finite domain $\Lambda$ and a probability distribution $\mu$ over $\Lambda$, we write $\Exs{x\sim \mu}{\cdot}$ and $\Prs{x\sim \mu}{\cdot}$ to denote expectation and probability, respectively, when $x\in \Lambda$ is a random element following the distribution $\mu$. A \emph{probability mass function} on $\Lambda$ is a function $f:\Lambda\rightarrow[0,+\infty)$ such that $\sum_{x\in \Lambda}f(x)=1$. A \emph{sub-probability mass function} on $\Lambda$ is a function $f:\Lambda\rightarrow [0,+\infty)$ such that $\sum_{x\in \Lambda}f(x)\leq 1$. Similarly, a \emph{probability vector} indexed by $\Lambda$ is a vector $p\in [0,1]^{\Lambda}$ such that $\sum_{x\in \Lambda}p_{x}=1$, while a vector $p\in [0,1]^{\Lambda}$ is called a \emph{sub-probability vector} if $\sum_{x\in \Lambda}p_{x}\leq 1$.

\paragraph{Sampling.} Given a finite domain $\Lambda$, a \emph{sample} from a sub-probability vector $p\in [0,1]^{\Lambda}$ is a random element $y$ of an extended domain $\Lambda\cup\{\nil\}$ such that 
\[
\Prs{y}{y=x}=p_{x}\text{ for any }x\in \Lambda\quad\text{and}\quad \Prs{y}{y=\nil}=1-\sum_{x\in \Lambda}p_{x}.
\]
The special symbol $\nil$ will always be used as an ``outside'' placeholder element in such contexts. Samples from sub-probability mass functions are similarly defined. 

\paragraph{Empirical vectors.} Given a sequence of elements $y_{1},\dots,y_{m}\in \Lambda$, we define the \emph{empirical count vector} of this sequence to be the vector $w\in \bN^{\Lambda}=\{0,1,2,\dots\}^{\Lambda}$ where the coordinate $w_{x}$ equals the number of indices $i\in [m]$ such that $y_{i}=x$, for each element $x\in \Lambda$. The \emph{empirical indicator vector} of this sequence is the vector $w'\in \{0,1\}^{\Lambda}$ defined by $w'_{x}=\ind{w_{x}\geq 1}$ for all $x\in \Lambda$.

\paragraph{Sampling processes.} Suppose $p\in [0,1]^{\Lambda}$ is a sub-probability vector, and $f:\Lambda\rightarrow [0,1]$ is the sub-probability mass function associated with $p$ (i.e. $f(x)=p_{x}$ for all $x\in \Lambda$). Consider the following canonical sampling process:
\begin{enumerate}
\item Take a batch of $m$ independent samples $y_{1},\dots,y_{m}$ from $p$.
\item Let $w\in \bN^{\Lambda}$ be the empirical count vector of the sequence $y_{1},\dots,y_{m}$, and output $w$. 
\end{enumerate}
We use $\bcS(p,m)$ or $\bcS(f,m)$ to denote the distribution of the output vector $w$ in the above process.\footnote{Note that if $p\in [0,1]^{\Lambda}$ is a probability vector, then $\bcS(p,m)$ is a multinomial distribution. However, if $p$ is only a \emph{sub}-probability vector, then $\bcS(p,m)$ may not be supported on the layer $\left\{w\in \bN^{\Lambda}\,\middle|\,\sum_{x\in \Lambda}w_{x}=m \right\}$.} If the empirical count vector in step 2 of the process is replaced with the empirical indicator function, the resulting output distribution over $\{0,1\}^{\Lambda}$ is denoted by $\bcS'(p,m)$ or $\bcS'(f,m)$.

\paragraph{Combinatorial structures.} For a fixed positive integer $n$, we define various combinatorial structures associated with the edge set $\binom{[n]}{2}$. We define
\[
\Square(n):=\left\{\big\{\{a,b\},\{b,c\},\{c,d\},\{d,a\}\big\}\subseteq \binom{[n]}{2}\,\middle|\,a,b,c,d\text{ are distinct elements of }[n]\right\}
\]
to be the collection of all four-edge sets that correspond to squares. Two edges in $\binom{[n]}{2}$ are said to form a \emph{wedge} if they have exactly one common vertex, and we correspondingly define
\[
\Wedge(n):=\left\{\big\{\{a,b\},\{b,c\}\big\}\subseteq \binom{[n]}{2}\,\middle|\,a,b,c\text{ are distinct elements of }[n]\right\}
\]
to be the collection of wedges on the vertex set $[n]$. A wedge $\big\{\{a,b\},\{b,c\}\big\}$ can also be viewed as an ordered pair $\big(\{a,c\},b\big)\in \binom{[n]}{2}\times [n]$. By an abuse of notation, we identify the collection $\Wedge(n)$ with the subset
\begin{equation}\label{eq:wedge_in_grid}
\Wedge(n):=\left\{\big(\{a,c\},b\big)\in \binom{[n]}{2}\times [n]\,\middle|\,b\not\in \{a,c\}\right\}\subseteq \binom{[n]}{2}\times [n].
\end{equation}

\paragraph{Subgraph-freeness.} Given any constant $\varepsilon\in (0,1)$ and a fixed simple graph $H$, a sub-probability vector $p\in [0,1]^{\binom{[n]}{2}}$ is said to be $\varepsilon$-far from $H$-free if for any edge set $E\in \cG^{H\textup{-free}}_{n}$ (see the statement of Theorem~\ref{thm:tree-freeness}), we have
\[
\sum_{e\in\binom{[n]}{2}\setminus E}p_{e}\geq \varepsilon.
\]

\subsection{Stochastic Domination}

Sampling processes (as formally introduced in Section~\ref{subsec:general_notations}) are of central importance in this paper. In order to meaningfully compare different sampling processes, we make the following definition of stochastic domination. Recall from Section~\ref{subsec:general_notations} that the output of a sampling process is a random element of the space $\bN^{\Lambda}$ for some index set $\Lambda$, so it suffices to ``compare'' distributions over $\bN^{\Lambda}$.

\begin{definition}\label{def:stochastic_domination}
Let $\Lambda$ be a finite set and let $\mu,\nu$ be probability distributions over $\bN^{\Lambda}$. For parameters $\lambda_{1},\lambda_{2}\in(0,1]$, we say $\nu$ is \emph{$(\lambda_{1},\lambda_{2})$-dominated} by $\mu$, written as \[\nu\leq_{(\lambda_{1},\lambda_{2})}\mu,\]
if there exists a coupling distribution $\rho$ over $\bN^{\Lambda}\times \bN^{\Lambda}$ such that the following conditions hold:
\begin{enumerate}[label=(\arabic*)]
\item $\Prs{(w,z)\sim \rho}{w\succeq z}\geq \lambda_{1}$.\footnote{For vectors $w,z\in \bN^{\Lambda}$, we write $w\succeq z$ if $w_{x}\geq z_{x}$ for all $x\in \Lambda$.}
\item 
For any subset $S\subseteq \bN^{\Lambda}$, we have $\Prs{(w,z)\sim \rho}{w\in S}=\mu(S)$.
\item For any subset $S\subseteq \bN^{\Lambda}$, we have $\lambda_{2}\cdot \Prs{(w,z)\sim \rho}{z\in S}\leq  \nu(S)$.
\end{enumerate}
When $\lambda_{1}=\lambda_{2}=1$, we simply say that $\nu$ is dominated by $\mu$, omitting the $(\lambda_{1},\lambda_{2})$.
\end{definition}
Our definition of stochastic domination has the following basic property.

\begin{proposition}\label{prop:stochastic_domination}
Let $\Lambda$ be a finite set, and let $\lambda_{1},\lambda_{2}\in(0,1]$ be constants. Suppose $\mu$ and $\nu$ are probability distributions over $\bN^{\Lambda}$ such that $\nu$ is $(\lambda_{1},\lambda_{2})$-dominated by $\mu$. Then for any downward-closed subset $S\subseteq \bN^{\Lambda}$, we have
\[
\mu(S)\leq \lambda_{2}^{-1}\cdot \nu(S)+(1-\lambda_{1}).
\]
\end{proposition}
\begin{proof}
It suffices to notice the union bound inequality
\[
\Pru{(x,y)\sim\rho}{x\in S}\leq \Pru{(x,y)\sim\rho}{y\in S}+\Pru{(x,y)\sim\rho}{x\not\succeq y},
\]
and then replace the three terms in the inequality by the desired quantities, using the three conditions in Definition~\ref{def:stochastic_domination}.
\end{proof}

\section{Testing Bipartiteness}

The goal of this section is to prove Theorem~\ref{thm:graph_homomorphism}. Note that when the graph $H$ is a single edge (on two vertices), the collection $\cG^{H\textup{-hom}}_{n}$ is identical to $\cG^{\textup{bip}}_{n}$; so the bipartiteness testing result stated in \eqref{eq:bipartiteness_result} is a special case of Theorem~\ref{thm:graph_homomorphism}. The proofs of both the upper bound and the lower bound are similar to \cite[Theorem 4.6]{goldreich2016sample}.

\subsection{Upper Bound}

Suppose $H$ is a fixed simple graph on the vertex set $[k]$, and for any indices $i,j\in [k]$ we denote 
\[
H_{ij}=\begin{cases}
1, &\text{if }i\neq j\text{ and }H\text{ contains the edge }\{i,j\},\\
0, &\text{otherwise}.
\end{cases}
\]
Given an \emph{assignment map} $\tau:[n]\rightarrow [k]$, let $F_{\tau,H}$ be the collection of edges $\{a,b\}\in \binom{[n]}{2}$ such that $H_{\tau(a),\tau(b)}=0$. By definition, an edge set $E\subseteq \binom{[n]}{2}$ belongs to the collection $\cG^{H\textup{-hom}}_{n}$ if and only if $E\cap F_{\tau,H}=\emptyset$ for some $\tau:[n]\rightarrow [k]$.

Let $\mu$ be a probability distribution over $\binom{[n]}{2}$. For any edge set $E\subseteq \binom{[n]}{2}$, it is easy to see that 
\begin{equation}\label{eq:homomorphism_distance}
\min_{E'\in \cG^{H\textup{-hom}}_{n}}\mu(E\triangle E')=\min_{\tau:[n]\rightarrow [k]}\mu(E\cap F_{\tau,H})
\end{equation}
We claim that for any edge set $E\subseteq \binom{[n]}{2}$ that is $\varepsilon$-far from $\cG^{H\textup{-hom}}_{n}$ with respect to $\mu$, the canonical tester for $\cG^{H\textup{-hom}}_{n}$ (described in Remark~\ref{rem:canonical_tester_monotone}) rejects $E$ with probability at least $2/3$ after receiving $O(n/\varepsilon)$ labeled samples from $\mu$. By \eqref{eq:homomorphism_distance} and Remark~\ref{rem:canonical_tester_monotone}, it suffices to prove the following lemma.

\begin{lemma}
Let $\varepsilon\in (0,1)$ be a constant. Suppose $E\subseteq \binom{[n]}{2}$ is an edge set and $\mu$ is a distribution over $\binom{[n]}{2}$ such that $\mu(E\cap F_{\tau,H})\geq \varepsilon$ for any map $\tau:[n]\rightarrow [k]$. For any integer $m\geq \varepsilon^{-1}(2+n\ln k)$, in $m$ independent samples from $\mu$, the probability is at least $2/3$ that for any $\tau:[n]\rightarrow [k]$, there is a sampled edge that belongs to $E\cap F_{\tau,H}$.
\end{lemma}

\begin{proof}
For any fixed map $\tau:[n]\rightarrow [k]$, the probability that no sample falls in $E\cap F_{\tau,H}$ is at most $(1-\varepsilon)^{m}\leq \exp(-\varepsilon m)\leq \frac{1}{3}\exp(-n\ln k)=\frac{1}{3}k^{-n}$. By union bound over all maps $\tau:[n]\rightarrow [k]$, it follows that with probability at least $2/3$, for any $\tau$ there is a sampled edge falling in $E\cap F_{\tau,H}$.
\end{proof}

\begin{corollary}\label{cor:bipartiteness_upper_bound}
For any fixed simple graph $H$ with at least one edge, we have $\sam\big(\cG^{H\textup{-hom}}_{n},\varepsilon\big)\leq O(n/\varepsilon)$.
\end{corollary}

\subsection{Lower Bound}\label{subsec:bipartiteness_lower_bound}

In this subsection, we prove the lower bound part of Theorem~\ref{thm:graph_homomorphism}. Throughout this subsection, we let $k\geq 3$ be a fixed integer. A basic tool in the proof is the fact that a complete regular $k$-partite graph is far from $(k-1)$-colorable.

\begin{lemma}\label{lem:complete_k_partite_to_k_minus_1_colorable}
Let $V_1,\dots,V_k$ be pairwise disjoint sets, each of size $n$, and let
\[
\Gamma:=\bigl\{\{u,v\}: u\in V_i,\ v\in V_j,\ 1\le i<j\le k\bigr\}.
\]
Thus $\Gamma$ is the edge set of the complete $k$-partite graph with parts
$V_1,\dots,V_k$. If $E\subseteq \Gamma$ is such that the graph $(V_1\cup\cdots\cup V_k,E)$ is
$(k-1)$-colorable, then
\(|\Gamma\setminus E|\ge n^2\).
\end{lemma}

\begin{proof}
Fix a proper $(k-1)$-coloring of the graph $(V_1\cup\cdots\cup V_k,E)$, and let
$C_1,\dots,C_{k-1}$ denote its color classes. For each $i\in [k]$ and
$c\in [k-1]$, set \(w_{i,c}:=|V_i\cap C_c|\). Then
\[
w_{i,1}+\cdots+w_{i,k-1}=n
\qquad\text{for every }i\in [k].
\]

Now fix a color $c\in [k-1]$ and two distinct indices $i,j\in [k]$. Every pair of vertices \(u\in V_i\cap C_c\) and \(v\in V_j\cap C_c\) forms an edge of $\Gamma$, but cannot belong to $E$, since $u$ and $v$ have the same color. Hence all \(
w_{i,c}w_{j,c}\) such edges lie in $\Gamma\setminus E$. Summing over all colors and all pairs
$i<j$, we obtain
\begin{equation}\label{eq:remove_edges_in_w}
|\Gamma\setminus E|
\ge
\sum_{c=1}^{k-1}\sum_{1\le i<j\le k} w_{i,c}\,w_{j,c}.
\end{equation}

Define
\(p_{i,c}:=w_{i,c}/n\),
so that each vector
\(
p_i:=(p_{i,1},\dots,p_{i,k-1})
\)
lies in the simplex
\[
\Delta_{k-2}
:=
\Bigl\{
(t_1,\dots,t_{k-1})\in \mathbb{R}_{\ge 0}^{k-1}:
t_1+\cdots+t_{k-1}=1
\Bigr\}.
\]
By \eqref{eq:remove_edges_in_w}, it suffices to prove \(
\sum_{1\le i<j\le k}\sum_{c=1}^{k-1} p_{i,c}\,p_{j,c}\ge 1\) for all \(p_1,\dots,p_k\in \Delta_{k-2}\).

Since $\sum_{1\le i<j\le k}\sum_{c=1}^{k-1} p_{i,c}\,p_{j,c}$ depends linearly on each vector $p_{i}$, by minimizing successively in each
variable, we may assume that each $p_i$ is an extreme point of $\Delta_{k-2}$, that is, one of the standard basis vectors. Since there are $k$
vectors but only $k-1$ possible basis vectors, the pigeonhole principle implies that $p_{i^*}=p_{j^*}$ for some $\{i^*,j^*\}\in\binom{[k]}{2}$, and hence $\sum_{1\le i<j\le k}\sum_{c=1}^{k-1} p_{i,c}\,p_{j,c}\geq 1$, as desired.
\end{proof}

We then proceed to define the hardness distributions used in the lower bound proof.

\begin{definition}
Consider the vertex set $[n]\times [k]\times \bF_{2}$ of size $2kn$. For any vector $x\in\bF_2^{n\times k}$, we define two edge sets $E^{\yes}(x),E^{\no}(x)\subseteq \binom{[n]\times [k]\times \bF_{2}}{2}$ by
\begin{align*}
E^{\yes}(x)&=\left\{
\bigl\{(a,i,t+x_{a,i}),(b,j,t+x_{b,j}+1)\bigr\}\,\middle|\,a,b\in [n],\;\{i,j\}\in\binom{[k]}{2},\;t\in \bF_{2}\right\},\text{ and}\\
E^{\no}(x)&=\left\{
\bigl\{(a,i,t+x_{a,i}),(b,j,t+x_{b,j})\bigr\}\,\middle|\,a,b\in [n],\;\{i,j\}\in\binom{[k]}{2},\;t\in \bF_{2}\right\}.
\end{align*}
\end{definition}

Note that both $E^{\yes}(x)$ and $E^{\no}(x)$ have cardinality $(k-1)kn^{2}$. Furthermore, they have the following nice properties.
\begin{proposition}\label{prop:property_of_Eyes_Eno}
For any vector $x\in \bF_{2}^{n\times k}$, we have:
\begin{enumerate}[label=(\arabic*)]
\item The graph $\big([n]\times [k]\times\bF_{2},E^{\yes}(x)\big)$ is bipartite.
\item Any subset $E'\subseteq E^{\no}(x)$ such that the graph $\big([n]\times [k]\times \bF_{2},\,E'\big)$ is $(k-1)$-colorable (equivalently, admits a homomorphism to the $(k-1)$-vertex complete graph) must satisfy \[\bigl|E^{\no}(x)\setminus E'\bigr|\geq k^{-2}\bigl|E^{\no}(x)\bigr|.\]
\end{enumerate}
\end{proposition}
\begin{proof}
The graph $\big([n]\times [k]\times\bF_{2},E^{\yes}(x)\big)$ is clearly bipartite because of the partitioning map $\tau_{x}:[n]\times [k]\times \bF_{2}\rightarrow \bF_{2}$ given by $\tau(a,i,t)=x_{a,i}+t$. On the other hand, the graph $\big([n]\times [k]\times\bF_{2},E^{\no}(x)\big)$ is the vertex-disjoint union of two complete regular $k$-partite graphs (the vertex sets of the two connected components are $\tau^{-1}_{x}(0)$ and $\tau_{x}^{-1}(1)$, respectively). Therefore, it follows from Lemma~\ref{lem:complete_k_partite_to_k_minus_1_colorable} that
\[
\bigl|E^{\no}(x)\setminus E'\bigr|\geq 2n^{2}\geq k^{-2}\bigl|E^{\no}(x)\bigr|.\qedhere
\]
\end{proof}

We next show that when $x$ is randomized, the edge sets $E^{\yes}(x)$ and $E^{\no}(x)$ are indistinguishable for algorithms that only take $o(n)$ samples.

\begin{lemma}\label{lem:bipartiteness_indistinguishable}
Suppose there is a randomized map\footnote{A randomized map is a probability distribution over deterministic maps.}
\(
\cA:\binom{[n]\times[k]\times \bF_{2}}{2}^{m}\rightarrow \{0,1\} 
\)
that satisfies the following.
\begin{enumerate}[label=(\arabic*)]
\item For a uniformly random $x\in\bF_{2}^{n\times k}$ and independent edge samples $e_{1},\dots,e_{m}\in E^{\yes}(x)$, we have $\Pr{\cA(e_{1},\dots,e_{m})=1}\geq 2/3$. 
\item For a uniformly random $x\in\bF_{2}^{n\times k}$ and independent edge samples $e_{1},\dots,e_{m}\in E^{\no}(x)$, we have $\Pr{\cA(e_{1},\dots,e_{m})=0}\geq 2/3$. 
\end{enumerate}
Then we must have $m\geq n/3$.
\end{lemma}

\begin{proof}
In the two assumptions on $\cA$ stated in the lemma, the input $(e_{1},\dots,e_{m})$ to $\cA$ follows two different distributions. It suffices to show that these two distributions over $\binom{[n]\times [k]\times \bF_{2}}{2}^{m}$, which we denote by $\cD^{\yes}$ and $\cD^{\no}$, respectively, have total variation distance less than $1/3$ if $m<n/3$. Both $\cD^{\yes}$ and $\cD^{\no}$ can be alternatively generated by first sampling edges $\{u_{1},v_{1}\},\dots,\{u_{m},v_{m}\}$ uniformly at random from the edge set
\[
E_{n,k}=\left\{\bigl\{(a,i),(b,j)\bigr\}\,\middle|\,a,b\in [n],\;\{i,j\}\in \binom{[k]}{2}\right\},
\]
and then letting
\[
e_{i}=\big\{(u_{i},t_{i}),(v_{i},s_{i})\big\}\text{ for some suitably chosen }s_{i},t_{i}\in \bF_{2}
\]
for all $i\in [m]$. Note that the first step (choosing $u_{i}$'s and $v_{i}$'s) is identical for $\cD^{\yes}$ and $\cD^{\no}$, while the second step may be implemented differently for the two. Furthermore, if the collection $\bigl\{\{u_{1},v_{1}\},\dots,\{u_{m},v_{m}\}\bigr\}$ sampled in the first step does not contain a cycle or repeated edges, the second step is also identical for $\cD^{\yes}$ and $\cD^{\no}$. Since $\big\{\{u_{1},v_{1}\},\dots,\{u_{m},v_{m}\}\big\}$ contains a cycle or repeated edges with probability at most (by union bound)
\[
\sum_{r=2}^{+\infty}(kn)^{r}\cdot\frac{m^{r}}{\binom{k}{2}^{r}n^{2r}}\leq \sum_{r=2}^{+\infty}\left(\frac{m}{n}\right)^{r},
\]
we have $\|\cD^{\yes}-\cD^{\no}\|_{\mathrm{TV}}\leq \sum_{r=2}^{+\infty}\left(m/n\right)^{r}<1/3$ if $m<n/3$.
\end{proof}

\begin{corollary}\label{cor:bipartiteness_lower_bound}
For any $k\geq 3$ and any fixed simple graph $H$ with $(k-1)$ vertices and at least one edge, we have $\sam\big(\cG^{H\textup{-hom}}_{2kn},k^{-2}\big)\geq n/3$.
\end{corollary}

\begin{proof}
For any $x\in \bF_{2}^{n\times k}$, since the graph $\big([n]\times [k]\times \bF_{2},\,E^{\yes}(x)\big)$ is bipartite by Proposition~\ref{prop:property_of_Eyes_Eno}(1), it is also homomorphic to $H$ (because we can map the vertex set $[n]\times [k]\times \bF_{2}$ homomorphically to the two endpoints of a single edge in $H$). On the other hand, it follows from Proposition~\ref{prop:property_of_Eyes_Eno}(2) that if we let $\mu$ denote the uniform distribution over $E^{\no}(x)$ (considered as an edge set over $[2kn]$), then
\[
\mu\big(E^{\no}(x)\triangle E'\big)\geq k^{-2}
\]
for any $E'\in\cG^{H\textup{-hom}}_{2kn}$ (because any graph homomorphic to $H$ is also homomorphic to the $(k-1)$-vertex complete graph). Therefore, any sample-based distribution-free tester for $\cG^{H\textup{-hom}}_{2kn}$ with proximity parameter $\varepsilon=k^{-2}$ and sample complexity $m$, when considered as a randomized map $\cA:\binom{[n]\times [k]\times \bF_{2}}{2}^{m}\rightarrow \{0,1\}$, must satisfy the conditions of Lemma~\ref{lem:bipartiteness_indistinguishable} and hence $m\geq n/3$.
\end{proof}

Combining Corollaries~\ref{cor:bipartiteness_upper_bound} and~\ref{cor:bipartiteness_lower_bound} yields Theorem~\ref{thm:graph_homomorphism}.

\section{Upper Bound for Square-Freeness}\label{sec:square_upper_bound}

In this section, we prove our main result $\sam\big(\cG^{\textup{squ}}_{n},\varepsilon\big)\leq O(n^{9/8}/\varepsilon)$, following the ideas outlined in Section~\ref{subsec:overview_graph}. In Section~\ref{subsec:birthday_lemmas}, we develop some new birthday-paradox-type lemmas (similar to Lemma~\ref{lem:hypergraph_birthday_informal}). Then we show how to apply them to the problem of testing square-freeness in Section~\ref{subsec:case_analysis}. 

\subsection{Birthday Paradox Lemmas}\label{subsec:birthday_lemmas}

\subsubsection{Birthday Paradox in Grids}

Suppose there is a probability distribution over the cells in a grid with $n$ rows and $r$ columns. The classical birthday paradox states that in $O(\sqrt{n})$ samples from the distribution, with high probability there are two samples falling in the same row. It turns out that for our applications, it is important to additionally ask that the two samples fall in \emph{different cells} of the same row. However, such an event is no longer guaranteed to happen with high probability if the distribution over cells is arbitrary (for example, consider the case where the distribution is supported on a single column). The following lemma identifies a condition (on the distribution) under which this event must happen with high probability in $O(\sqrt{n})$ samples.
 
\begin{lemma}\label{lem:birthday_minus_max}
Let $\varepsilon,\delta\in (0,1)$ be constants, and let $r,n$ be positive integers. Suppose $p\in [0,1]^{n\times r}$ is a sub-probability vector such that
\[
\sum_{a=1}^{n}\left(\sum_{b=1}^{r}p_{ab}-\max_{b\in [r]}p_{ab}\right)\geq \varepsilon.
\]
Then, in
\(
m=64\left\lceil \varepsilon^{-1}\log(2/\delta)\sqrt{n}\right\rceil
\)
independent samples drawn from $p$, the probability is at least $1-\delta$ that there exist two sampled pairs $(a,b)$ and $(a,c)$ with the same first coordinate $a\in[n]$ but different second coordinates $b,c\in[r]$.
\end{lemma}

\begin{proof}
We define a sub-probability vector $q\in [0,1]^{n}$ by letting
\[
q_{a}=\sum_{b=1}^{r}p_{ab}-\max_{b\in [r]}p_{ab}
\]
for each $a\in [n]$. Let $A=\{a\in [n]\mid q_{a}\geq \varepsilon/(2n)\}$. Since $\sum_{a=1}^{n}q_{a}\geq \varepsilon$ and \[\sum_{a\in [n]\setminus A}q_{a}<n\cdot \varepsilon/(2n)=\varepsilon/2,\]
it follows that $\sum_{a\in A}q_{a}\geq \varepsilon/2$.

\paragraph{Processing first half of samples.} Let $X_{1},\dots,X_{m}$ be a sequence of $m$ independent samples drawn from $p$. For integers $k=1,2,\dots, m/2 $, we iteratively define (random) subsets $S_{k}\subseteq [n]$ and $Q_{k}\subseteq [n]\times [r]$ using the following procedure. 
\begin{enumerate}
\item Initialize $S_{0}=\emptyset$.
\item Repeat the following for $k=1,2,\dots, m/2$:
\begin{enumerate}[label=(\roman*)]
\item If $X_{k}=\nil$ or the first coordinate of $X_{k}$ is not in $A\setminus S_{k-1}$, let $S_{k}=S_{k-1}$ and $Q_{k}=\emptyset$.
\item If $X_{k}=(a,b)$ for some $a\in A\setminus S_{k-1}$ and $b\in [r]$, let $S_{k}=S_{k-1}\cup\{a\}$ and $Q_{k}=\{(a,c)\mid c\in [r]\setminus \{b\}\}$.
\end{enumerate}
\end{enumerate}
We then define random variables $Z_{k}$ by
\[
Z_{k}:=\begin{cases}
1, &\text{if }\sum_{a\in S_{k-1}}q_{a}\geq \varepsilon/4,\\
q_{a^*},&\text{if }\sum_{a\in S_{k-1}}q_{a}<\varepsilon/4\text{ and }S_{k}\setminus S_{k-1}\text{ is a singleton set }\{a^*\},\\
0,&\text{if }\sum_{a\in S_{k-1}}q_{a}<\varepsilon/4\text{ and }S_{k}\setminus S_{k-1}=\emptyset.
\end{cases}
\]
for all positive integers $k\leq  m/2$. It is easy to see that for any such $k$, we have\footnote{To see this, notice that if $\sum_{a\in S_{k-1}}q_{a}\geq \varepsilon/4$, then $Z_{k}=1$ with conditional probability 1. Otherwise, the (conditional) probability that the first coordinate of $X_{k}$ lies in $A\setminus S_{k-1}$ is at least $\left(\sum_{a\in A}q_{a}\right)-\left(\sum_{a\in S_{k-1}}q_{a}\right)\geq \varepsilon/4$.}
\[
\Pr{Z_{k}\geq \varepsilon/(2n)\,\middle|\,X_{1},\dots,X_{k-1}}\geq \varepsilon/4.
\]
Therefore, the random variable $2\varepsilon^{-1}n\cdot\sum_{k=1}^{ m/2}Z_{k}$ stochastically dominates the sum of $ m/2$ independent Bernoulli random variables with mean $\varepsilon/4$. By the multiplicative Chernoff bound, we have
\[\Pr{2\varepsilon^{-1}n\cdot \sum_{k=1}^{ m/2 }Z_{k}\leq \frac{\varepsilon m}{16}}\leq \exp\left(-\frac{1}{8}\cdot\frac{\varepsilon m}{8}\right)\leq \frac{\delta }{2}.\]
Since $\sum_{a\in S_{m/2}}q_{a}\geq \min\left\{\varepsilon/4,\sum_{k=1}^{m/2}Z_{k}\right\}$ by construction, it therefore follows that
\begin{equation}\label{eq:first_half_of_samples}
\Pr{\sum_{a\in S_{m/2}}q_{a}< \min\left\{\frac{\varepsilon^{2}m}{32n},\frac{\varepsilon}{4}\right\}}\leq \frac{\delta}{2}.
\end{equation}

\paragraph{Processing second half of samples.} Now consider the set $Q=\bigcup_{k=1}^{m/2}Q_{k}\subseteq [n]\times [r]$. If any of the second half of samples $X_{m/2+1},\dots,X_{m}$ falls in $Q$, by definition there exist a pair $(a,b)\in \{X_{1},\dots,X_{m/2}\}$ and a pair $(a,c)\in \{X_{m/2+1},\dots,X_{m}\}$ with the same first coordinate $a\in[n]$ but different second coordinates $b,c\in[r]$. So it suffices to show that with probability at least $1-\delta$, at least one of $X_{m/2+1},\dots,X_{m}$ falls in $Q$. Since
\[
\sum_{(a,b)\in Q}p_{ab}\geq \sum_{a\in S_{m/2}}\left(\sum_{b=1}^{r}p_{ab}-\max_{b\in [r]}p_{ab}\right)=\sum_{a\in S_{m/2}}q_{a},
\]
we have
\begin{align}
&\quad \Pr{Q\cap\{X_{m/2+1},\dots,X_{m}\}=\emptyset\,\middle|\, \sum_{a\in S_{m/2}}q_{a}\geq \min\left\{\frac{\varepsilon^{2}m}{32n},\frac{\varepsilon}{4}\right\}}\nonumber\\
&\leq \left(1-\min\left\{\frac{\varepsilon^{2}m}{32n},\frac{\varepsilon}{4}\right\}\right)^{m/2}\leq \exp\left(-\min\left\{\frac{\varepsilon^{2}m^{2}}{64n},\frac{\varepsilon m}{8}\right\}\right)\leq \frac{\delta }{2}.\label{eq:second_half_of_samples}
\end{align}
The desired conclusion then follows by combining \eqref{eq:first_half_of_samples} and \eqref{eq:second_half_of_samples}.
\end{proof}

\subsubsection{Vertex Cover and Fractional Matching}

As discussed in Section~\ref{subsec:overview_graph}, the duality between vertex covers and fractional matchings is the key idea behind Lemma~\ref{lem:hypergraph_birthday_informal}. The duality argument is formalized in the following lemma.

\begin{lemma}\label{lem:LP_weak_duality}
Let $\varepsilon\in (0,1)$ be a constant, and let $G=(V,E)$ be a $k$-uniform hypergraph. Suppose $p\in[0,1]^{V}$ is a sub-probability vector such that for any vertex cover $C$ of $G$ we have $\sum_{v\in C}p_{v}\geq \varepsilon$. Then there exists a sub-probability vector $\lambda=(\lambda_{e})_{e\in E}\in [0,1]^{E}$ such that the following holds:
\begin{enumerate}[label=(\arabic*)]
\item The indicator vectors $1_{e}\in \{0,1\}^{V}$,\footnote{The indicator vector $1_{e}$ is defined by $1_{e}(v)=1$ if $v\in e$ and $1_{e}(v)=0$ if $v\not\in e$, for vertices $v\in V$.} for $e\in \supp(\lambda)$, are linearly independent in $\bR^{V}$.
\item We have the coordinate-wise vector inequality $\sum_{e\in E}\lambda_{e}\cdot 1_{e}\preceq p$.
\item The sum $\sum_{e\in E}\lambda_{e}$ lies in the range $[\varepsilon/k,1/k]$.
\end{enumerate}
\end{lemma}

\begin{proof}
Consider the following linear program over the variables $\lambda_{e}$, for $e\in E$:
\begin{alignat}{2}
\text{maximize }&&\sum_{e\in E}\lambda_{e}\nonumber\\
\text{subject to }&&\sum_{e\in E}\lambda_{e}\cdot 1_{e}&\preceq p \label{eq:LP_domination}\\
&& \lambda_{e}&\geq 0,\quad\text{for all }e\in E.\label{eq:LP_nonnegativity}
\end{alignat}
Let $\lambda^*=(\lambda^*_{e})_{e\in E}$ be an optimal solution to the linear program with minimum possible support size. Consider the edge set
\[
C:=\left\{v\in V\,\middle|\,\sum_{e\in E}\lambda^*_{e}\cdot 1_{e}(v)=p_{v}\right\}.
\]
By the optimality of $\lambda^*$, it is easy to see that $C$ is a vertex cover of $G$, and hence
\[
\varepsilon\leq \sum_{v\in C}p_{v}=\sum_{v\in C}\sum_{e\in E}\lambda^*_{e}\cdot 1_{e}(v)\leq k\sum_{e\in E}\lambda^*_{e}.
\]
Note that the constraint \eqref{eq:LP_domination} implies
\[
\sum_{e\in E}\lambda^*_{e}=\frac{1}{k}\sum_{v\in V}\sum_{e\in E}\lambda^*_{e}\cdot 1_{e}(v)\leq \frac{1}{k}\sum_{v\in V}p_{v}\leq \frac{1}{k},
\]
and hence the vector $\lambda^*$ satisfies the requirement (3) of the lemma. Combined with \eqref{eq:LP_nonnegativity}, this also implies $\lambda^*$ is a sub-probability vector. The requirement (2) is obviously satisfied by $\lambda^*$ due to \eqref{eq:LP_domination}. It now remains to show that $\lambda^*$ satisfies requirement (1) of the lemma.

Suppose the indicator vectors $1_{e}\in \{0,1\}^{V}$, for $e\in \supp(\lambda^*)$, are not linearly independent. Then there exists a not-all-zero coefficient vector $c=(c_{e})_{e\in \supp(\lambda^*)}$ such that
\begin{equation}\label{eq:LP_feasibility_after_move}
\sum_{e\in \supp(\lambda^*)}c_{e}\cdot 1_{e}=0.
\end{equation}
In particular, we have
\begin{equation}\label{eq:LP_optimality_after_move}
\sum_{e\in \supp(\lambda^*)}c_{e}=\frac{1}{k}\sum_{v\in V}\sum_{e\in \supp(\lambda^*)}c_{e}\cdot 1_{e}(v)=0.
\end{equation}
Therefore, there exists some $e\in \supp(\lambda^*)$ such that $c_{e}<0$. Define a positive number $d$ to be
\[
d:=\min_{e\in \supp(\lambda^*),\;c_{e}<0}\left(-\frac{\lambda^*_{e}}{c_{e}}\right).
\]
It is easy to see that $\lambda^*+d\cdot c$ is another optimal solution to the linear program introduced at the beginning of the proof (the feasibility is due to \eqref{eq:LP_feasibility_after_move} and the definition of $d$, while the optimality is due to \eqref{eq:LP_optimality_after_move}). Furthermore, by the definition of $d$, the support size of $\lambda^*+d\cdot c$ is smaller than the support size of $\lambda^*$ by at least 1. This contradicts the definition of $\lambda^*$. Therefore, the indicator vectors $1_{e}$, for $e\in \supp(\lambda^*)$, must be linearly independent.
\end{proof}

We note that a weaker version of the above lemma, where the linear independence condition in the first item is replaced with $|\supp(\lambda)|\leq |V|$, serves as a crucial step in the proof of Lemma~\ref{lem:hypergraph_birthday_informal} by \cite{chen2024distribution}. As discussed in Section~\ref{subsec:overview_graph}, for applications in testing square-freeness, we need to extract this step from \cite[Proof of Lemma 2.2]{chen2024distribution} and strengthen it slightly (from $|\supp(\lambda)|\leq |V|$ to linear independence).

\subsubsection{Birthday Paradox in Hypergraphs}

Another component of \cite{chen2024distribution}'s proof of Lemma~\ref{lem:hypergraph_birthday_informal} is a ``classical birthday paradox'' on grids (somewhat similar to Lemma~\ref{lem:birthday_minus_max}). For our applications, we also need to slightly strengthen this component (\cite[Lemma 6.5]{chen2024distribution}), as stated in the next lemma. 

\begin{lemma}\label{lem:classical_birthday}
Let $k,n$ be positive integers. Suppose $q\in [0,1]^{n+1}$ is a sub-probability vector such that $\sum_{i=1}^{n+1}q_{i}= 1/k$. Define a probability vector $\widetilde{q}\in [0,1]^{(n+1)\times k}$ such that $\widetilde{q}_{ij}=q_{i}$ for each $(i,j)\in [n+1]\times [k]$. For any integer $m$ with
\[
m\geq 18k\left(\sum_{i=1}^{n}q_{i}^{k}\right)^{-1/k},
\]
in $m$ independent samples from $\widetilde{q}$, with probability at least $4/5$ there exists $i\in [n]$ such that the pairs $(i,j)$, for $j\in [k]$, are all sampled.
\end{lemma}

\begin{proof}
Let $x\in \bN^{(n+1)\times k}$ be the empirical count vector of the $m$ samples. Let $y\in \bN^{(n+1)\times k}$ be a random vector such that the coordinates $y_{ij}$, for $(i,j)\in [n+1]\times [k]$, are mutually independent, and each coordinate $y_{ij}$ follows the Poisson distribution with mean $mq_{i}$. A basic property of Poisson distribution is that for any nonnegative integer $t$, the distribution of $y$ conditioned on the event $\left\{\sum_{i=1}^{n+1}\sum_{j=1}^{k}y_{ij}=t\right\}$ is exactly the distribution of the empirical count vector of $t$ independent samples from $\widetilde{q}$. Furthermore, since $\sum_{i=1}^{n+1}\sum_{j=1}^{k}y_{ij}$ follows the Poisson distribution with mean
\[
\sum_{i=1}^{n+1}\sum_{j=1}^{k}mq_{i}=m,
\]
it is easy to see (for example by the Berry Esseen theorem) that
\[
\Prs{y}{\sum_{i=1}^{n+1}\sum_{j=1}^{k}y_{ij}\leq m}\geq \frac{1}{4}.
\]
It is then straightforward to deduce that the distribution of $y$ is $(1,\frac{1}{4})$-dominated by the distribution of $x$, by constructing a coupling as per Definition~\ref{def:stochastic_domination}.

Let $S\subseteq \bN^{(n+1)\times k}$ be the downward-closed subset defined by
\[
S=\left\{z\in \bN^{(n+1)\times k}\,\middle|\,\forall i\in [n],\,\exists j\in [k]\text{ such that }z_{ij}=0\right\}.
\]
By independence between the coordinates of $y$, we have
\[
\Prs{y}{y\in S}=\prod_{i=1}^{n}\Pr{\big.\exists j\in [k]\text{ such that }y_{ij}=0}=\prod_{i=1}^{n}\left(1-\big(1-\exp(-mq_{i})\big)^{k}\right).
\]
If there exists $i\in [n]$ such that $mq_{i}\geq 3k$, then
\[
\Prs{y}{y\in S}\leq 1-\big(1-\exp(-3k)\big)^{k}\leq k\exp(-3k)\leq \frac{1}{20}.
\]
If $mq_{i}\leq 3k$ for all $i\in [n]$, then $\exp(-mq_{i})\leq 1-mq_{i}/(6k)$ for all $i$ and we have
\[
\Prs{y}{y\in S}\leq \prod_{i=1}^{n}\left(1-\left(\frac{mq_{i}}{6k}\right)^{k}\right)\leq \exp\left(-\sum_{i=1}^{n}\left(\frac{mq_{i}}{6k}\right)^{k}\right)\leq \exp(-3^{k})\leq \frac{1}{20}.
\]
Therefore, in either case we have $\Pr{y\in S}\leq 1/20$. It then follows from Proposition~\ref{prop:stochastic_domination} and the conclusion of the last paragraph that
\[
\Prs{x}{x\in S}\leq 4\cdot \Prs{y}{y\in S}\leq \frac{1}{5},
\]
as desired.
\end{proof}

We are now ready to prove Lemma~\ref{lem:hypergraph_birthday_informal}, the formal version of which is given below.

\begin{lemma}[{\cite[Lemma 2.2]{chen2024distribution}}]\label{lem:hypergraph_birthday}
Let $\varepsilon\in (0,1)$ be a constant, and let $G=(V,E)$ be a $k$-uniform hypergraph. Suppose $p\in[0,1]^{V}$ is a sub-probability vector such that for any vertex cover $C$ of $G$ we have $\sum_{v\in C}p_{v}\geq \varepsilon$. Then for any integer $m$ with
\[
m\geq \frac{18k^{2}|V|^{(k-1)/k}}{\varepsilon},
\]
in $m$ independent samples from $p$, with probability at least $4/5$ there exists an edge in $E$ such that all vertices of the edge are sampled.
\end{lemma}
\begin{proof}
We first apply Lemma~\ref{lem:LP_weak_duality} to obtain a sub-probability vector $\lambda\in [0,1]^{E}$ that satisfies the three conditions stated in Lemma~\ref{lem:LP_weak_duality}. By the first condition, the indicator vectors $1_{e}$, for $e\in \supp(\lambda)$, are linearly independent in $\bR^{V}$. In particular, this means $|\supp(\lambda)|\leq |V|$. We denote the elements of $\supp(\lambda)$ by $e_{1},\dots,e_{n}$. By the third condition in Lemma~\ref{lem:LP_weak_duality}, we know that
\[
\frac{\varepsilon}{k}\leq \sum_{i=1}^{n}\lambda_{e_{i}}\leq \frac{1}{k}.
\]

There obviously exists a map $\varphi:[n+1]\times [k]\rightarrow V\cup\{\nil\}$ that maps the set $\{i\}\times [k]$ bijectively to the vertices of $e_{i}$ for each $i\in [n]$, and maps the set $\{n+1\}\times [k]$ to $\{\nil\}$. Consider the probability vector $\widetilde{q}\in [0,1]^{(n+1)\times k}$ defined by
\begin{alignat*}{2}
\widetilde{q}_{ij}=q_{i}&=\lambda_{e_{i}} \quad&&\text{for each }(i,j)\in [n]\times [k],\\
\widetilde{q}_{n+1,j}=q_{n+1}&=\frac{1}{k}-\sum_{i=1}^{n}\lambda_{e_{i}}\quad&&\text{for each }j\in [k].
\end{alignat*}
It follows from the second condition in Lemma~\ref{lem:LP_weak_duality} that
\begin{equation}\label{eq:hypergraph_birthday_domination}
\sum_{(i,j)\in \varphi^{-1}(v)}\widetilde{q}_{ij}=\sum_{i=1}^{n}\lambda_{e_{i}}\cdot 1_{e_{i}}(v)\leq p_{v},\quad\text{for each }v\in V.
\end{equation}

Since
\[
m\geq \frac{18k^{2}|V|^{(k-1)/k}}{\varepsilon}\geq \frac{18kn^{(k-1)/k}}{\sum_{i=1}^{n}q_{i}}\geq \frac{18k}{\left(\sum_{i=1}^{n}q_{i}^{k}\right)^{1/k}},
\]
Lemma~\ref{lem:classical_birthday} implies that in $m$ independent samples from $\widetilde{q}$, with probability at least $4/5$ there exists $i\in [n]$ such that the pairs $(i,j)$, for $j\in [k]$, are all sampled. Due to the stochastic domination given by  \eqref{eq:hypergraph_birthday_domination}, it follows that in $m$ independent samples from $p$, with probability at least $4/5$ there exists an edge in $\{e_{1},\dots,e_{n}\}$ such that all vertices of the edge are sampled.
\end{proof}

Although Lemma~\ref{lem:hypergraph_birthday} is not used in the analysis of the square-freeness tester (only its predecessors Lemmas~\ref{lem:LP_weak_duality} and~\ref{lem:classical_birthday} are), it immediately implies Theorem~\ref{thm:subgraph_freeness_upper_bound} and in particular the triangle-freeness upper bound $\sam\big(\cG^{\textup{tri}}_{n},\varepsilon\big)\leq O(n^{4/3}/\varepsilon)$.

\subsection{The Case Analysis}\label{subsec:case_analysis}

In this section, we prove the upper bound $\sam\big(\cG^{\textup{squ}}_{n},\varepsilon\big)\leq O(n^{9/8}/\varepsilon)$. As discussed in Section~\ref{subsec:overview_graph}, the core of the proof is a case analysis where we divide into the ``dilute'' case and the ``concentrated'' case. In the following, we first formalize the notion of diluteness; then the two cases will be handled in Section~\ref{subsubsec:dilute} and~\ref{subsubsec:concentrated}, respectively.

\subsubsection{From Edges to Squares to Edges}\label{subsubsec:witness_and_descendent}

The following two definitions will play a key role in the formalization of diluteness.

\begin{definition}\label{def:witness}
Suppose $p\in [0,1]^{\binom{[n]}{2}}$ is a sub-probability vector and $\varepsilon\in (0,1)$ is a constant. A sub-probability vector $q\in [0,1]^{\Square(n)}$ that satisfies the following conditions is called an \emph{$\varepsilon$-square-witness} of $p$:
\begin{enumerate}[label=(\arabic*)]
\item The indicator vectors $1_{\zeta}\in \{0,1\}^{\binom{[n]}{2}}$, for $\zeta\in \supp(q)$, are linearly independent in $\bR^{\binom{[n]}{2}}$.
\item We have $\sum_{\zeta\in \Square(n)}q_{\zeta}\cdot 1_{\zeta}\preceq  p$.
\item We have $\sum_{\zeta\in \Square(n)}q_{\zeta}\geq \varepsilon/4$.  
\end{enumerate}
\end{definition}

\begin{definition}\label{def:descendant}
Suppose $q\in [0,1]^{\Square(n)}$ is a sub-probability vector. A sub-probability vector $p'\in [0,1]^{\binom{[n]}{2}}$ is called a \emph{descendant} of $q$ if there exists an injective map $\varphi:\supp(q)\rightarrow \binom{[n]}{2}$ such that the following holds:
\begin{enumerate}[label=(\arabic*)]
\item For each square $\zeta\in \supp(q)$, the image $\varphi(\zeta)$ is an edge contained in $\zeta$.
\item For each square $\zeta\in \supp(q)$, we have $p'_{\varphi(\zeta)}=q_{\zeta}$.
\item For each edge $\{a,b\}$ that is not an image of $\varphi$, we have $p'_{ab}=0$.
\end{enumerate}
\end{definition}

The next lemma ensures that the objects in Definitions~\ref{def:witness} and~\ref{def:descendant} are obtainable if we start with a sub-probability vector that is $\varepsilon$-far from square-free.

\begin{lemma}\label{lem:existence_witness_descendant}
Let $\varepsilon\in (0,1)$ be a constant, and suppose $p\in [0,1]^{\binom{[n]}{2}}$ is a sub-probability vector that is $\varepsilon$-far from square-free. Then there exist sub-probability vectors $q\in [0,1]^{\Square(n)}$ and $p'\in [0,1]^{\binom{[n]}{2}}$ such that $q$ is an $\varepsilon$-square-witness of $p$ and $p'$ is a descendant of $q$.
\end{lemma}

\begin{proof}
Consider the $4$-uniform hypergraph whose vertex set is $\binom{[n]}{2}$ and whose edge set is $\Square(n)$. Since $p\in [0,1]^{\binom{[n]}{2}}$ is $\varepsilon$-far from square-free, the condition in Lemma~\ref{lem:LP_weak_duality} is satisfied. The existence of an $\varepsilon$-square-witness $q\in [0,1]^{\Square(n)}$ therefore follows directly from the conclusion of Lemma~\ref{lem:LP_weak_duality}.

We now claim that if $q\in [0,1]^{\Square(n)}$ is a sub-probability vector whose indicator vectors $(1_{\zeta})_{\zeta\in \supp(q)}$ are linearly independent in $\bR^{\binom{[n]}{2}}$, then there exists a sub-probability vector $p'\in [0,1]^{\binom{[n]}{2}}$ that is a descendant of $q$.

By Definition~\ref{def:descendant}, it suffices to construct an injective map
\(
\varphi:\supp(q)\rightarrow \binom{[n]}{2}
\)
such that for every square $\zeta\in\supp(q)$ the edge $\varphi(\zeta)$ belongs to $\zeta$. 

To this end, consider the bipartite incidence relation between $\supp(q)$ and $\binom{[n]}{2}$, where a square $\zeta\in\supp(q)$ is adjacent to an edge $e\in\binom{[n]}{2}$ whenever $e\in\zeta$. By Hall's matching theorem, it suffices to show that for every subset $Z\subseteq \supp(q)$, the set of edges covered by $Z$ has size at least $|Z|$. Let
\(
E(Z):=\bigcup_{\zeta\in Z}\zeta
\)
denote the set of edges covered by $Z$.

Suppose for contradiction that $|E(Z)|\le |Z|-1$. Since each indicator vector satisfies
\(
1_{\zeta}=\sum_{e\in\zeta}1_{\{e\}}
\), it follows that
\[
\mathrm{span}\big((1_{\zeta})_{\zeta\in Z}\big)
\subseteq
\mathrm{span}\big((1_{\{e\}})_{e\in E(Z)}\big).
\]
Hence
\[
\dim \mathrm{span}\big((1_{\zeta})_{\zeta\in Z}\big)
\le
\dim \mathrm{span}\big((1_{\{e\}})_{e\in E(Z)}\big)
\le |E(Z)|
\le |Z|-1,
\]
which contradicts the assumed linear independence of the vectors
$(1_{\zeta})_{\zeta\in\supp(q)}$.
\end{proof}

As stated in Lemma~\ref{lem:existence_witness_descendant}, our intention is to start with an arbitrary sub-probability vector $p\in [0,1]^{\binom{[n]}{2}}$ that is $\varepsilon$-far from square-free, first transition to a ``witness'' vector in the space $[0,1]^{\Square(n)}$, and then transition back to a ``descendant'' vector in the original space $[0,1]^{\binom{[n]}{2}}$. The reason for transitioning from edge distributions to square distributions is not hard to understand: similarly to the analysis of the triangle-freeness tester, the ``witness'' vector acts as a ``fractional matching'' on which Lemma~\ref{lem:classical_birthday} can be applied. 

\begin{lemma}\label{lem:concentrated_case_prelim}
Let $\varepsilon\in (0,1)$ be a constant, and suppose $p\in [0,1]^{\binom{[n]}{2}}$ and $q\in [0,1]^{\Square(n)}$ are sub-probability vectors such that $q$ is an $\varepsilon$-square-witness of $p$. For any integer $m$ with
\[
m\geq 72\left(\sum_{\zeta\in \Square(n)}q_{\zeta}^{4}\right)^{-1/4},
\]
in $m$ independent samples from $p$, with probability at least $4/5$ there exists a square $\zeta\in \Square(n)$ such that all edges of $\zeta$ are sampled.
\end{lemma}

\begin{proof}
The conclusion follows easily from Lemma~\ref{lem:classical_birthday} and stochastic domination, similarly to the proof of Lemma~\ref{lem:hypergraph_birthday}.
\end{proof}

However, as will become clear in Section~\ref{subsubsec:case_analysis}, we will only apply Lemma~\ref{lem:concentrated_case_prelim} in the ``concentrated'' case; in the ``dilute'' case we will have to transition back to the ``descendant'' vector in $[0,1]^{\binom{[n]}{2}}$.

\subsubsection{The Diluteness Notion}\label{subsubsec:case_analysis}

The notion of ``diluteness'' hinges upon the following three natural definitions. First of all, a distribution over $\binom{[n]}{2}$ naturally induces a distribution over the vertex set $[n]$:

\begin{definition}
Let $p\in [0,1]^{\binom{[n]}{2}}$ be a sub-probability vector. We define a sub-probability mass function $\deg_{p}:[n]\rightarrow [0,1]$ by letting
\[
\deg_{p}(a)=\frac{1}{2}\sum_{b\in [n]\setminus\{a\}}p_{ab}\quad\text{for all }a\in [n].
\]
\end{definition}

The function $\deg_{p}(\cdot)$ is a sub-probability mass function because of the identity \[\sum_{a\in [n]}\deg_{p}(a)=\sum_{\{a,b\}\in \binom{[n]}{2}}p_{ab}.\] 

A distribution over $\binom{[n]}{2}$ also induces a distribution over the collection $\Wedge(n)$ (see Section~\ref{subsec:general_notations}), since the wedges $(\{a,c\},b)$ correspond to \emph{length-2 walks} (from $a$ to $b$ to $c$) on $[n]$.

\begin{definition}\label{def:walk}
Let $p\in [0,1]^{\binom{[n]}{2}}$ be a sub-probability vector, and let $B\subseteq [n]$ be a subset of vertices. We define a sub-probability mass function $\Walk[p,B]:\mathsf{Wedge}(n)\rightarrow [0,1]$ as follows. For all wedges $\big(\{a,c\},b\big)\in \Wedge(n)$, we let
\[
\Walk[p,B]\big(\{a,c\},b\big):=
\begin{cases}
p_{ab}p_{bc}/(2\deg_{p}(b)), &\text{if }b\in B,\\
0, &\text{if }b\not\in B,
\end{cases}
\]
with the convention that expressions of the form $0/0$ are taken to be $0$. When $B=[n]$, we abbreviate $\Walk[p]:=\Walk[p,[n]]$
\end{definition}

Note that for any sub-probability vector $p\in [0,1]^{\binom{[n]}{2}}$ we have
\begin{equation}\label{eq:total_mass_walk}
\sum_{(\{a,c\},b)\in \Wedge(n)}\Walk[p,B]\big(\{a,c\},b\big)\leq \sum_{b\in B}\frac{1}{2\deg_{p}(b)}\cdot\frac{1}{2}\left(\sum_{a\in [n]\setminus\{b\}}p_{ab}\right)^{2}=\sum_{b\in B}\deg_{p}(b)\leq 1.
\end{equation}

By taking the marginal on the start and end vertices of the distribution on length-2 walks, we naturally obtain a distribution over ``hops.'' 

\begin{definition}\label{def:hop}
Let $p\in [0,1]^{\binom{[n]}{2}}$ be a sub-probability vector, and let $B\subseteq [n]$ be a subset of vertices. We define two functions $\Hop[p,B],\HopD[p,B]:\binom{[n]}{2}\rightarrow [0,1]$ as follows. For all $\{a,c\}\in \binom{[n]}{2}$, we let 
\begin{align}
\Hop[p,B](a,c)&:=\sum_{b\in [n]\setminus \{a,c\}}\Walk[p,B]\big(\{a,c\},b\big),\nonumber\\
\HopD[p,B](a,c)&:=\Hop[p,B](a,c)-\max_{b\in [n]\setminus \{a,c\}}\Walk[p,B]\big(\{a,c\},b\big).\label{eq:def_HopD}
\end{align}
When $B=[n]$, we abbreviate $\Hop[p]:=\Hop[p,[n]]$ and $\HopD[p]:=\HopD[p,[n]]$.
\end{definition}

We are now ready to define the notion of diluteness: given a sub-probability vector $p\in [0,1]^{\binom{[n]}{2}}$ that is far from square-free, we classify it as ``dilute'' if for some \emph{descendant} $p'$ of some \emph{$\varepsilon$-square-witness} $q$ of $p$, the sub-probability mass function $\HopD[p']$ has sufficiently large total mass. Note that the following lemma will be applied to the descendant $p'$ instead of the original vector $p$.

\begin{lemma}[Dilute case lemma]\label{lem:dilute_case}
Let $\varepsilon\in (0,1)$ be a constant. Suppose $p\in [0,1]^{\binom{[n]}{2}}$ is a sub-probability vector such that
\begin{equation}\label{eq:dilute_case_assumption}
\sum_{\{a,c\}\in \binom{[n]}{2}}\HopD[p](a,c)\geq \varepsilon.
\end{equation}
Then, in $m=10^{3}\left\lceil\varepsilon^{-1}n\log(12n)\right\rceil$ independent samples drawn from $p$, the probability is at least $2/3$ that there exist four sampled pairs $\{a,b\}$,~$\{b,c\}$,~$\{c,d\}$ and $\{d,a\}$ forming a square. 
\end{lemma}

Lemma~\ref{lem:dilute_case} will be proved in Section~\ref{subsubsec:dilute}.

Intuitively speaking, if the total mass of $\HopD[p']$ is too small, then either the first inequality of \eqref{eq:total_mass_walk} is very loose, or the subtraction of the maximum in \eqref{eq:def_HopD} loses too much weight. In both cases, we can argue that $p'$ is ``concentrated'' in a suitable sense. This intuition will be formalized in Section~\ref{subsubsec:concentrated}, where we prove the following lemma. 

\begin{lemma}[Concentrated case lemma]\label{lem:concentrated_case}
Let $\varepsilon\in (0,1)$ be a constant. Suppose $p\in [0,1]^{\binom{[n]}{2}}$ is a sub-probability vector such that
\begin{equation}\label{eq:concentrated_case_assumption}
\sum_{\{a,b\}\in \binom{[n]}{2}}p_{ab}-\sum_{\{a,c\}\in \binom{[n]}{2}}\HopD[p](a,c)\geq \varepsilon.
\end{equation}
Then we have
\[
\sum_{\{a,b\}\in \binom{[n]}{2}}p_{ab}^{4}\geq \frac{2\varepsilon^{4}}{n^{9/2}}.
\]
\end{lemma}

Before proving Lemmas~\ref{lem:dilute_case} and~\ref{lem:concentrated_case}, we first derive from these lemmas the desired upper bound for testing square-freeness:

\begin{theorem}\label{thm:square_freeness_upper}
Let $\varepsilon\in (0,1)$ be a constant and let $n$ be a positive integer. Suppose $p\in [0,1]^{ \binom{[n]}{2}}$ is a sub-probability vector that is $\varepsilon$-far from square-free. Then in $O(n^{9/8}/\varepsilon)$ independent samples from $p$, with probability at least $2/3$ there exist four sampled pairs $\{a,b\},\{b,c\},\{c,d\}$ and $\{d,a\}$ forming a square.
\end{theorem}
\begin{proof}[Proof assuming Lemmas~\ref{lem:dilute_case} and~\ref{lem:concentrated_case}]
We first apply Lemma~\ref{lem:existence_witness_descendant} to obtain sub-probability vectors $q\in [0,1]^{\Square(n)}$ and $p'\in [0,1]^{\binom{[n]}{2}}$ such that $q$ is an $\varepsilon$-square-witness of $p$ and $p'$ is a descendant of $q$. By Definitions~\ref{def:descendant} and~\ref{def:witness} we know that
\[
\sum_{\{a,b\}\in \binom{[n]}{2}}p'_{ab}= \sum_{\zeta\in \Square(n)}q_{\zeta}\geq \frac{\varepsilon}{4}.
\]
Therefore, we have either
\begin{equation}\label{eq:dilute_case}
\sum_{\{a,c\}\in \binom{[n]}{2}}\HopD[p'](a,c)\geq \frac{\varepsilon}{8}
\end{equation}
or
\begin{equation}\label{eq:concentrated_case}
\sum_{\{a,b\}\in \binom{[n]}{2}}p'_{ab}-\sum_{\{a,c\}\in \binom{[n]}{2}}\HopD[p'](a,c)\geq \frac{\varepsilon}{8}.
\end{equation}

If \eqref{eq:dilute_case} holds, then by Lemma~\ref{lem:dilute_case} we know that it takes $O(\varepsilon^{-1}n\log n)$ samples from $p'$ to get all four edges of a square with probability at least $2/3$. Since Definitions~\ref{def:descendant} and~\ref{def:witness} implies
\[
p'\preceq \sum_{\zeta\in\Square(n)}q_{\zeta}\cdot 1_{\zeta}\preceq p,
\]
it also takes at most $O(\varepsilon^{-1}n\log n)$ samples from $p$ to get all four edges of a square with probability at least $2/3$.

If \eqref{eq:concentrated_case} holds, then by Lemma~\ref{lem:concentrated_case} we know that
\[
\frac{2(\varepsilon/8)^{4}}{n^{9/2}}\leq \sum_{\{a,b\}\in \binom{[n]}{2}}(p'_{ab})^{4}=\sum_{\zeta\in \Square(n)}q_{\zeta}^{4}.
\]
Therefore, it follows from Lemma~\ref{lem:concentrated_case_prelim} that it takes at most $O(n^{9/8}/\varepsilon)$ samples from $p$ to get all four edges of a square with probability at least $2/3$.

In conclusion, in either of the cases it takes at most $O(n^{9/8}/\varepsilon)$ samples from $p$ to get all four edges of a square with probability at least $2/3$.
\end{proof}

\begin{corollary}\label{cor:square_freeness_upper}
We have $\sam\big(\cG^{\textup{squ}}_{n},\varepsilon\big)\leq O(n^{9/8}/\varepsilon)$.
\end{corollary}

\begin{proof}
We show that the canonical one-sided-error tester (see Remark~\ref{rem:canonical_tester_monotone}) rejects with probability at least $2/3$ if the unknown edge set $E$ and distribution $\mu$ over $\binom{[n]}{2}$ satisfy $\mu(E\triangle E')\geq \varepsilon$ for any $E'\in \cG^{\textup{squ}}_{n}$. Consider the sub-probability vector $p\in [0,1]^{\binom{[n]}{2}}$ defined by 
\[p_{e}=\begin{cases}
\mu(\{e\}),&\text{if }e\in E\\
0, &\text{if }e\not\in E.
\end{cases}\] 
It is clear that $E$ is $\varepsilon$-far from square-free under $\mu$ if and only if $p$ is $\varepsilon$-far from square-free. The conclusion thus follows from Theorem~\ref{thm:square_freeness_upper}.
\end{proof}

\subsubsection{The Dilute Case}\label{subsubsec:dilute}

The main idea of the proof of Lemma~\ref{lem:dilute_case} is to apply Lemma~\ref{lem:birthday_minus_max}, the birthday paradox lemma on grids. The grid on which we will apply the birthday paradox is $\binom{[n]}{2}\times [n]$, which contains $\Wedge(n)$ as a subset (see \eqref{eq:wedge_in_grid}). In order to apply Lemma~\ref{lem:birthday_minus_max}, we have to be able to sample \emph{cells} of the grid, which in our setting corresponds to sampling from the sub-probability mass function $\Walk[p]$ given sampling access to the sub-probability vector $p\in [0,1]^{\binom{[n]}{2}}$. A natural way to sample wedges given edge samples is the following:

\begin{definition}
Let $p\in [0,1]^{\binom{[n]}{2}}$ be a sub-probability vector. For any integer $m\geq 1$, we define $\bcW(p,m)$ to be the output distribution of the following sampling process:
\begin{enumerate}
\item Independently sample $m$ edges from the distribution $p$, and let $x\in \bN^{\binom{[n]}{2}}$ be the empirical count vector of the $m$ samples.
\item For each wedge $(\{a,c\},b)\in \Wedge(n)$, we let $w_{ac,b}=x_{ab}x_{bc}$.
\item Output the vector $w\in \bN^{\Wedge(n)}$.
\end{enumerate}
\end{definition}

The next lemma shows that if the degree of every vertex is not too small, the sampling process defined above is sufficient for simulating access to $\Walk[p]$ (as far as stochastic domination is concerned).

\begin{lemma}\label{lem:dilute_domination}
Let $\varepsilon,\delta\in (0,1)$ be constants. Suppose $B\subseteq [n]$ is a subset of vertices and $p\in [0,1]^{\binom{[n]}{2}}$ is a sub-probability vector such that $\deg_{p}(b)\geq \varepsilon/(2n)$ for all $b\in B$. For any integer $m\geq 6\varepsilon^{-1}n\log(2\delta^{-1}n)$, we have (see Definition~\ref{def:stochastic_domination})
\[\bcS(\Walk[p,B],m)\leq_{(1-\delta,1)}\bcW(p,5m).\]
\end{lemma}
\begin{proof}
We divide the proof into 4 steps.

\paragraph{Step 1: preparation.} Let 
\begin{equation}\label{eq:5m_samples_from_p}
(a_{1},b_{1}),(a_{2},b_{2}),\dots,(a_{m},b_{m}),(b_{m+1},c_{m+1}),(b_{m+2},c_{m+2}),\dots,(b_{5m},c_{5m})
\end{equation}
be a sequence of $5m$ independent samples drawn from $p$.\footnote{If the $i$-th sample we get is $\nil$ for some $i\in [m]$, we let $(a_{i},b_{i})=(\nil,\nil)$. Similarly, if the $i$-th sample we get is $\nil$ for some $i\in \{m+1,m+2,\dots,5m\}$, we let $(b_{i},c_{i})=(\nil,\nil)$.} For each element $b\in B$, we let $I_{1}(b)$ be the set of indices $i\in \{1,2,\dots,m\}$ such that $b_{i}=b$, and let $I_{2}(b)$ be the set of indices $i\in \{m+1,\dots,5m\}$ such that $b_{i}=b$. Note that $|I_{1}(b)|$ is the sum of $m$ independent Bernoulli random variables with mean $\deg_{p}(b)$. Therefore, by the multiplicative Chernoff bound we have
\begin{equation}\label{eq:I1b_concentration}
\Pr{|I_{1}(b)|\geq 2m\deg_{p}(b)}\leq \exp\left(-\frac{1}{3}\cdot m\deg_{p}(b)\right)\leq \frac{\delta}{2n},
\end{equation}
using the guaranteed lower bounds on $\deg_{p}(b)$ and $m$. Similarly, since the expected value of $|I_{2}(b)|$ is $4m\deg_{p}(b)$, we have
\begin{equation}\label{eq:I2b_concentration}
\Pr{|I_{2}(b)|\leq 2m\deg_{p}(b)}\leq \exp\left(-\frac{1}{8}\cdot 4m\deg_{p}(b)\right)\leq\frac{\delta}{2n}.
\end{equation}

\paragraph{Step 2: construction of coupling.} We now describe a procedure that generates $m$ independent samples $X_{1},\dots,X_{m}$ from $\Walk[p,B]$, based on the $5m$ samples \eqref{eq:5m_samples_from_p} from $p$. We repeat the following for $i=1,2,\dots,m$:
\begin{enumerate}
\item If $b_{i}\not\in B$, let $X_{i}=\nil$.
\item If $b_{i}\in B$, suppose $i$ is the $j$-th smallest index in the set $I_{1}(b_{i})$.
\begin{enumerate}[label=(\roman*)]
\item If $j>|I_{2}(b_{i})|$, draw a random vertex $c^*\in [n]\setminus \{b_{i}\}$ according to the probability vector
\[
\left(\frac{p_{b_{i}c}}{2\deg_{p}(b_{i})}\right)_{c\in [n]\setminus \{b_{i}\}}\in [0,1]^{[n]\setminus \{b_{i}\}}.
\]
If $c^*=a_{i}$, let $X_{i}=\nil$. Otherwise let $X_{i}=(\{a_{i},c^*\},b_{i})$.
\item If $j\leq |I_{2}(b_{i})|$, let $\ell$ be the $j$-th smallest index in the set $I_{2}(b_{i})$. If $c_{\ell}=a_{i}$, let $X_{i}=\nil$. Otherwise, let $X_{i}=(\{a_{i},c_{\ell}\},b_{i})$.
\end{enumerate}
\end{enumerate}

\paragraph{Step 3: analysis of coupling.} It is easy to see that if \eqref{eq:5m_samples_from_p} are $5m$ independent samples drawn from $p$, the above procedure produces $m$ independent samples from $\Walk[p,B]$. Indeed, for each wedge $(\{a,c\},b)\in \Wedge(n)$ such that $b\in B$, the probability that $(a_{i},b_{i})=(a,b)$ is $p_{ab}/2$, and 
\[
\Pr{X_{i}=(\{a,c\},b)\,\big|\, (a_{i},b_{i})=(a,b)}=\frac{p_{bc}}{2\deg_{p}(b)}.
\]
Therefore we have
\begin{align*}
\Pr{X_{i}=(\{a,c\},b)}&=\Pr{(a_{i},b_{i})=(a,b)}\cdot\Pr{X_{i}=\{a,c\}\,\big|\, (a_{i},b_{i})=(a,b)}+\\
&\qquad\quad\Pr{(a_{i},b_{i})=(c,b)}\cdot \Pr{X_{i}=\{a,c\}\,\big|\, (a_{i},b_{i})=(c,b)}\\
&=\left(\frac{p_{ab}}{2}\cdot\frac{p_{bc}}{2\deg_{p}(b)}+\frac{p_{bc}}{2}\cdot\frac{p_{ab}}{2\deg_{p}(b)}\right)=\Walk[p,B]\big(\{a,c\},b\big).
\end{align*}

\paragraph{Step 4: wrapping up.} Given $5m$ independent samples \eqref{eq:5m_samples_from_p} drawn from $p$, we let $x\in \bN^{\binom{[n]}{2}}$ be the empirical count vector of the $5m$ samples. Define a vector $w\in \bN^{\Wedge(n)}$ by letting
\(w_{ac,b}=x_{ab}x_{bc}\)
for all $(\{a,c\},b)\in \Wedge(n)$. We generate $m$ samples $X_{1},\dots,X_{m}$ according to Step 2, and let $z\in \bN^{\Wedge(n)}$ be the empirical count vector of the samples $X_{1},\dots,X_{m}$. Finally, let $\rho$ be the joint distribution of the vector pair $(w,z)\in \bN^{\Wedge(n)}\times \bN^{\Wedge(n)}$. 

It is clear that the marginal distribution of $\rho$ in the $w$ coordinate is identical to $\bcW(p,5m)$, while its marginal distribution in the $z$ coordinate is identical to $\bcS(\Walk[p,B],m)$. Note that whenever $|I_{1}(b)|\leq |I_{2}(b)|$ holds for every $b\in B$, the case 2(i) in the procedure of Step 2 is never activated, which leads to $w_{ac,b}=x_{ab}x_{bc}\geq z_{ac,b}$ for all $(\{a,c\},b)\in\Wedge(n)$. Therefore, by the concentration inequalities \eqref{eq:I1b_concentration}, \eqref{eq:I2b_concentration} and a union bound over all $b\in B$, we have
\[
\Pru{(w,z)\sim \rho}{w\succeq z}\geq 1-\sum_{b\in B}\Pr{|I_{1}(b)|>|I_{2}(b)|\big.}\geq 1-2|B|\cdot\frac{\delta}{2n}\geq 1-\delta,
\]
as desired.
\end{proof}

The next lemma is a standard argument showing that vertices with too small degrees can be safely ignored when choosing the middle vertex of a length-2 walk.

\begin{lemma}\label{lem:dilute_case_Markov}
Let $\varepsilon\in (0,1)$ be a constant. Suppose $p\in [0,1]^{\binom{[n]}{2}}$ is a sub-probability vector and $B$ is the set of vertices $b\in [n]$ such that $\deg_{p}(b)\geq \varepsilon/(2n)$. Then we have 
\[\sum_{\{a,c\}\in \binom{[n]}{2}}\HopD[p](a,c)-\sum_{\{a,c\}\in \binom{[n]}{2}}\HopD[p,B](a,c)\leq \frac{\varepsilon
}{2}.\]
\end{lemma}

\begin{proof}
By Definition~\ref{def:walk}, the function $\Walk[p,B]$ is no larger than the function $\Walk[p]$ on any input, so for any $\{a,c\}\in \binom{[n]}{2}$ we have
\[
\max_{b\in [n]\setminus \{a,c\}}\Walk[p,B]\big(\{a,c\},b\big)\leq \max_{b\in [n]\setminus \{a,c\}}\Walk[p]\big(\{a,c\},b\big).
\]
By Definition~\ref{def:hop}, it follows that
\begin{equation}\label{eq:hop_comparison}
\HopD[p](a,c)-\HopD[p,B](a,c)\leq \Hop[p](a,c)-\Hop[p,B](a,c).
\end{equation}
Expanding and rearranging using Definitions~\ref{def:walk} and~\ref{def:hop}, we have
\begin{align*}
\sum_{\{a,c\}\in \binom{[n]}{2}}\Hop[p](a,c)-\sum_{\{a,c\}\in \binom{[n]}{2}}\Hop[p,B](a,c)&=\sum_{b\in [n]\setminus B}\sum_{\{a,c\}\in \binom{[n]\setminus \{b\}}{2}}\Walk[p]\big(\{a,c\},b\big)\\
&\leq \frac{1}{2}\sum_{b\in [n]\setminus B}\sum_{a,c\in [n]\setminus \{b\}}\frac{p_{ab}p_{bc}}{2\deg_{p}(b)}\\
&=\frac{1}{2}\sum_{b\in [n]\setminus B}2\deg_{p}(b)\\
&\leq (n-|B|)\cdot \frac{\varepsilon}{2n}\leq \frac{\varepsilon}{2}.
\end{align*}
Combining this with \eqref{eq:hop_comparison} immediately yields the conclusion.
\end{proof}

We are now ready to prove the dilute case lemma, Lemma~\ref{lem:dilute_case}.

\begin{proof}[Proof of Lemma~\ref{lem:dilute_case}]
Let $B$ be the set of vertices $b\in [n]$ such that $\deg_{p}(b)\geq \varepsilon/(2n)$. By Lemma~\ref{lem:dilute_case_Markov} and the assumption \eqref{eq:dilute_case_assumption}, we have
\[
\sum_{\{a,c\}\in \binom{[n]}{2}}\HopD[p,B](a,c)\geq \frac{\varepsilon
}{2}.
\]

We now apply Lemma~\ref{lem:birthday_minus_max} to the sub-probability mass function $\Walk[p,B]$ over the set of wedges $\Wedge(n)\subseteq \binom{[n]}{2}\times [n]$. It follows that given at least
\[
\frac{m}{5}\geq 64\left\lceil(\varepsilon/2)^{-1}\log(4/\delta)\sqrt{\binom{n}{2}}\right\rceil
\]
independent samples from $\Walk[p,B]$, with probability at least $1-\delta/2$ there exist two sampled wedges $(\{a,c\},b)$ and $(\{a,c\},d)$ with the same first coordinate $\{a,c\}\in \binom{[n]}{2}$ and different second coordinates $b,d\in [n]$. Note that since $\Walk[p,B]$ is supported on $\Wedge(n)$, we may assume $a,c,b,d$ are distinct vertices.

We next apply Lemma~\ref{lem:dilute_domination} to the probability vector $p$ and the vertex set $B$. Due to the guaranteed lower bound on $m$, the conclusion of Lemma~\ref{lem:dilute_domination} yields that
\[\bcS(\Walk[p,B],m/5)\leq_{(1-\delta/2,1)}\bcW(p,m).\] Since the set
\[S=\left\{w\in \bN^{\Wedge(n)}\,\middle|\,\text{there are no distinct }a,b,c,d\in [n]\text{ s.t. } w_{ac,b},w_{ac,d}\geq 1\right\}\]
is a downward-closed subset of $\bN^{\Wedge(n)}$, it follows from Proposition~\ref{prop:stochastic_domination} that
\[
\Pru{w\sim\bcW(p,m)}{w\in S}\leq \Pru{w\sim \bcS(\Walk[p,B],m/5)}{w\in S}+\frac{\delta}{2}.
\]
By the conclusion of the last paragraph, the first summand on the right-hand side is at most $\delta/2$. Therefore, we conclude that $\Prs{w\sim \bcW(p,m)}{w\in S}\leq \delta$. In other words, with probability at least $1-\delta$, there exist distinct vertices $a,b,c,d\in [n]$ such that all four pairs $\{a,b\},\{b,c\},\{c,d\},\{d,a\}$ appear in a batch of $m$ independent samples from $p$, as desired.
\end{proof}

\subsubsection{The Concentrated Case}\label{subsubsec:concentrated}

To prove the concentrated case lemma, we need the following result from spectral graph theory.

\begin{proposition}\label{prop:spectral_graph}
Let $n$ be a positive integer, and let $S$ be a symmetric subset of $[n]\times [n]$ (i.e. for any $(i,j)\in S$ we have $(j,i)\in S$ as well). For any real numbers $x_{1},\dots,x_{n}$, we have
\[
\sqrt{|S|}\cdot \sum_{i=1}^{n}x_{i}^{2}\geq \sum_{(i,j)\in S}x_{i}x_{j}.
\]
    
\end{proposition}
\begin{proof}
Consider the symmetric matrix $M\in \{0,1\}^{n\times n}$ defined by
\[
M_{ij}=\begin{cases}
1,&\text{if }(i,j)\in S,\\
0,&\text{if }(i,j)\not\in S,
\end{cases}\qquad\text{for all }(i,j)\in [n]\times [n].
\]
We know that $M$ has $n$ real eigenvalues, and we order them decreasingly as $\lambda_{1}\geq \lambda_{2}\geq \dots\geq \lambda_{n}$. The matrix $\lambda_{1} I_{n}-M$ is positive semi-definite, where $I_{n}$ is the $n\times n$ identity matrix. In particular, we have
\[
0\leq \sum_{i=1}^{n}\sum_{j=1}^{n}x_{i}(\lambda_{1} I_{n}-M)_{ij} x_{j}=\lambda_{1}\sum_{i=1}^{n}x_{i}^{2}-\sum_{(i,j)\in S}x_{i}x_{j}.
\]
The conclusion thus follows from the fact that $\lambda_{1}\leq \sqrt{|S|}$. To see this, note that
\[
\sum_{i=1}^{n}\lambda_{i}^{2}=\mathrm{trace}(M^{2})=\sum_{i=1}^{n}\sum_{j=1}^{n}M_{ij}M_{ji}=|S|,
\]
which clearly implies $\lambda_{1}\leq \sqrt{|S|}$.
\end{proof}

We have now arrived at the crux of the proof --- showing that \eqref{eq:concentrated_case_assumption} implies a nontrivial lower bound on the $\ell_{4}$-norm of the vector $p$.

\begin{proof}[Proof of Lemma~\ref{lem:concentrated_case}]
For each pair $\{a,c\}\in \binom{[n]}{2}$, we let $h\big(\{a,c\}\big)$ be an arbitrary vertex in the set of maximizers 
\[\argmax_{b\in [n]\setminus \{a,c\}}\Walk[p]\big(\{a,c\},b\big).\]
For each $b\in [n]$, we define a set $S_{b}\subseteq [n]\times [n]$ by
\[
S_{b}:=\big\{(a,c)\,\big|\,a,c\in [n]\setminus \{b\}\text{ such that }a\neq c\text{ and }h\big(\{a,c\}\big)=b\big\}\cup\big\{(a,a)\,\big|\,a\in [n]\setminus \{b\}\big\}.
\]
Note that
\begin{equation}\label{eq:sum_of_S}
\sum_{b=1}^{n}|S_{b}|=\sum_{\{a,c\}\in \binom{[n]}{2}}2+\sum_{b=1}^{n}(n-1)=2n(n-1).
\end{equation}
For each $b\in [n]$, we define
\[
\beta_{b}:=\sum_{a,c\in [n]\setminus \{b\}}p_{ab}p_{cb}=(2\deg_{p}(b))^{2}\quad\text{and}\quad \gamma_{b}:=\sum_{(a,c)\in S_{b}}p_{ab}p_{cb}.
\]

Now by Definitions~\ref{def:walk} and~\ref{def:hop} we have
\begin{align*}
\sum_{\{a,c\}\in \binom{[n]}{2}}\HopD[p](a,c)&=\sum_{\{a,c\}\in \binom{[n]}{2}}\Hop[p](a,c)-\sum_{\{a,c\}\in \binom{[n]}{2}}\max_{b\in [n]\setminus\{a,c\}}\Walk[p]\big(\{a,c\},b\big)\\
&=\sum_{\{a,c\}\in\binom{[n]}{2}}\sum_{b\in [n]\setminus\{a,c\}}\frac{p_{ab}p_{cb}}{2\deg_{p}(b)}-\sum_{\{a,c\}\in\binom{[n]}{2}}\Walk[p]\big(\{a,c\},h\big(\{a,c\}\big)\big)\\
&=\frac{1}{2}\sum_{b=1}^{n}\left(\sum_{a,c\in [n]\setminus \{b\}}\frac{p_{ab}p_{cb}}{2\deg_{p}(b)}-\sum_{(a,c)\in S_{b}}\frac{p_{ab}p_{cb}}{2\deg_{p}(b)}\right)\\
&=\frac{1}{2}\sum_{b=1}^{n}\frac{\beta_{b}-\gamma_{b}}{\sqrt{\beta_{b}}}\geq \frac{1}{2}\sum_{b=1}^{n}\frac{\beta_{b}-\gamma_{b}}{\sqrt{\beta_{b}}+\sqrt{\gamma_{b}}}=\frac{1}{2}\sum_{b=1}^{n}\left(\sqrt{\beta_{b}}-\sqrt{\gamma_{b}}\right)\\
&=\frac{1}{2}\sum_{b=1}^{n}\left(2\deg_{p}(b)-\sqrt{\gamma_{b}}\right)=\sum_{\{a,b\}\in \binom{[n]}{2}}p_{ab}-\frac{1}{2}\sum_{b=1}^{n}\sqrt{\gamma_{b}}.
\end{align*}
Therefore, the assumption \eqref{eq:concentrated_case_assumption} implies $\sum_{b=1}^{n}\sqrt{\gamma_{b}}\geq 2\varepsilon$.
On the other hand, we have
\begin{align*}
\sum_{\{a,b\}\in \binom{[n]}{2}}p_{ab}^{4}&=\frac{1}{2}\sum_{b= 1}^{n}\sum_{a\in [n]\setminus \{b\}}p_{ab}^{4}\geq \frac{1}{2}\sum_{b\in [n],\;S_{b}\neq \emptyset}\left(\frac{1}{\sqrt{|S_{b}|}}\sum_{(a,c)\in S_{b}}p_{ab}^{2}p_{cb}^{2}\right)\tag{using Proposition~\ref{prop:spectral_graph}}\\
&\geq\frac{1}{2}\sum_{b\in [n],\;S_{b}\neq \emptyset}\left(\frac{1}{|S_{b}|^{3/2}}\left(\sum_{(a,c)\in S_{b}}p_{ab}p_{cb}\right)^{2}\right)
\tag{by Cauchy-Schwarz}\\
&\geq \frac{1}{2}\cdot\frac{\left(\sum_{b=1}^{n}\sqrt{\gamma_{b}}\right)^{4}}{\left(\sum_{b=1}^{n}\sqrt{|S_{b}|}\right)^{3}}\geq \frac{1}{2}\cdot\frac{\left(\sum_{b=1}^{n}\sqrt{\gamma_{b}}\right)^{4}}{\left(n\sum_{b=1}^{n}|S_{b}|\right)^{3/2}}\tag{by H\"older's inequality}\\
&\geq \frac{1}{8n^{9/2}}\left(\sum_{i=1}^{n}\sqrt{\gamma_{b}}\right)^{4}\geq \frac{2\varepsilon^{4}}{n^{9/2}},\tag{using \eqref{eq:sum_of_S}}
\end{align*}
as desired.
\end{proof}

\section{Upper Bound for Tree-Freeness}\label{sec:tree_upper_bound}

The goal of this section is to prove the upper bounds for testing tree-freeness and testing cliques, as stated in Theorems~\ref{thm:tree-freeness} and~\ref{thm:clique}, respectively. Although they are seemingly unrelated results, the proofs of these two upper bounds are quite similar, and in particular they rely on the same type of birthday-paradox argument. In Sections~\ref{subsec:more_birthday} and~\ref{subsec:tree_induction}, we develop the birthday-paradox-type lemmas underlying the proofs. The two sample complexity upper bounds will then be proved in Section~\ref{subsec:tree_proofs}.

\subsection{More Birthday Paradox}\label{subsec:more_birthday}

In one formulation of the classical birthday paradox, two batches of samples are drawn from the same probability distribution, and the goal is to show that, with high probability, there exists a common sample appearing in both batches. In this subsection, we take this perspective a step further: we show that the set of common samples of the two batches can, in a rough sense, be viewed as a single batch of samples drawn from the same distribution.

The ``effective size'' of this derived batch depends on the sizes of the two original batches. The classical birthday paradox can then be interpreted as establishing that this effective size is at least $1$, and hence that the intersection of the two batches is likely to be non-empty.

The notion of taking the ``intersection'' of two batches of samples can be formalized as follows.

\begin{definition}\label{def:product_of_empirical_indicators}
Suppose $w^{(1)},w^{(2)}\in \bN^{n}$ are vectors with nonnegative integer coordinates. We let $\cP(w^{(1)},w^{(2)})$ be the entry-wise product vector $w\in \bN^{n}$ defined by $w_{b}=w^{(1)}_{b}w^{(2)}_{b}$ for all $b\in [n]$. If $\mu$ and $\nu$ are probability distributions over $\bN^{n}$, let $\bcP(\mu,\nu)$ be the distribution of $\cP(w^{(1)},w^{(2)})$ where $w^{(1)}\sim \mu$ and $w^{(2)}\sim \nu$ are independent random vectors.
\end{definition}

We will also need the following ``matrix-vector multiplication'' version of Definition~\ref{def:product_of_empirical_indicators}.

\begin{definition}\label{def:matrix_vector_product}
Suppose $w^{(1)}\in \bN^{n}$ and $w^{(2)}\in \bN^{n\times n}$. We let $\cJ(w^{(1)},w^{(2)})$ be the vector $w\in \bN^{n}$ defined by
\[
w_{b}=\sum_{a=1}^{n}w^{(1)}_{a}w^{(2)}_{ab}\qquad\text{for all }b\in [n].
\]
Let $\cJ'(w^{(1)},w^{(2)})$ be the vector $w'\in \bN^{n}$ defined by
\[
w_{b}'=\max_{a\in [n]}w^{(1)}_{a}w^{(2)}_{ab}\qquad\text{for all }b\in [n].
\]
If $\mu$ and $\nu$ are probability distributions over $\bN^{n}$ and $\bN^{n\times n}$, respectively, then we let $\bcJ(\mu,\nu)$ be the distribution of $\cJ(w^{(1)},w^{(2)})$ where $w^{(1)}\sim \mu$ and $w^{(2)}\sim \nu$ are independent random vectors. Similarly, let $\bcJ'(\mu,\nu)$ be the corresponding distribution of $\cJ'(w^{(1)},w^{(2)})$.
\end{definition}

To prepare for the main lemma of this subsection, we make the following two standard definitions. The first allows us to take ``marginals'' of sub-probability mass functions:

\begin{definition}
Let $f:[n]^{2}\rightarrow [0,1]$ be a sub-probability mass function. Define  sub-probability mass functions $\pi_{1}f,\pi_{2}f:[n]\rightarrow [0,1]$ by letting
\[
\pi_{1}f(a)=\sum_{b=1}^{n}f(a,b)\qquad\text{for all }a\in [n],\qquad\text{and}\qquad
\pi_{2}f(b)=\sum_{a=1}^{n}f(a,b)\qquad\text{for all }b\in [n].
\]
\end{definition}

The next definition is motivated by the standard trick of ``ignoring elements with too small weights,'' which have already been used in Section~\ref{sec:square_upper_bound} (see, for example, Lemma~\ref{lem:dilute_case_Markov})  and will continue to come into play frequently in this section.

\begin{definition}
For any finite domain $\Lambda$ and two sub-probability mass functions $f,g:\Lambda\rightarrow [0,1]$, we say that $g$ is an $\varepsilon$-pruning of $f$ if $g(x)\leq f(x)$ for all $x\in \Lambda$ and $\sum_{x\in \Lambda}(f(x)-g(x))\leq \varepsilon$.
\end{definition}

We are now ready to state the main lemma of this subsection.

\begin{lemma}\label{lem:tree_bipartite_birthday_paradox}
Let $\beta,\gamma,\delta,\varepsilon\in (0,1)$ and $C>0$ be constants such that $\beta+\gamma\leq 1$ and $\beta\delta\varepsilon\cdot C\geq 16$. For sufficiently large positive integers $n$ and
\[m_{1}=\left\lceil Cn^{1-\beta}\right\rceil,\qquad m_{2}=\left\lceil Cn^{1-\gamma}\right\rceil,\qquad\text{and}\qquad m_{3}=\left\lceil Cn^{1-\beta-\gamma}\right\rceil,\] we have (recall the notion of stochastic domination in Definition~\ref{def:stochastic_domination}):
\begin{enumerate}[label=(\arabic*)]
\item Any sub-probability mass function $f:[n]\rightarrow [0,1]$ has an $\varepsilon$-pruning $g$ such that 
\begin{equation}\label{eq:batch_intersection_dominate}
\bcS(g,m_{3})\leq_{(1-\delta,1)}\bcP\big(\bcS'(g,m_{1}),\bcS(g,m_{2})\big).
\end{equation}
\item Any sub-probability mass function $f:[n]^{2}\rightarrow [0,1]$ has an $\varepsilon$-pruning $g$ such that
\[
\bcS(\pi_{2}g,m_{3})\leq_{(1-\delta,1)}\bcJ\big(\bcS'(\pi_{1}g,m_{1}),\bcS(g,m_{2})\big).
\]
\end{enumerate}
\end{lemma}

For an element $a\in [n]$ with very small weight $f(a)$ under a sub-probability mass function $f:[n]\to [0,1]$, the probability that $a$ appears in the intersection of two independent batches of samples is roughly proportional to $f(a)^2$, whereas the probability that it appears in a single batch is proportional to $f(a)$. Thus, for small values of $f(a)$, the former is significantly smaller than the latter. Consequently, a key difficulty in establishing the stochastic domination in~\eqref{eq:batch_intersection_dominate} is handling elements with small weight.

In particular, if there exist elements with extremely small weight --- namely those $a\in[n]$ with $0 < f(a) \lesssim 1/(Cn)$ --- then it is impossible for $\bcP(\bcS'(f,m_1),\bcS(f,m_2))$ to dominate $\bcS(f,m_3)$. This necessitates an $\varepsilon$-pruning step to exclude such elements. For elements with moderately small weights, for instance those $a\in[n]$ with $f(a)\approx 1/n$, the next lemma provides a useful bound on their appearance in a single batch of samples.

\begin{lemma}\label{lem:sample_number_control}
Let $\beta,\varepsilon,\delta\in (0,1)$ and $C\geq 1$ be constants. For sufficiently large positive integers $n$, the following statement holds. Suppose $f:[n]\rightarrow[0,1]$ is a sub-probability mass function such that for all $a\in [n]$, either $f(a)=0$ or $f(a)\geq \varepsilon/n$. Then for $m=\left\lceil Cn^{1-\beta}\right\rceil$, we have
\[
\Pru{w\sim \bcS(f,m)}{w_{a}\leq \frac{2n}{\beta\varepsilon}\cdot f(a)\textup{ for all }a\in [n]}\geq 1-\delta.
\]
\end{lemma}
\begin{proof}
For each $a\in [n]$ such that $f(a)\neq 0$, the coordinate $w_{a}$ is the sum of $m$ Bernoulli random variables with mean $f(a)$. By Chernoff bound, we have
\begin{equation}\label{eq:binomial_Chernoff_bound}
\Pr{w_{a}\geq t}\leq \left(\frac{4\cdot \Ex{w_{a}}}{t}\right)^{t}\qquad\text{for any }t\geq \Ex{w_{a}}.
\end{equation}
Now let $t_{a}=2\beta^{-1}\varepsilon^{-1}n\cdot f(a)$. Since $f(a)\geq \varepsilon/n$, we have $t_{a}\geq 2\beta^{-1}$. Furthermore, we have
\[
\frac{\Ex{w_{a}}}{t_{a}}=\frac{m\cdot f(a)}{2\beta^{-1}\varepsilon^{-1}n\cdot f(a)}\leq \frac{C\beta \varepsilon }{n^{\beta}}.
\]
Plugging into \eqref{eq:binomial_Chernoff_bound}, it follows that
\[
\Pr{w_{a}\geq t_{a}}\leq \left(\frac{4C\beta\varepsilon}{ n^{\beta}}\right)^{2\beta^{-1}}=\frac{(4C\beta\varepsilon)^{2\beta^{-1}}}{n^{2}}\leq \frac{\delta}{n},
\]
where we used the condition that $n$ is sufficiently large in the last transition. Now, taking a union bound over all $a\in [n]$ such that $f(a)\neq 0$ yields the conclusion.
\end{proof}

We are now ready to prove Lemma~\ref{lem:tree_bipartite_birthday_paradox}.

\begin{proof}[Proof of Lemma~\ref{lem:tree_bipartite_birthday_paradox}]

It is not hard to see that the first statement implies the second statement. In fact, for any sub-probability mass function $f:[n]^{2}\rightarrow [0,1]$, the distribution $\bcS(\pi_{2}f,m_{3})$ equals the output distribution of the following process:
\begin{enumerate}
\item Sample $w^{(1)}\sim \bcS(\pi_{1}f,m_{3})$.
\item Initialize $w^{(2)}\in \{0,1\}^{n}$ to be the all-zero vector. For each $a\in [n]$, repeat the following $w^{(1)}_{a}$ times:
\begin{itemize}
\item Sample an element $b\in [n]$ with probability proportional to $f(a,b)$.
\item Update $w^{(2)}_{b}\gets w_{b}^{(2)}+1$.
\end{itemize}
\item Output $w^{(2)}$.
\end{enumerate}
If the distribution $\bcS(\pi_{1}f,m_{3})$ in the first step is replaced with $\bcP\big(\bcS'(\pi_{1}f,m_{1}),\bcS(\pi_{1}f,m_{2})\big)$, then the distribution of the output in the third step becomes $\bcJ\big(\bcS'(\pi_{1}f,m_{1}),\bcS(f,m_{2})\big)$. Therefore, to prove the second statement of Lemma~\ref{lem:tree_bipartite_birthday_paradox}, we apply the first statement to the sub-probability mass function $\pi_{1}f$. This yields an $\varepsilon$-pruning $g_{1}:[n]\rightarrow [0,1]$ of $\pi_{1}f$ such that
\[
\bcS(g_{1},m_{3})\leq_{(1-\delta,1)}\bcP\big(\bcS'(g_{1},m_{1}),\bcS(g_{1},m_{2})\big).
\]
Since there clearly exists an $\varepsilon$-pruning $g$ of $f$ such that $\pi_{1}g=g_{1}$, it follows from the argument above that for this $g$, we have
\[
\bcS(\pi_{2}g,m_{3})\leq_{(1-\delta,1)}\bcJ\big(\bcS'(\pi_{1}g,m_{1}),\bcS(g,m_{2})\big).
\]

In the rest of the proof, we prove the first statement of Lemma~\ref{lem:tree_bipartite_birthday_paradox}.

Fix an arbitrary sub-probability mass function $f:[n]\rightarrow [0,1]$. We define $g:[n]\rightarrow [0,1]$ by
\[
g(a)=f(a)\cdot\ind{f(a)\geq \frac{\varepsilon}{n}}\qquad\text{for all }a\in [n].
\]
It is clear that $g\leq f$ pointwise and
\[
\sum_{a=1}^{n}(f(a)-g(a))=\sum_{a=1}^{n}f(a)\cdot\ind{f(a)< \frac{\varepsilon}{n}}<n\cdot \frac{\varepsilon}{n}=\varepsilon.
\]
So $g$ is an $\varepsilon$-pruning of $f$. Furthermore, for each $a\in [n]$, either $g(a)=0$ or $g(a)\geq \varepsilon/n$.

We next define and analyze three sampling processes $\frakP_{1}$, $\frakP_{2}$ and $\frakP_{2}'$. Note that both $\frakP_{2}$ and $\frakP_{2}'$ operate on the output of $\frakP_{1}$.  

\paragraph{The process $\frakP_{1}$.} Let $b_{1},\dots,b_{m_{1}}$ be a sequence of $m_{1}$ independent samples drawn from $g$.

\paragraph{Analysis of $\frakP_{1}$.} For each $a\in [n]$, let $I_{a}$ be the set of indices $i\in [m_{1}]$ such that $b_{i}=a$. Let $\cE_{1}$ be the event that \[|I_{a}|\leq 2\beta^{-1}\varepsilon^{-1}n\cdot g(a)\qquad\text{for any }a\in[n].\]
It follows from Lemma~\ref{lem:sample_number_control} that 
\begin{equation}\label{eq:cE_1_prob}
\Prs{\frakP_{1}}{\cE_{1}}\geq 1-\frac{\delta}{2}
\end{equation}
when $n$ is sufficiently large.

\paragraph{The process $\frakP_{2}$.} If the event $\cE_{1}$ does not happen, output the zero vector in $\bN^{n}$. If the event $\cE_{1}$ happens, run the following procedure:
\begin{enumerate}
\item Let $a_{1},\dots,a_{m_{2}}$ be a sequence of $m_{2}$ independent samples drawn from $g$.
\item For each $i\in [m_{2}]$, if $a_{i}=\nil$ then let $r_{i}=\nil$. If $a_{i}\neq \nil$ (i.e. $a_{i}\in [n]$), let $r_{i}$ be a uniformly random element of $I_{a_{i}}$ with probability
\[
p_{i}=\frac{\beta \varepsilon\cdot|I_{a_{i}}|}{2n\cdot g(a_{i})}\in [0,1],
\]
and let $r_{i}=\nil$ with probability $1-p_{i}$. 
\item Initialize $w\in \bN^n$ to be the all-zero vector, and initialize a variable $a\gets 0$.
\item For each $i\in [m_{2}]$, if $r_{i}\neq \nil$, update $a\gets a_{i}$ and $w_{a}\gets w_{a}+1$.
\item Output $w$.
\end{enumerate}

\paragraph{Analysis of $\frakP_{2}$.} We assume $\cE_{1}$ happens. It is easy to see that $r_{1},r_{2},\dots,r_{m_{2}}$ are independent random variables, each being a uniformly random element of $[m_{1}]$ with probability
\[
p=\sum_{a=1}^{n}g(a)\cdot \frac{\beta \varepsilon\cdot|I_{a}|}{2n\cdot g(a)}=\frac{\beta \varepsilon m_{1}}{2n}\leq 1,
\]
and being $\nil$ with probability $1-p$. Let $S=\{i\in [m_{2}]\mid r_{i}\neq \nil\}$, and let $T=\{r_{i}\mid i\in S\}\subseteq [m_{1}]$. Let $s=|S|$ and $t=|T|$. Let $\cE_{2}$ be the event that $t\geq m_{3}$. We next show that (when $n$ is sufficiently large)
\begin{equation}\label{eq:cE_2_prob}
\Prs{\frakP_{2}}{\cE_{2}\mid \cE_{1}}\geq 1-\frac{\delta}{2}.
\end{equation}

Note that when $n$ is sufficiently large, 
\[\Ex{s}=m_{2}p=\frac{\beta\varepsilon}{2}\cdot\frac{m_{1}m_{2}}{n}\geq \frac{\beta\varepsilon}{2}\cdot C^{2}n^{1-\beta-\gamma}\geq 4m_{3}\]
Thus, by Chernoff bound, we have
\begin{equation}\label{eq:S_at_most_3m_3}
\Pr{s\leq 3m_{3}}\leq \exp(-m_{3}/8)\leq \frac{\delta}{4}.
\end{equation}
Conditioned on any fixed $S$ with $|S|\geq 3m_{3}$, we have
\begin{equation}\label{eq:expectation_R_bound}
\Ex{t\mid S}\geq m_{1}\left(1-\left(1-\frac{1}{m_{1}}\right)^{3m_{3}}\right)\geq 2m_{3},
\end{equation}
where we used $m_{1}\geq 4m_{3}$ in the last transition. Conditioned on $S$, the random variables $\ind{r\in T}$ (where $r$ ranges in $[m_{1}]$) are pairwise negatively correlated. So we have
\begin{equation}\label{eq:variance_R_bound}
\Var{t\mid S}\leq \sum_{r=1}^{m_{1}}\Var{\ind{r\in T}\,\big|\, S}\leq\sum_{r=1}^{m_{1}}\Ex{\ind{r\in T}\,\big|\, S}=\Ex{t\mid S}.
\end{equation}
It then follows from Chebyshev's inequality that
\begin{align*}
\Pr{t\leq m_{3}\mid S}&\leq \frac{\Var{t\mid S}}{\bigl(\Ex{t\mid S}-m_{3}\bigr)^{2}}\tag{using~\eqref{eq:expectation_R_bound}}\\
&\leq \frac{\Ex{t\mid S}}{\bigl(\Ex{t\mid S}-m_{3}\bigr)^{2}}\tag{using \eqref{eq:variance_R_bound}}\\
&\leq \frac{2m_{3}}{m_{3}^{2}}\leq \frac{\delta}{4}\tag{using \eqref{eq:expectation_R_bound}}.
\end{align*}
Combining the above with \eqref{eq:S_at_most_3m_3}, we obtain \eqref{eq:cE_2_prob}.

\paragraph{The process $\frakP_{2}'$.} Recall that $b_{1},\dots,b_{m_{1}}\in [n]$ are the samples drawn in the process $\frakP_{1}$. Let $T'\subseteq [m_{1}]$ be a uniformly random subset of size $m_{3}$, and output the empirical count vector of the sequence $\big(b_{r}\big)_{r\in T'}$.

\paragraph{Analysis of $\frakP_{2}'$. } For fixed samples $b_{1},\dots,b_{m_{1}}$ drawn in the process $\frakP_{1}$ such that $\cE_{1}$ happens, we consider the output distributions of $\frakP_{2}$ and $\frakP_{2}'$ when running on $b_{1},\dots,b_{m_{1}}$. Since the set $T$ defined in the analysis of $\frakP_{2}$ is a uniformly random subset of $[m_{1}]$ with (random) size $t$, it follows that conditioned on $\cE_{2}=\{t\geq m_{3}\}$, the output distribution of $\frakP_{2}$ dominates the output distribution of $\frakP_{2}'$ (recall Definition~\ref{def:stochastic_domination}).  

\paragraph{Putting things together.} We use $\frakP_{2}'\circ \frakP_{1}$ to denote the output distribution of $\frakP_{2}'$ running on the output of $\frakP_{1}$. Note that $\frakP_{2}'\circ \frakP_{1}$ is a distribution over $\bN^{n}$. It is easy to see that
\begin{equation}\label{eq:domination_1}
\bcS(g,m_{3})=\frakP_{2}'\circ \frakP_{1}
\end{equation}
Analogously, we use $\frakP_{2}\circ \frakP_{1}$ to denote the output distribution of  $\frakP_{2}$ running on the output of $\frakP_{1}$. By the definition of $\frakP_{2}$, it is easy to see that 
\begin{equation}\label{eq:domination_2}
\frakP_{2}\circ \frakP_{1}\leq_{(1,1)}\bcP\big(\bcS'(g,m_{1}),\bcS(g,m_{2})\big).
\end{equation}
Furthermore, as we have argued, conditioned on the event $\cE_{1}\cap \cE_{2}$ we have
\begin{equation}\label{eq:domination_3}
\big[\frakP_{2}'\circ \frakP_{1}\,\big|\,\cE_{1}\cap\cE_{2}\big]\leq_{(1,1)}\big[\frakP_{2}\circ \frakP_{1}\,\big|\, \cE_{1}\big].
\end{equation}
Since $\Pr{\cE_{1}\cap\cE_{2}}\geq 1-\delta$ by \eqref{eq:cE_1_prob} and \eqref{eq:cE_2_prob}, combining \eqref{eq:domination_1},~\eqref{eq:domination_2} and~\eqref{eq:domination_3} yields
\[
\bcS(g,m_{3})\leq_{(1-\delta,\,1)}\bcP\big(\bcS'(g,m_{1}),\bcS(g,m_{2})\big).\qedhere
\]
\end{proof}

By the definition of empirical count vectors (see Section~\ref{subsec:general_notations}), Lemma~\ref{lem:tree_bipartite_birthday_paradox} immediately implies the following corollary:

\begin{corollary}\label{cor:tree_bipartite_birthday_paradox}
Under the same notation and conditions as in Lemma~\ref{lem:tree_bipartite_birthday_paradox}, we have:
\begin{enumerate}[label=(\arabic*)]
\item Any sub-probability mass function $f:[n]\rightarrow [0,1]$ has an $\varepsilon$-pruning $g$ such that
\[
\bcS'(g,m_{3})\leq_{(1-\delta,1)}\bcP\big(\bcS'(g,m_{1}),\bcS'(g,m_{2})\big).
\]
\item Any sub-probability mass function $f:[n]^{2}\rightarrow [0,1]$ has an $\varepsilon$-pruning $g$ such that
\[
\bcS'(\pi_{2}g,m_{3})\leq_{(1-\delta,1)}\bcJ'\big(\bcS'(\pi_{1}g,m_{1}),\bcS'(g,m_{2})\big).
\]
\end{enumerate}
\end{corollary}
\subsection{Induction on the Number of Edges}\label{subsec:tree_induction}

Our main idea for proving the sample complexity upper bound in Theorem~\ref{thm:tree-freeness} is to induct on the number of edges in the tree. However, we cannot directly use the statement of Theorem~\ref{thm:tree-freeness} as an induction hypothesis. Instead, we will formulate a ``tree version'' of the birthday-paradox-type statement in Corollary~\ref{cor:tree_bipartite_birthday_paradox} that is specifically designed to be provable by induction.

For the convenience of the induction argument, we view the edges of a tree as directed edges that converge to a designated root vertex.

\begin{definition}
Let $V$ be a finite set. Given a finite set $T$ of $(|V|-1)$ ordered pairs $(u,v)\in V^{2}$ and a distinguished element $v^*\in V$, the set $T$ is called a \emph{directed rooted tree} on $V$ with root $v^*$ if the following hold: 
\begin{enumerate}[label=(\arabic*)]
\item For each $(u,v)\in T$, we have $u\neq v$.
\item For every $u\in V\setminus\{v^*\}$, there is exactly one $v\in V$ such that $(u,v)\in T$.
\item Every vertex has a path to $v^*$: for every $v_{0}\in V$, there is an integer $\ell\geq 0$ and elements $v_{1},\dots,v_{\ell}\in V$ such that $v_{\ell}=v^*$ and $(v_{i},v_{i+1})\in T$ for all $i\in \{0,1,\dots,\ell-1\}$.
\end{enumerate}
\end{definition}

To formulate a ``tree version'' of Definition~\ref{def:matrix_vector_product}, we make the following two standard definitions.

\begin{definition}
Fix a directed rooted tree $T$ on a finite set $V$. Given a map $\varphi:T\rightarrow [n]^{2}$ and a vector $y\in [n]^{V}$, we say $\varphi$ is \emph{compatible with} $y$ if $\varphi(u,v)=(y_{u},y_{v})$ for all $(u,v)\in T$. 
\end{definition}

\begin{definition}
Let $V$ be a finite set, and let  $f:[n]^{V}\rightarrow [0,1]$ be a sub-probability mass function. For any subset $U\subseteq V$, we define a sub-probability mass function $\pi_{U}[f]:[n]^{U}\rightarrow [0,1]$ by letting 
\[
\pi_{U}[f](z)=\sum_{y\in [n]^{V}}\ind{y_{u}=z_{u}\text{ for all }u\in U}\cdot f(y)\qquad\text{for all }z\in [n]^{U}.
\]
When the cardinality of $U$ is 1 or 2, we slightly abuse the notation as follows. For any two distinct vertices $u,v\in V$, define $\pi_{u,v}f:[n]^{2}\rightarrow [0,1]$ by letting
\[
\pi_{u,v}f(a,b)=\sum_{y\in[n]^{V}}\ind{y_{u}=a\textup{ and }y_{v}=b}\cdot f(y)\qquad\textup{for all }a,b\in [n].
\]
For any single element $v\in V$, analogously define $\pi_{v}f:[n]\rightarrow [0,1]$ by letting
\[
\pi_{v}f(a)=\sum_{y\in[n]^{V}}\ind{y_{v}=a}\cdot f(y)\qquad\text{for all }a\in [n].
\]
\end{definition}

The ``tree version'' of Definition~\ref{def:matrix_vector_product} can now be stated as follows.

\begin{definition}\label{def:J_tree}
Suppose $T$ is a directed rooted tree on a finite set $V$ with root $v^*$. Given a sub-probability mass function $f:[n]^{V}\rightarrow [0,1]$ and a positive integer $m$, let $\bcJ^{T}_{v^*}(f,m)$ be the output distribution of the following process:
\begin{enumerate}
\item For each pair $(u,v)\in T$, independently draw $m$ samples from the sub-probability mass function $\pi_{u,v}f$, and let $X^{(u,v)}\subseteq [n]^{2}$ be the set formed by the $m$ samples.
\item Initialize $w\in \{0,1\}^{n}$ to be the all-zero vector. For each $b\in [n]$, let $w_{b}=1$ if there exists a map $\varphi:T\rightarrow [n]^{2}$ such that
\begin{itemize}
\item $\varphi$ is compatible with some vector $y\in [n]^{V}$ such that $y_{v^*}=b$; and
\item $\varphi(u,v)\in X^{(u,v)}$ for each $(u,v)\in T$.
\end{itemize}
\item Output the vector $w$.
\end{enumerate}
\end{definition}

The next lemma is the ``tree version'' of Corollary~\ref{cor:tree_bipartite_birthday_paradox}, and is proved via induction on the number of edges in the tree.

\begin{lemma}\label{lem:tree_birthday}
Let $k,t$ be positive integers such that $k\geq t$. Let $\delta,\varepsilon\in (0,1)$ and $\delta\varepsilon\cdot C\geq 16k$ be constants. The following statement holds for sufficiently large positive integers $n$. Suppose $T$ is a directed rooted tree on a finite set $V$ with root $v^*$, where $|V|=t+1$. If \[m_{1}=\left\lceil Cn^{(k-1)/k}\right\rceil\qquad\text{and}\qquad m_{t}=\left\lceil Cn^{(k-t)/k}\right\rceil\] 
then any sub-probability mass function $f:[n]^{V}\rightarrow [0,1]$ has a $(t-1)\varepsilon$-pruning $g$ such that
\[
\bcS'(\pi_{v^*}g,m_{t})\leq_{(1-(t-1)\delta ,1)}\bcJ^{T}_{v^*}(f,m_{1}).
\]
\end{lemma}
\begin{proof}
We proceed by induction on $t$. The base case $t=1$ is straightforward: when $T$ consists of a single edge, for any sub-probability mass function $f:[n]^{V}\rightarrow [0,1]$ and $m_1=\left\lceil C n^{(k-1)/k}\right\rceil$ we have
\[
\bcS'(\pi_{v^*}f,m_{1})=\bcJ^{T}_{v^*}(f,m_{1}).
\]
In the following, we assume $t\geq 2$ and the statement in the lemma holds for all smaller values of $t$. 

Suppose $T$ is a directed rooted tree on $V$ with $t$ edges and a root vertex $v^*$. Let $u^*\in V$ be a vertex such that $(u^*,v^*)\in T$. Then the edge set $T$ can be uniquely partitioned into three sets: a sub-tree $T_{1}$ rooted at $u^*$, a sub-tree $T_{2}$ rooted at $v^*$, and the singleton edge $(u^*,v^*)$. Let $|T_{1}|=t_{1}$ and $|T_{2}|=t_{2}$. Let $V_{1}$ and $V_{2}$ be the vertex sets of the sub-trees $T_{1}$ and $T_{2}$, respectively. Thus $V$ is the disjoint union of $V_{1}$ and $V_{2}$. We also denote \[m_{r}=\left\lceil Cn^{(k-r)/k}\right\rceil\qquad\text{ for each }r\in \{1,2,\dots,t\}.\]

\paragraph{Case 1: $t_{2}\geq 1$.} Let $V_{3}=V_{1}\cup\{v^*\}$ and $T_{3}=T_{1}\cup\{(u^*,v^*)\}$, so $T_{3}$ is a directed rooted tree on $V_{3}$ with root $v^*$. Since $1\leq |T_{3}|=t_{1}+1=t-t_{2}<t$, we can apply the induction hypothesis to $\pi_{V_{3}}[f]$ and obtain a $t_{1}\varepsilon$-pruning $f_{3}$ of $\pi_{V_{3}}[f]$ such that
\begin{equation}\label{eq:case1_1}
\bcS'(\pi_{v^*}f_{3},m_{t_{1}+1})\leq_{(1- t_{1}\delta,1)}\bcJ^{T_{3}}_{v^*}\big(\pi_{V_{3}}[f],m_{1}\big).
\end{equation}
There clearly exists an $t_{1}\varepsilon$-pruning $f'$ of $f$ such that $f_{3}=\pi_{V_{3}}[f']$. Since $1\leq |T_{2}|=t_{2}=t-t_{1}-1<t$, we can apply the induction hypothesis again to $\pi_{V_{2}}[f']$ and obtain a $(t_{2}-1)\varepsilon$-pruning $f_{2}$ of $\pi_{V_{2}}[f']$ such that
\begin{equation}\label{eq:case1_2}
\bcS'(\pi_{v^*}f_{2},m_{t_{2}})\leq_{(1-(t_{2}-1)\delta,1)}\bcJ^{T_{2}}_{v^*}\big(\pi_{V_{2}}[f'],m_{1}\big).
\end{equation}
There clearly exists a $(t_{2}-1)\varepsilon$-pruning $f''$ of $f'$ such that $f_{2}=\pi_{V_{2}}[f'']$. We then apply Corollary~\ref{cor:tree_bipartite_birthday_paradox}(1) to obtain an $\varepsilon$-pruning $f_{4}$ of $\pi_{v^*}f''$ such that
\begin{equation}\label{eq:case1_3}
\bcS'(f_{4},m_{t})\leq_{(1-\delta,1)}\bcP\big(\bcS'(f_{4},m_{t_{1}+1}),\bcS'(f_{4},m_{t_{2}})\big).
\end{equation}
Combining \eqref{eq:case1_1}, \eqref{eq:case1_2} and \eqref{eq:case1_3}, it follows that (using $f_{4}\leq \pi_{v^*}f''=\pi_{v^*}f_{2}\leq \pi_{v^*}f'=\pi_{v^*}f_{3}$)
\begin{equation}\label{eq:case1_4}
\bcS'(f_{4},m_{t})\leq_{(1-(t-1)\delta ,1)}\bcP\big(\bcJ^{T_{3}}_{v^*}\big(\pi_{V_{3}}[f],m_{1}\big),\bcJ^{T_{2}}_{v^*}\big(\pi_{V_{2}}[f'],m_{1}\big)\big).
\end{equation}
There clearly exists an $\varepsilon$-pruning $g$ of $f''$ such that $f_{4}=\pi_{v^*}g$. See Figure~\ref{fig:case1} for an illustration of the relations among the sub-probability mass functions $f,f',f''$ and $g$.

\begin{figure}[H]
\centering
\begin{tikzcd}
f \arrow[d, "\pi_{V_{3}}"] \arrow[r, dotted, "t_{1}\varepsilon\text{-pruning}"] &[2.5em]
f' \arrow[d, "\pi_{V_{3}}"] \arrow[dr, "\pi_{V_{2}}"] \arrow[rr, dotted, "(t_{2}-1)\varepsilon\text{-pruning}"]&
&[4em] f'' \arrow[d, "\pi_{V_{2}}"] \arrow[dr, "\pi_{v^*}"] \arrow[rr, dotted, "\varepsilon\text{-pruning}"] & &[2em] g  \arrow[d, "\pi_{v^*}"]\\
\pi_{V_3}[f] \arrow[r, "t_{1}\varepsilon\text{-pruning}"] &
f_3 &
\pi_{V_2}[f'] \arrow[r, "(t_{2}-1)\varepsilon\text{-pruning}"] & f_{2} & \pi_{v^*}f'' \arrow[r, "\varepsilon\text{-pruning}"] & f_{4}
\end{tikzcd}
\caption{Relations between functions in Case 1}
\label{fig:case1}
\end{figure}

Note that by definition we have
\begin{equation}\label{eq:case1_5}
\bcP\big(\bcJ^{T_{3}}_{v^*}\big(\pi_{V_{3}}[f],m_{1}\big),\bcJ^{T_{2}}_{v^*}\big(\pi_{V_{2}}[f],m_{1}\big)\big)=\bcJ^{T}_{v^*}(f,m_{1}).
\end{equation}
Combining \eqref{eq:case1_4} and \eqref{eq:case1_5}, it follows that (using $f'\leq f$)
\[
\bcS'(\pi_{v^*}g,m_{t})\leq_{(1-(t-1)\delta,1)}\bcJ^{T}_{v^*}(f,m_{1}).
\]
Since $g$ is a $(t-1)\varepsilon$-pruning of $f$, we conclude the proof in Case 1.

\paragraph{Case 2: $t_{2}=0$.} Since $1\leq |T_{1}|=t_{1}=t-1$, we can apply the induction hypothesis to $\pi_{V_{1}}[f]$ and obtain a $(t_{1}-1)\varepsilon$-pruning $f_{1}$ of $\pi_{V_{1}}[f]$ such that
\begin{equation}\label{eq:case2_1}
\bcS'(\pi_{u^*}f_{1},m_{t_{1}})\leq_{(1-(t_{1}-1)\delta,1)}\bcJ^{T_{1}}_{u^*}\big(\pi_{V_{1}}[f],m_{1}\big).
\end{equation}
There clearly exists a $(t_{1}-1)\varepsilon$-pruning $f^{\natural}$ of $f$ such that $f_{1}=\pi_{V_{1}}[f^{\natural}]$. We then apply Corollary~\ref{cor:tree_bipartite_birthday_paradox}(2) to obtain an $\varepsilon$-pruning $f_{0}$ of $\pi_{u^*,v^*}f^{\natural}$ such that
\begin{equation}\label{eq:case2_2}
\bcS'(\pi_{2}f_{0},m_{t})\leq_{(1-\delta,1)} \bcJ'\big(\bcS'(\pi_{1}f_{0},m_{t_{1}}),\bcS'(f_{0},m_{1})\big).
\end{equation}
Combining \eqref{eq:case2_1} and \eqref{eq:case2_2}, it follows that (using $\pi_{1}f_{0}\leq \pi_{u^*}f^{\natural}=\pi_{u^*}f_{1}$)
\begin{equation}\label{eq:case2_3}
\bcS'(\pi_{2}f_{0},m_{t})\leq_{(1-t_{1}\delta,1)} \bcJ'\big(\bcJ^{T_{1}}_{u^*}\big(\pi_{V_{1}}[f],m_{1}\big),\bcS'(f_{0},m_{1})\big).
\end{equation}
There clearly exists an $\varepsilon$-pruning $g$ of $f^{\natural}$ such that $f_{0}=\pi_{u^*,v^*}g$. See Figure~\ref{fig:case2} for an illustration of the relations among the sub-probability mass functions $f,f^{\natural}$ and $g$.

\begin{figure}[H]
\centering
\begin{tikzcd}
f \arrow[d, "\pi_{V_{1}}"] \arrow[r, dotted, "(t_{1}-1)\varepsilon\text{-pruning}"] &[4em]
f^{\natural} \arrow[d, "\pi_{V_{1}}"] \arrow[dr, "\pi_{u^*,v^*}"] \arrow[rr, dotted, "\varepsilon\text{-pruning}"]&
&[2em] g \arrow[d, "\pi_{u^*,v^*}"] \\
\pi_{V_1}[f] \arrow[r, "(t_{1}-1)\varepsilon\text{-pruning}"] &
f_1 &
\pi_{u^*,v^*}f^{\natural} \arrow[r, "\varepsilon\text{-pruning}"] & f_{0}
\end{tikzcd}
\caption{Relations between functions in Case 2}
\label{fig:case2}
\end{figure}

Note that by definition, we have
\begin{equation}\label{eq:case2_4}
\bcJ'\big(\bcJ^{T_{1}}_{u^*}\big(\pi_{V_{1}}[f],m_{1}\big),\bcS'(\pi_{u^*,v^*}f,m_{1})\big)=\bcJ^{T}_{v^*}(f,m_{1}).
\end{equation}
Combining \eqref{eq:case2_3} and \eqref{eq:case2_4}, it follows that (using $f_{0}=\pi_{u^*,v^*}g\leq \pi_{u^*,v^*}f$)
\[
\bcS'(\pi_{v^*}g,m_{t})\leq_{(1-(t-1)\delta,1)}\bcJ^{T}_{v^*}(f,m_{1}).
\]
Since $g$ is a $(t-1)\varepsilon$-pruning of $f$, we conclude the proof in Case 2.
\end{proof}

\subsection{Tree-Freeness and Cliques}\label{subsec:tree_proofs}

Lemma~\ref{lem:tree_birthday} provides the birthday-paradox tool that we need for proving the upper bounds on testing tree-freeness (Theorem~\ref{thm:tree_upper_bound}) and testing cliques (Theorem~\ref{thm:clique_upper_bound}).

In the proof of Theorem~\ref{thm:tree_upper_bound}, we will use the following notation (similar notations have been defined in Section~\ref{subsec:general_notations} and used in Section~\ref{sec:square_upper_bound}).

\begin{definition}
For a fixed positive integer $n$ and a fixed tree $H$ with $t$ edges, we define $\mathsf{Tree}_{H}(n)$ to be the collection of all $t$-edge subsets $E\subseteq \binom{[n]}{2}$ such that the graph $(V(E),E)$ is isomorphic to $H$, where $V(E)$ denotes the set of vertices incident to some edge in $E$. 
\end{definition}

\begin{theorem}\label{thm:tree_upper_bound}
Let $H$ be a fixed tree with $t$ edges, and let $\varepsilon\in (0,1)$ be a constant. Suppose $p\in [0,1]^{\binom{[n]}{2}}$ is a sub-probability vector that is $\varepsilon$-far from $H$-free. Then in $O(n^{(t-1)/t}/\varepsilon)$ independent samples from $p$, with probability at least $2/3$ there exists $t$ sampled edges forming a subgraph isomorphic to $H$.
\end{theorem}
\begin{proof}
We consider a $t$-uniform hypergraph whose vertex set is $\binom{[n]}{2}$ and whose edge set is the collection of all $t$-edge subsets $E\subseteq \binom{[n]}{2}$ such that the subgraph formed by $E$ is isomorphic to $H$. Since $p$ is $\varepsilon$-far from $H$-free, we can apply Lemma~\ref{lem:LP_weak_duality} to this hypergraph and obtain a sub-probability vector $\lambda=(\lambda_{E})$, where $E$ ranges in the collection $\mathsf{Tree}_{H}(n)$, that satisfies the three conditions listed in Lemma~\ref{lem:LP_weak_duality}.\footnote{We only need the second and third conditions for this proof.}

Let $V$ be the vertex set of $H$. For each $E\in \mathsf{Tree}_{H}(n)$, let $V(E)\subseteq [n]$ denote the set of vertices incident to $E$, and choose an isomorphism map $\psi_{E}:V(E)\rightarrow V$ from the graph $(V(E),E)$ to $H$. Then define a vector $y^{(E)}\in [n]^{V}$ by letting
\[
y^{(E)}_{v}=\psi_{E}^{-1}(v)\in [n]\qquad\text{for all }v\in V.
\]
It is clear that for any $y\in [n]^{V}$, there is at most one $E\in \mathsf{Tree}_{H}(n)$ such that $y=y^{(E)}$. We now define a sub-probability mass function $f:[n]^{V}\rightarrow [0,1]$ by\footnote{Note that $f$ is a sub-probability mass function because $\sum_{y\in [n]^{V}}f(y)=\sum_{E\in \mathsf{Tree}_{H}(n)}\lambda_{E}\leq 1$.}
\[
f(y)=\begin{cases}
\lambda_{E},&\text{if }y=y^{(E)}\text{ for some }E\in \mathsf{Tree}_{H}(n),\\
0, &\text{otherwise}.
\end{cases}
\]
We thus have
\[
\sum_{y\in [n]^{V}}f(y)=\sum_{E\in\mathsf{Tree}_{H}(n)}\lambda_{E}\geq \frac{\varepsilon
}{t},
\]
where in the last transition we used the third condition of the conclusion of Lemma~\ref{lem:LP_weak_duality}. Furthermore, for any edge $\{a,b\}\in \binom{[n]}{2}$ and any edge $(u,v)\in T$, we have
\begin{equation}\label{eq:tree_f_vs_p}
\pi_{u,v}f(a,b)=\sum_{y\in [n]^{V}}\ind{y_{u}=a\text{ and }y_{v}=b}\cdot f(y)\leq \sum_{E\in \mathsf{Tree}_{H}(n)}\ind{\{a,b\}\in E}\cdot \lambda_{E}\leq p_{ab},
\end{equation}
where in the last transition we used the second condition of the conclusion of Lemma~\ref{lem:LP_weak_duality}.

Now we pick an arbitrary vertex $v^*\in V$ and let $T$ be a directed rooted tree (with root $v^*$) on $V$ such that the edges of $T$ (when viewed as undirected edges) coincide with the edges of $H$. Given a positive integer $m$, let $\frakP_{1}(m)$ be the following process:
\begin{enumerate}
\item For each pair $(u,v)\in T$, independently draw $m$ samples from $\pi_{u,v}f$, and let $X^{(u,v)}\subseteq [n]^{2}$ be the set formed by the $m$ samples.
\item Output 1 if there exists a map $\varphi:T\rightarrow [n]^{2}$ such that $\varphi$ is compatible with some vector $y\in [n]^{V}$, and $\varphi(u,v)\in X^{(u,v)}$ for each $(u,v)\in T$. Otherwise, output 0.
\end{enumerate}
By Lemma~\ref{lem:tree_birthday}, if $n$ is sufficiently large and 
\begin{equation}\label{eq:tree_choice_C_m}
C=\frac{288t^{4}}{\varepsilon},\qquad m=\left\lceil Cn^{(t-1)/t}\right\rceil,
\end{equation}
there exists an $(\varepsilon/(2t))$-pruning $g$ of $f$ such that
\[
\bcS'(\pi_{v^*}g,\lceil C\rceil)\leq_{(5/6,1)}\bcJ^{T}_{v^*}(f,m).
\]
By the definition of $\bcJ^{T}_{v^*}(f,m)$ (Definition~\ref{def:J_tree}), it follows that
\begin{equation}\label{eq:tree_frakP_1}
\Pr{\frakP_{1}(m)\text{ outputs }1\big.}\geq \Pru{w\sim \bcJ^{T}_{v^*}(f,m)}{w\neq \vec{0}\big.} \geq \Pru{w\sim \bcS'(\pi_{v^*}g, \lceil C\rceil)}{w\neq \vec{0}\big.}-\frac{1}{6}\geq \frac{2}{3},
\end{equation}
where we used the fact that $\sum_{a=1}^{n}\pi_{v^*}g(a)\geq -\varepsilon/(2t)+\sum_{y\in [n]^{V}}f(y)\geq \varepsilon/(2t)$ in the last transition.

Now consider the following process denoted by $\frakP_{2}(m)$:
\begin{enumerate}
\item For each pair $(u,v)\in T$, independently draw $m$ samples from the sub-probability vector $p$, and let $Y^{(u,v)}\subseteq \binom{[n]}{2}$ be the set formed by the $m$ samples.
\item Output 1 if there exists a map $\varphi:T\rightarrow \binom{[n]}{2}$ such that
\[
\bigl\{\varphi(u,v)\,\big|\,(u,v)\in T\bigr\}\in \mathsf{Tree}_{H}(n).
\]
Otherwise, output 0.
\end{enumerate}
Due to \eqref{eq:tree_f_vs_p}, there is an obvious coupling between the processes $\frakP_{1}(m)$ and $\frakP_{2}(m)$ under which the output of the latter process is always at least the output of the former. By \eqref{eq:tree_frakP_1}, this means that $\frakP_{2}(m)$ outputs 1 with probability at least $2/3$ if $m$ is chosen as in \eqref{eq:tree_choice_C_m}. On the other hand, note that $\frakP_{2}(m)$ takes a total number of $tm$ independent samples from $p$, and whenever it outputs $1$, there are $t$ edges among the $tm$ samples that form a subgraph isomorphic to $H$. Therefore, we conclude that when $n$ is sufficiently large, in
\[
tm=t\cdot\left\lceil \frac{288 t^{4}}{\varepsilon
}n^{(t-1)/t}\right\rceil=O(n^{(t-1)/t}/\varepsilon)
\]
samples from $p$, with probability at least $2/3$ there are $t$ sampled edges forming a subgraph isomorphic to $H$.
\end{proof}

\begin{corollary}\label{cor:tree_upper_bound}
For any fixed tree $H$ with $t$ edges, we have $\sam\big(\cG^{H\textup{-free}}_{n}\big)\leq O(n^{(t-1)/t}/\varepsilon)$.
\end{corollary}
\begin{proof}
Theorem~\ref{thm:tree_upper_bound} implies Corollary~\ref{cor:tree_upper_bound} in the same way as Theorem~\ref{thm:square_freeness_upper} implies Corollary~\ref{cor:square_freeness_upper}. We refer to the proof of Corollary~\ref{cor:square_freeness_upper} for an outline of the argument.
\end{proof}

Perhaps somewhat surprisingly, the proof of the upper bound for testing cliques follows the same route as the proof of Theorem~\ref{thm:tree_upper_bound}. The reason is that in any violation hypergraph against the property $\cG^{\textup{cliq}}_{n}$ (see Definition~\ref{def:violation_hypergraph} for the definition of violation hypergraphs), all hyperedges correspond to length-3 paths in the $n$-vertex complete graph (in particular, violation hypergraphs against $\cG^{\textup{cliq}}_{n}$ are always 3-uniform). Since the length-3 path is a tree, the birthday-paradox tools (specifically, Lemma~\ref{lem:tree_birthday}) we have developed for analyzing the tree-freeness tester are also well-suited for analyzing the clique tester.

\begin{theorem}\label{thm:clique_upper_bound}
We have $\sam\big(\cG^{\textup{cliq}}_{n},\varepsilon\big)\leq O(n^{2/3}/\varepsilon)$.
\end{theorem}

\begin{proof}
It is easy to see that for any $E\subseteq \binom{[n]}{2}$, the minimal $E$-violations (recall Definition~\ref{def:violation_hypergraph}) of $\cG^{\textup{cliq}}_{n}$ are exactly the three-edge sets
\[
\bigl\{\{a,b\},\{b,c\},\{c,d\}\bigr\}\subseteq \binom{[n]}{2}
\]
such that $\{a,b\},\{c,d\}\in E$ and $\{b,c\}\notin E$.\footnote{Note that here $a$ and $d$ are not necessarily distinct.} We refer to such three-edge sets as \emph{$E$-alternating paths}.

By the discussion in Section~\ref{subsec:general_framework}, it suffices to show the following for any fixed $E\subseteq \binom{[n]}{2}$: if $\mu$ is a distribution over $\binom{[n]}{2}$ such that
\[
\mu(E\triangle E')\geq \varepsilon\qquad\text{for any }E'\in \cG^{\textup{cliq}}_{n},
\]
then in $O(n^{2/3}/\varepsilon)$ independent samples from $\mu$, with probability at least $2/3$ there are three sampled edges forming an $E$-alternating path.

We apply Lemma~\ref{lem:LP_weak_duality} to the violation hypergraph of $E$ against $\cG^{\textup{cliq}}_{n}$. This yields a sub-probability vector $\lambda=(\lambda_{P})$, where $P$ ranges over all $E$-alternating paths, that satisfies the three conditions in Lemma~\ref{lem:LP_weak_duality}.\footnote{As in the proof of Theorem~\ref{thm:tree_upper_bound}, we only need the second and third conditions.}

For each $E$-alternating path $P=\big\{\{a,b\},\{b,c\},\{c,d\}\big\}$, define a vector $y^{(P)}\in [n]^{4}$ by letting\footnote{Here one can order the four vertices either as $a,b,c,d$ or as $d,c,b,a$.}
\[
y^{(P)}_{1}=a,\qquad y^{(P)}_{2}=b,\qquad y^{(P)}_{3}=c,\qquad\text{and}\qquad y^{(P)}_{4}=d.
\]
We define a sub-probability mass function $f:[n]^{4}\rightarrow [0,1]$ by
\[
f(y)=\begin{cases}
\lambda_{P},&\text{if }y=y^{(P)}\text{ for some }E\text{-alternating path }P,\\
0,&\text{otherwise}.
\end{cases}
\]
We thus have
\[
\sum_{y\in [n]^{4}}f(y)=\sum_{E\text{-alternating paths }P}\lambda_{P}\geq \frac{\varepsilon
}{3},
\]
where in the last transition we used the third condition of the conclusion of Lemma~\ref{lem:LP_weak_duality}. Furthermore, for any edge $\{a,b\}\in \binom{[n]}{2}$ and any $j\in \{1,2,3\}$, we have
\begin{align*}
\pi_{j,j+1}f(a,b)&=\sum_{y\in [n]^{4}}\ind{y_{j}=a\text{ and }y_{j+1}=b}\cdot f(y)\\
&\leq \sum_{E\text{-alternating paths }P}\ind{\{a,b\}\in P}\cdot \lambda_{P}\\
&\leq \mu(\{a,b\}),
\end{align*}
where in the last transition we used the second condition of the conclusion of Lemma~\ref{lem:LP_weak_duality}.

The rest of the proof is entirely analogous to the proof of Theorem~\ref{thm:tree_upper_bound} and is thus omitted.\footnote{The main idea is to apply Lemma~\ref{lem:tree_birthday} to the directed rooted tree $T=\{(1,2),(2,3),(3,4)\}$ with root $4$.}
\end{proof}

\section{Lower Bounds for Subgraph-Freeness}

In this section, we prove the sample complexity lower bounds for testing triangle-freeness, square-freeness and tree-freeness, stated in \eqref{eq:triangle_result}, \eqref{eq:square_result} and Theorem~\ref{thm:tree-freeness}, respectively. As is the case with upper bounds (see Section~\ref{sec:tree_upper_bound}), we will also prove the lower bound for testing cliques (stated in Theorem~\ref{thm:clique}) in Section~\ref{subsec:tree_lower_bound}, along with the lower bound for tree-freeness, because their proofs are similar to each other.

\subsection{Triangle-Freeness Constructions}\label{subsec:triangle_lower_bound}

As discussed in Section~\ref{subsubsec:sparse_removal}, the lower bound for testing triangle-freeness is proved by combining the Ruzsa-Szemer\'edi construction (Proposition~\ref{prop:ruzsa_semeredi}) with a standard technique that lifts lower bounds for one-sided-error tester to two-sided-error tester. 

The technique is reminiscent of that used in Section~\ref{subsec:bipartiteness_lower_bound}. 
Given an edge set $E \subseteq \binom{[n]}{2}$, we consider the two-fold blow-up of the graph $([n],E)$, 
in which each vertex $a \in [n]$ is replaced by a pair of copies. 
For any two such pairs corresponding to vertices $a,b \in [n]$ with $\{a,b\} \in E$, 
the blow-up graph contains all four possible edges between the two pairs.

The key idea is to retain exactly two of these four edges for each $\{a,b\}\in E$. 
The structure of the resulting graph can then vary in an interesting way, 
depending on how the two edges are selected in each case. 
We formalize this operation in the following definition.

\begin{definition}
For any edge set $E\subseteq \binom{[n]}{2}$ and any vector $y\in\bF_2^E$, we define an edge set
\[
R_y(E)
=
\Bigl\{
\bigl\{(a,t),(b,y_{ab}+t)\bigr\}\,\Big|\,\{a,b\}\in E\text{ and }t\in \bF_{2}
\Bigr\}\subseteq \binom{[n]\times \bF_{2}}{2}.
\]
over the vertex set $[n]\times \bF_{2}$.
\end{definition}

Note that the vector $y\in \bF_{2}^{E}$ specifies for each $\{a,b\}\in E$ how two of the four edges between the $a$-copies $(a,0),(a,1)$ and the $b$-copies $(b,0),(b,1)$ are selected. The main observation is that if every edge in $E$ is contained in exactly one triangle, then we can easily make $R_{y}(E)$ either triangle-free or far-from triangle-free, by picking suitable vectors $y$ for each case.

\begin{definition}
Suppose $E\subseteq \binom{[n]}{2}$ is an edge set such that every edge in $E$
is contained in exactly one triangle. We define two collections of vectors $Y^{\yes}_{\triangle}(E)$ and $Y^{\no}_{\triangle}(E)$ by
\begin{align*}
Y^{\yes}_{\triangle}(E)&=\Big\{
y\in \bF_2^{E}\,\Big|\, y_{ab}+y_{bc}+y_{ca}=1\text{ for all triangles }
\{\{a,b\},\{b,c\},\{c,a\}\}\subseteq E
\Big\},\text{ and}\\
Y^{\no}_{\triangle}(E)&=\Big\{
y\in \bF_2^{E}\,\Big|\, y_{ab}+y_{bc}+y_{ca}=0\text{ for all triangles }
\{\{a,b\},\{b,c\},\{c,a\}\}\subseteq E
\Big\},
\end{align*}
\end{definition}

\begin{proposition}\label{prop:property_of_Yyes_Yno}
Suppose $E\subseteq \binom{[n]}{2}$ is an edge set such that every edge in $E$
is contained in exactly one triangle. We have
\begin{enumerate}[label=(\arabic*)]
\item For any $y\in Y^{\yes}_{\triangle}(E)$, the edge set $R_{y}(E)$ is triangle-free. 
\item For any $y\in Y^{\no}_{\triangle}(E)$, the edge set $R_{y}(E)$ is the edge-disjoint union of $2|E|/3$ triangles. Consequently, we have $|R_{y}(E)\setminus E'|\geq |R_{y}(E)|/3$ for any triangle-free edge set $E'\subseteq \binom{[n]\times \bF_{2}}{2}$.
\end{enumerate}
\end{proposition}

\begin{proof}
The second statement is obvious. For the first statement, it suffices to note that for any $y\in \bF_{2}^{E}$, any triangle in $R_{y}(E)$ must ``project'' to a triangle in $E$ under the canonical projection map from the vertex set $[n]\times \bF_{2}$ to the vertex set $[n]$.
\end{proof}

We next show that when $y$ is randomized in either $Y^{\yes}_{\triangle}(E)$ or $Y^{\no}_{\triangle}(E)$, it is impossible to distinguish between the two cases if one is only given $o(|E|^{2/3})$ edge samples from $R_{y}(E)$.

\begin{lemma}\label{lem:triangle_indistinguishable}
Fix an edge set $E\subseteq \binom{[n]}{2}$ such that every edge in $E$
is contained in exactly one triangle. Suppose there is a randomized map
\(
\cA:\binom{[n]\times\bF_{2}}{2}^{m}\rightarrow \{0,1\} 
\)
that satisfies the following.
\begin{enumerate}[label=(\arabic*)]
\item For a uniformly random $y\in Y^{\yes}_{\triangle}(E)$ and independent edge samples $e_{1},\dots,e_{m}\in R_{y}(E)$, we have $\Pr{\cA(e_{1},\dots,e_{m})=1}\geq 2/3$. 
\item For a uniformly random $y\in Y^{\no}_{\triangle}(E)$ and independent edge samples $e_{1},\dots,e_{m}\in R_{y}(E)$, we have $\Pr{\cA(e_{1},\dots,e_{m})=0}\geq 2/3$. 
\end{enumerate}
Then we must have $m\geq |E|^{2/3}$.
\end{lemma}
\begin{proof}
In the two assumptions on $\cA$ stated in the lemma, the input $(e_{1},\dots,e_{m})$ to $\cA$ follows two different distributions. It suffices to show that these two distributions over $\binom{[n]\times \bF_{2}}{2}^{m}$, which we denote by $\cD^{\yes}$ and $\cD^{\no}$, respectively, have total variation distance less than $1/3$ if $m<|E|^{2/3}$. Both $\cD^{\yes}$ and $\cD^{\no}$ can be alternatively generated by first sampling edges $\{u_{1},v_{1}\},\dots,\{u_{m},v_{m}\}$ uniformly at random from $E$ and then letting 
\[e_{i}=\big\{(u_{i},t_{i}),(v_{i},s_{i})\big\}\text{ for some suitably chosen }s_{i},t_{i}\in \bF_{2}\]
for all $i\in [m]$. Note that the first step (choosing $u_{i}$'s and $v_{i}$'s) is identical for $\cD^{\yes}$ and $\cD^{\no}$, while the second step may be implemented differently for the two. Furthermore, if the collection $\big\{\{u_{1},v_{1}\},\dots,\{u_{m},v_{m}\}\big\}$ sampled in the first step does not contain a triangle, the second step is also identical for $\cD^{\yes}$ and $\cD^{\no}$. Since $\big\{\{u_{1},v_{1}\},\dots,\{u_{m},v_{m}\}\big\}$ contains a triangle with probability at most (by union bound)
\[
\frac{|E|}{3}\cdot \frac{m^{3}}{|E|^{3}}=\frac{1}{3}m^{3}|E|^{-2},
\]
we have $\|\cD^{\yes}-\cD^{\no}\|_{\mathrm{TV}}\leq \frac{1}{3}m^{3}|E|^{-2}<\frac{1}{3}$ if $m<|E|^{2/3}$.
\end{proof}

\begin{corollary}
We have $\sam\big(\cG^{\textup{tri}}_{2n},1/3\big)\geq n^{4/3}\exp\bigl(-O\bigl(\sqrt{\log n}\bigr)\bigr)$.
\end{corollary}
\begin{proof}
We use Proposition~\ref{prop:ruzsa_semeredi} to obtain an edge set $E\subseteq \binom{[n]}{2}$ in which every edge is contained in exactly one triangle, and $|E|=\ex^{=1}(n,C_{3})=n^{2} \exp\bigl(-O\bigl(\sqrt{\log n}\bigr)\bigr)$. For any $y\in Y^{\yes}_{\triangle}(E)$, the graph $R_{y}(E)$ is triangle-free by Proposition~\ref{prop:property_of_Yyes_Yno}(1). On the other hand, it follows from Proposition~\ref{prop:property_of_Yyes_Yno}(2) that if we let $\mu_{y}$ denote the uniform distribution over $R_{y}(E)$ (considered as an edge set over $[2n]$), then
\[
\mu_{y}\big(R_{y}(E)\triangle E'\big)\geq \frac{1}{3}\qquad\text{for any }y\in Y^{\no}_{\triangle}(E)\text{ and any }E'\in\cG^{\textup{tri}}_{2n}. 
\]
Therefore, any sample-based distribution-free tester for $\cG^{\textup{tri}}_{2n}$ with proximity parameter $\varepsilon=1/3$ and sample complexity $m$, when considered as a randomized map $\cA:\binom{[n]\times \bF_{2}}{2}^{m}\rightarrow \{0,1\}$, must satisfy the conditions of Lemma~\ref{lem:triangle_indistinguishable} and hence $m\geq |E|^{2/3}=n^{4/3}\exp\bigl(-O\bigl(\sqrt{\log n}\bigr)\bigr)$.
\end{proof}

\subsection{Square-Freeness Constructions}\label{subsec:square_lower_bound}

As in Section~\ref{subsec:triangle_lower_bound}, it suffices to prove for any positive integer $n$ that
\begin{equation}\label{eq:square_lower_bound_suffices}
\sam\big(\cG^{\textup{squ}}_{2n},1/4\big)\geq \bigl(\ex^{=1}(n,C_{4})\bigr)^{3/4}.
\end{equation}
The desired lower bound \[\sam\big(\cG^{\textup{squ}}_{2n},1/4\big)\geq n^{9/8}\exp\left(-O\left(\sqrt{\log n}\right)\right)\]
then follows by plugging Proposition~\ref{prop:ruzas_square} into \eqref{eq:square_lower_bound_suffices}. The proof of \eqref{eq:square_lower_bound_suffices} is essentially the same as the corresponding proof for triangle-freeness in Section~\ref{subsec:triangle_lower_bound}. In particular, for any edge set $E\subseteq \binom{[n]}{2}$ in which every edge is contained in exactly one square, we can define two collections of vectors
\[
Y^{\yes}_{\square}(E),Y^{\no}_{\square}(E)\subseteq \bF_{2}^{E}
\]
by requiring their members $y$ to satisfy $y_{ab}+y_{bc}+y_{cd}+y_{da}=1$ (respectively, $=0$) for all squares $\{\{a,b\},\{b,c\},\{c,d\},\{d,a\}\}\subseteq E$. The important observation is that for any $y\in \bF_{2}^{E}$ and any edge set $E\subseteq \binom{[n]}{2}$, any square in $R_{y}(E)$ must ``project'' to a square in $E$ under the canonical projection map $[n]\times \bF_{2}\rightarrow [n]$.\footnote{Note that this argument would fail if we were considering the property $C_{6}$-freeness, because a 6-cycle in $R_{y}(E)$ does not necessarily project to a 6-cycle in $E$ (there may be repeated vertices after the projection).} The rest of the argument is entirely analogous to Section~\ref{subsec:triangle_lower_bound}, and thus we omit the proof of \eqref{eq:square_lower_bound_suffices}.

In the rest of this subsection, we sketch the proof of Proposition~\ref{prop:ruzas_square} that is implicit in the paper by Timmons and Verstra\"ete \cite{timmons2015counterexample}.

As is the case with the proof of Proposition~\ref{prop:ruzsa_semeredi} by \cite{ruzsa1978triple}, the construction of graphs in which every edge is contained in exactly one square relies on additive combinatorics. While the Ruzsa-Szemer\'edi construction for $\ex^{=1}(n,C_{3})$ is based on integer sets without 3-term arithmetic progressions, Timmons and Verstra\"ete \cite{timmons2015counterexample} observed that one can similarly obtain constructions for $\ex^{=1}(n,C_{4})$ using certain integer sets known as \emph{$k$-fold Sidon sets}, which were first defined by Lazebnik and Verstra\"ete \cite{lazebnik2003hypergraphs}.

\begin{definition}\label{def:trivial_solution}
Let $c_{1},\dots,c_{r}$ be nonzero integers such that $\sum_{i=1}^{r}c_{i}=0$. Given an Abelian group $\Gamma$, a solution $(a_{1},\dots,a_{r})\in \Gamma^{r}$ to the equation
\[
c_{1}x_{1}+\dots+c_{r}x_{r}=0
\]
is called a \emph{trivial solution} if there exists a partition of $[r]$ into nonempty sets $T_{1},\dots,T_{m}$ such that for every $i\in [m]$, we have $\sum_{j\in T_{i}}c_{j}=0$ and $a_{j_{1}}=a_{j_{2}}$ whenever $j_{1},j_{2}\in T_{i}$.
\end{definition}

\begin{definition}[\cite{lazebnik2003hypergraphs}]
Let $k$ be a positive integer and let $\Gamma$ be an Abelian group. A subset $A\subseteq \Gamma$ is called a \emph{$k$-fold Sidon set} if any solution $(a_{1},\dots,a_{4})\in A^{4}$ to any equation of the form
\[
c_{1}x_{1}+c_{2}x_{2}+c_{3}x_{3}+c_{4}x_{4}=0,
\]
where $c_{1},\dots,c_{4}$ are integers such that $|c_{i}|\leq k$ for all $i\in [4]$ and $c_{1}+c_{2}+c_{3}+c_{4}=0$, must be trivial.
\end{definition}

\begin{proposition}[{\cite[Theorem 7.1]{timmons2015counterexample}}]\label{prop:3fold_Sidon_suffices}
Suppose $n$ is a positive integer not divisible by $2$ or $3$, and $\Gamma$ is an Abelian group of order $n$. If $A\subseteq \Gamma$ is a $3$-fold Sidon set, we have $\ex^{=1}(4n,C_{4})\geq 4n|A|$.
\end{proposition}
\begin{proof}
We construct a graph with vertex set $\Gamma\times [4]$ where two vertices $(x,i),(y,j)\in \Gamma\times [4]$ are adjacent if and only if $\{i,j\}\in \{\{1,3\},\{1,4\},\{2,3\},\{2,4\}\}$ and
\[
y-x=(j-i)a\qquad\text{for some }a\in A.
\]
The number of edges in this graph is $4n|A|$. Furthermore, using the condition that $A$ is a 3-fold Sidon set, it is easy to see that every edge in this graph is contained in exactly one square.
\end{proof}

In light of Proposition~\ref{prop:3fold_Sidon_suffices} and the prime number theorem for arithmetic progressions, to prove Proposition~\ref{prop:ruzas_square} it suffices to prove the following lemma:

\begin{lemma}
Suppose $p$ is a prime number such that $p\equiv \pm 5\pmod{12}$. Then there is a $3$-fold Sidon set $A\subseteq \bF_{p}^{2}$ (here $\bF_{p}^{2}$ is an Abelian group under addition) of cardinality at least $p\cdot \exp\bigl(-O\bigl(\sqrt{\log p}\bigr)\bigr)$.
\end{lemma}
\begin{proof}[Proof Sketch.]
As pointed out in \cite{cilleruelo2014k}, this can be proved by adapting Ruzsa's proof of \cite[Theorem 7.3]{ruzsa1993solving}. For each $a\in \bF_{p}$, let $f(a)=(a,a^{2})\in \bF_{p}^{2}$. For any nonzero integers $c_{1},c_{2},c_{3},c_{4}\in [-3,3]$ such that $c_{1}+c_{2}+c_{3}+c_{4}=0$, consider solutions $(a_{1},a_{2},a_{3},a_{4})\in \bF_{p}^{4}$ to the equation
\begin{equation}\label{eq:quadratic_Sidon}
c_{1}f(x_{1})+c_{2}f(x_{2})+c_{3}f(x_{3})+c_{4}f(x_{4})=0.
\end{equation}
A solution $(a_{1},a_{2},a_{3},a_{4})$ to \eqref{eq:quadratic_Sidon} is said to be a trivial solution if \((f(a_{1}),f(a_{2}),f(a_{3}),f(a_{4}))\)
is a trivial solution to the linear equation $c_{1}x_{1}+c_{2}x_{2}+c_{3}x_{3}+c_{4}x_{4}=0$ (as per Definition~\ref{def:trivial_solution}). It now suffices to find a set $A\subseteq \bF_{p}$ of cardinality at least $p\cdot \exp\bigl(-O\bigl(\sqrt{\log p}\bigr)\bigr)$ such that for any equation of the form \eqref{eq:quadratic_Sidon} has only trivial solutions in $A$.

For each individual equation of the form
\begin{alignat}{2}
x_{1}+x_{2}+x_{3}&=3x_{4},\qquad &&\text{or}\label{eq:linear_Sidon_1}\\
d_{1}x_{1}+d_{2}x_{2}&=(d_{1}+d_{2})x_{3},\qquad&&\text{where }d_{1},d_{2}\in \{1,2,\dots,20\},\label{eq:linear_Sidon_2}
\end{alignat}
by Behrend's construction \cite{behrend1946sets} there is a set $B\subseteq \bF_{p}$ of cardinality at least $p\cdot \exp\bigl(-O\bigl(\sqrt{\log p}\bigr)\bigr)$ in which it has no nontrivial solutions. By taking random translations of all these individual sets $B$ and intersecting them, one gets a (random) set $A\subseteq \bF_{p}$ with (expected) size at least $p\cdot \exp\bigl(-O\bigl(\sqrt{\log p}\bigr)\bigr)$ in which no equation of the form \eqref{eq:linear_Sidon_1} or \eqref{eq:linear_Sidon_2} has nontrivial solutions. We claim that in such sets $A$, equations of the form \eqref{eq:quadratic_Sidon} also have no nontrivial solutions.

Case 1: if one of $c_{1},c_{2},c_{3},c_{4}$ has a different sign from the other three, then since $c_{1},c_{2},c_{3},c_{4}$ are integers in the range $[-3,3]$, the equation $c_{1}x_{1}+c_{2}x_{2}+c_{3}x_{3}+c_{4}x_{4}=0$ can only be of the form \eqref{eq:linear_Sidon_1}, which has no nontrivial solutions in $A$.

Case 2: if two of $c_{1},c_{2},c_{3},c_{4}$ are positive and the other two are negative, without loss of generality assume $c_{1},c_{2}>0$ and $c_{3},c_{4}<0$. Using the condition that no equation of the form \eqref{eq:linear_Sidon_2} has nontrivial solutions in $A$, it is easy to see that for any nontrivial solution $(a_{1},a_{2},a_{3},a_{4})$ to \eqref{eq:quadratic_Sidon}, the elements $a_{1},a_{2},a_{3},a_{4}$ must be pairwise distinct. Furthermore, we have
\begin{align*}
c_{1}c_{2}(a_{1}-a_{2})^{2}&=(c_{1}a_{1}^{2}+c_{2}a_{2}^{2})(c_{1}+c_{2})-(c_{1}a_{1}+c_{2}a_{2})^{2}\\
&=(c_{3}a_{3}^{2}+c_{4}a_{4}^{2})(c_{3}+c_{4})-(c_{3}a_{3}+c_{4}a_{4})^{2}=c_{3}c_{4}(a_{3}-a_{4})^{2}.
\end{align*}
This implies $c_{1}c_{2}c_{3}c_{4}$ must be a quadratic residue modulo $p$. Since $3$ is not a quadratic residue modulo $p$ (due to the condition $p\equiv \pm 5\pmod{12}$) and since $c_{1}c_{2}c_{3}c_{4}\in \{1,4,9,12,16,36,81\}$, it must be the case that $c_{1}c_{2}c_{3}c_{4}$ is a perfect square. Thus the quadratic equation $c_{1}c_{2}(a_{1}-a_{2})^{2}=c_{3}c_{4}(a_{3}-a_{4})^{2}$ in variables $a_{1},a_{2},a_{3},a_{4}$ can be factorized into two linear equations. Combining either of the two linear equations with the condition that $c_{1}a_{1}+c_{2}a_{2}+c_{3}a_{3}+c_{4}a_{4}=0$, one obtains a linear equation in the variables $a_{1},a_{2},a_{3}$. This three-variable equation either reduces to a two-variable equation, which would force two of $a_{1},a_{2},a_{3}$ to be equal, or has the form \eqref{eq:linear_Sidon_2}. We thus reach the conclusion that \eqref{eq:quadratic_Sidon} has no nontrivial solutions in $A$.
\end{proof}

\subsection{Tree-Freeness Constructions}\label{subsec:tree_lower_bound}

In this subsection, we prove the lower bound part of Theorems~\ref{thm:clique} and~\ref{thm:tree-freeness}. We first prove the lower bound for testing tree-freeness.

\begin{theorem}\label{thm:tree_lower_bound}
Let $H$ be a fixed tree with $t$ edges. Then there exists a constant $\varepsilon\in (0,1)$ such that $\sam\big(\cG^{H\textup{-free}}_{n},\varepsilon\big)\geq \Omega(n^{(t-1)/t})$.
\end{theorem}

\begin{proof}
The case $t=1$ is easy; we assume $t\geq 2$ in the following. 

\paragraph{Construction.} Suppose $H=(V,T)$ is a tree with $|T|=t$. We build two graphs $H^{(0)}$ and $H^{(1)}$ as follows:
\begin{enumerate}
\item Initialize $H^{(0)},H^{(1)}$ to be empty graphs (with empty vertex sets). 
\item For each subset $T'\subseteq T$, do the following:
\begin{itemize}
\item If $|T|-|T'|$ is even, add a copy of the graph $(V,T')$ to $H^{(0)}$ (so that $H^{(0)}$ gets $|V|=t+1$ new vertices and $|T'|$ new edges).
\item If $|T|-|T'|$ is odd, add a copy of the graph $(V,T')$ to $H^{(1)}$ (so that $H^{(1)}$ gets $|V|=t+1$ new vertices and $|T'|$ new edges).
\end{itemize}
\end{enumerate}
Since there are exactly $2^{t-1}$ subsets of $T$ with odd (or even) cardinality, both $H^{(0)}$ and $H^{(1)}$ have $2^{t-1}(t+1)$ vertices. We denote $r=2^{t-1}(t+1)$.

For each $j\in \{0,1\}$ and positive integer $n$, let $\cH^{(j)}_{n}$ be the output distribution of the following process:
\begin{enumerate}
\item Initialize $G$ to be a graph with the vertex set $[rn]$ and an empty edge set.
\item For each $i\in [n]$, do the following:
\begin{itemize}
\item Pick a random bijection $\varphi$ from the set $\{(i-1)r+1,\dots,ir\}$ to the vertex set of $H^{(j)}$.
\item For each edge $\{u,v\}$ in $H^{(j)}$, add to $G$ an edge between $\varphi^{-1}(u)$ and $\varphi^{-1}(v)$.
\end{itemize}
\item Output $G$.
\end{enumerate}
In words, a random graph $G\sim \cH^{(j)}_{n}$ is the vertex-disjoint union of $n$ copies of $H^{(j)}$, with the vertices of each copy randomly permuted.

Finally, for each $j\in \{0,1\}$ and positive integers $n,m$, let $\cD^{(j)}_{n,m}$ be the output distribution of the following process:
\begin{enumerate}
\item Sample a graph $G\sim \cH^{(j)}_{n}$.
\item Sample $m$ edges $e_{1},\dots,e_{m}$ independently and uniformly from the edge set of $G$.
\item Output the sequence $(e_{1},\dots,e_{m})$.
\end{enumerate}
A sequence $(e_{1},\dots,e_{m})\in \binom{[rn]}{2}^{m}$ sampled from $\cD_{n,m}^{(j)}$ is said to be \emph{well-behaved} if for each $i\in [n]$, there are at most $(t-1)$ indices $k\in [m]$ such that both endpoints of $e_{k}$ fall in $\{(i-1)r+1,\dots,ir\}$. In other words, the edge sequence $(e_{1},\dots,e_{m})$ is well-behaved if no $t$ edges come from the same copy of $H^{(j)}$.

\paragraph{Analysis.} For any edge $e\in T$, there are exactly $2^{t-2}$ copies of $e$ in both $H^{(0)}$ and $H^{(1)}$. Thus both $H^{(0)}$ and $H^{(1)}$ have $2^{t-2}t$ edges. The main observation is that, for any edge $e\in T$, if we remove all copies of $e$ from $H^{(0)}$ and $H^{(1)}$, the two graphs become isomorphic. From this observation, it is easy to see that the distributions $\cD^{(0)}_{1,m}$ and $\cD^{(1)}_{1,m}$ are identical if $m\leq t-1$. Consequently, for any positive integers $n$ and $m$, a random \emph{well-behaved} sample from $\cD^{(0)}_{n,m}$ is indistinguishable from a random well-behaved sample from $\cD^{(1)}_{n,m}$. For each $j\in \{0,1\}$, a random sample $(e_{1},\dots,e_{m})\sim \cD^{(j)}_{n,m}$ is well-behaved with probability at least (using union bound)
\[
1-n\cdot \frac{m^{t}}{n^{t}}> \frac{2}{3}\qquad\text{if }m<\frac{1}{3}n^{(t-1)/t}.
\]
Therefore, we have
\begin{equation}\label{eq:tree_TV_distance}
\left\|\cD^{(0)}_{n,m}-\cD^{(1)}_{n,m}\right\|< \frac{1}{3}\qquad\text{if }m< \frac{1}{3}n^{(t-1)/t}.
\end{equation}

On the other hand, since $H^{(1)}$ is $H$-free, any graph $G$ in the support of the distribution $\cH_{n}^{(1)}$ is $H$-free. Since $H^{(0)}$ contains a copy of $H$, for any graph $G$ in the support of $\cH_{n}^{(0)}$, at least $n$ edges must be removed from $G$ to make it $H$-free; in other words, the uniform distribution over the edge set of $G$ is $\varepsilon$-far from $H$-free, where $\varepsilon=2^{-(t-2)}t^{-1}$. Therefore, any sample-based distribution-free tester for $\cG^{H\textup{-free}}_{rn}$ with proximity parameter $\varepsilon=2^{-(t-2)}t^{-1}$ and sample complexity $m$ must distinguish $\cD^{(0)}_{n,m}$ from $\cD^{(1)}_{n,m}$ with probability at least $2/3$. By \eqref{eq:tree_TV_distance}, this requires $m\geq n^{(t-1)/t}/3$. We thus conclude that
\[
\sam\Big(\cG^{H\textup{-free}}_{rn},2^{-(t-2)}t^{-1}\Big)\geq \frac{1}{3}n^{(t-1)/t}.\qedhere
\]
\end{proof}

We next prove the lower bound for testing cliques, using the techniques in the proof of Theorem~\ref{thm:tree_lower_bound}.

\begin{theorem}
We have $\sam\big(\cG^{\textup{cliq}}_{6n},1/4\big)\geq n^{2/3}/3$.
\end{theorem}
\begin{proof}
Define two edge sets $E^{(0)},E^{(1)}\subseteq \binom{[6]}{2}$ as follows:
\[
E^{(0)}=\bigl\{\{1,2\},\{2,3\},\{3,4\},\{5,6\}\bigr\}\qquad\text{and}\qquad E^{(1)}=\big\{\{1,2\},\{2,3\},\{4,5\},\{5,6\}\big\}.
\]
Let $\cD^{\no}_{n}$ be the output distribution of the following process:
\begin{enumerate}
\item For each $i\in [n]$, pick a random bijection $\varphi_{i}:\{6i-5,\dots,6i\}\rightarrow\{1,2,\dots,6\}$.
\item Define a function $f:\binom{[6n]}{2}\rightarrow\{0,1\}$ as follows: for any $\{a,b\}\in \binom{[6n]}{2}$, let $f(\{a,b\})=1$ if and only if 
\[
\{a,b\}=\varphi_{i}^{-1}(\{1,2\})\qquad\text{or}\qquad\{a,b\}=\varphi_{i}^{-1}(\{3,4\})\qquad\text{for some }i\in [n].
\]
\item Let $\mu$ be the uniform distribution over
\[
\bigcup_{i\in [n]}\Bigl\{\varphi_{i}^{-1}(\{1,2\}),\varphi_{i}^{-1}(\{2,3\}),\varphi_{i}^{-1}(\{3,4\}),\varphi_{i}^{-1}(\{5,6\})\Bigr\}\subseteq \binom{[6n]}{2}.
\]
\item Output the pair $(f,\mu)$.
\end{enumerate}
Let $\cD^{\yes}_{n}$ be the output distribution of the following process:
\begin{enumerate}
\item For each $i\in [n]$, pick a random bijection $\varphi_{i}:\{6i-5,\dots,6i\}\rightarrow\{1,2,\dots,6\}$.
\item Define a function $f:\binom{[6n]}{2}\rightarrow\{0,1\}$ as follows: for any $\{a,b\}\in \binom{[6n]}{2}$, let $f(\{a,b\})=1$ if and only if $a,b$ belong to the vertex set
\[
\bigcup_{i\in [n]}\varphi_{i}^{-1}(\{1,2,4,5\})\subseteq [6n].
\]
\item Let $\mu$ be the uniform distribution over
\[
\bigcup_{i\in [n]}\Bigl\{\varphi_{i}^{-1}(\{1,2\}),\varphi_{i}^{-1}(\{2,3\}),\varphi_{i}^{-1}(\{4,5\}),\varphi_{i}^{-1}(\{5,6\})\Bigr\}\subseteq \binom{[6n]}{2}.
\]
\item Output the pair $(f,\mu)$.
\end{enumerate}

For any pair $(f,\mu)$ in the support of $\cD^{\yes}_{n}$, the graph $\big([6n],f^{-1}(1)\big)$ is a clique (of $4n$ vertices) and thus $f\in \cG^{\textup{cliq}}_{6n}$. On the other hand, it is easy to see that for any pair $(f,\mu)$ in the support of $\cD^{\no}_{n}$, we have
\[
\Prs{\{a,b\}\sim \mu}{f(\{a,b\})\neq g(\{a,b\})\big.}\geq \frac{1}{4}\qquad\text{for any }g\in\cG^{\textup{cliq}}_{6n}.
\]
However, using a birthday-paradox argument similar to the proof of Theorem~\ref{thm:tree_lower_bound}, one can show that in order to distinguish the no case $(f,\mu)\sim \cD^{\no}_{n}$ from the yes case $(f,\mu)\sim \cD^{\yes}_{n}$ with probability at least $2/3$, the number of $f$-labeled samples taken from $\mu$ must be at least $n^{2/3}/3$. We can thus conclude that $\sam\big(\cG^{\textup{cliq}}_{6n},1/4\big)\geq n^{2/3}/3$. 
\end{proof}

\section{Open Problems}

Let $\cH$ be a nonempty family of Boolean-valued functions on a finite domain $\Lambda$. We use $\mathsf{VC}(\cH)$ to denote the VC-dimension of $\cH$. A fundamental result in learning theory (see e.g. \cite{shalev2014understanding}) is that for any constant $\varepsilon\in (0,10^{-2})$ the number of $f$-labeled samples needed for PAC-learning a function $f\in \cH$ up to error $\varepsilon$ is $\Theta(\mathsf{VC(\cH)})$.\footnote{Even if query is allowed (see Remark~\ref{rem:tester_with_query}), the query complexity of PAC-learning is still $\Theta(\mathsf{VC}(\cH))$~\cite{turan1993lower}.} It was shown in \cite[Proposition 3.1.1]{goldreich1998property} that (sample-based distribution-free) testing cannot be harder than learning: we have
\[
\sam(\cH,\varepsilon)= O_{\varepsilon}(\mathsf{VC}(\cH))\qquad\text{for any constant }\varepsilon\in(0,1).
\]

A natural question is, for which families $\cH$ is distribution-free testing much easier than PAC-learning? For most of the well-studied function families $\cH_{n}$ (indexed by a parameter $n$ growing to infinity), such as linear threshold functions, conjunctions and decision lists on the hypercube $\{0,1\}^{n}$, there exists some constant $\varepsilon\in (0,1)$ such that $\sam(\cH_{n},\varepsilon)=\widetilde{\Omega}(\mathsf{VC}(\cH_{n}))$ (see \cite{blais2021vc} and \cite[Section 8]{chen2024distribution}). Blais, Ferreira Pinto Jr. and Harms \cite[Section 7]{blais2021vc} also gave two examples of natural function families $\cH_{n}$ for which there exists $c\in (0,1)$ such that 
\begin{equation}\label{eq:BFH_examples}
\sam(\cH_{n},\varepsilon)=O_{\varepsilon}\bigl(\mathsf{VC}(\cH_{n})^{1-c}\bigr)\qquad\text{for any constant }\varepsilon\in (0,1).
\end{equation}
Interestingly, in both of the examples given by \cite{blais2021vc}, the reason that testing can be more efficient than learning seems to be the \emph{birthday paradox}. 

Note that for the subgraph-freeness property $\cG^{H\textup{-free}}_{n}$ defined in the statement of Theorem~\ref{thm:tree-freeness}, we have $\mathsf{VC}(\cG^{H\textup{-free}}_{n})=\ex(n,H)$. Theorems~\ref{thm:main}, \ref{thm:tree-freeness} and~\ref{thm:subgraph_freeness_upper_bound} imply that \eqref{eq:BFH_examples} holds also for the subgraph-freeness property $\cH_{n}=\cG^{H\textup{-free}}_{n}$ if $H$ is a square, a tree with at least 2 edges\footnote{It is well-known that $\ex(n,H)=\Theta(n)$ for any tree $H$ with at least 2 edges.}, or a non-bipartite graph.\footnote{For non-bipartite graphs $H$ we easily have $\ex(n,H)=\Omega(n^{2})$.} Furthermore, our proofs seem to suggest that the reason we have \eqref{eq:BFH_examples} is again (variants of) the birthday paradox. This motivates the following conjecture:
\begin{conjecture}\label{conj:main}
For any connected simple graph $H$ with at least 2 edges, there exists a constant $c\in (0,1)$ such that
\[\sam\big(\cG^{H\textup{-free}}_{n},\varepsilon\big)=O_{\varepsilon}\bigl(\ex(n,H)^{1-c}\bigr)\qquad\text{for any constant }\varepsilon\in (0,1).\]
\end{conjecture}

As discussed in Section~\ref{subsec:general_framework}, if we restrict to sample-based testers with one-sided error, the tester must essentially be the ``canonical'' one. For two-sided-error testers, it is slightly less clear what is the best algorithm. Can there be a better two-sided-error tester for some properties?

\begin{problem}
Does there exist a connected simple graph $H$ such that testing $H$-freeness of edge distributions is much easier for two-sided-error testers than for one-sided-error testers, in terms of sample complexity?
\end{problem}

Another well-studied class of graph properties is the (homogeneous) partition properties \cite{fiat2021efficient}. Given a symmetric $0/1$-matrix $A\in \{0,1\}^{k\times k}$ and a graph $G=([n],E)$, we say that $G$ has the property $\cG^{A\textup{-part}}_{n}$ if there is a partition of the vertex set $\varphi:[n]\rightarrow [k]$ such that for any $\{a,b\}\in \binom{[n]}{2}$, we have $\{a,b\}\in E$ if and only if $A_{\varphi(a),\varphi(b)}=1$. We pose the following question:
\begin{problem}\label{prob:HPP}
Determine the sample complexity $\sam\big(\cG^{A\textup{-part}}_{n},\varepsilon\big)$ asymptotically in $n$ for any fixed symmetric $0/1$-matrix $A$.
\end{problem}
Note that the clique property $\cG^{\textup{cliq}}_{n}$ studied in Theorem~\ref{thm:clique} coincides with $\cG^{A\textup{-part}}_{n}$ for $A=\begin{bmatrix}1 & 0\\ 0& 0\end{bmatrix}$.

The power of ``query access'' in edge-distribution-free property testing has been left unexplored by this work. We pose the following questions:

\begin{problem}
If edge-query is allowed as in Remark~\ref{rem:tester_with_query}, can triangle-freeness be tested in $n^{4/3-\Omega(1)}$ queries? Can bipartiteness be tested in $n^{1-\Omega(1)}$ queries?
\end{problem}
\begin{problem}
If edge-query is allowed as in Remark~\ref{rem:tester_with_query}, what is the query complexity of testing threshold graphs (see Section~\ref{subsec:further_motivation} for the motivation)?
\end{problem}

\section*{Acknowledgements}

The author would like to thank Ronitt Rubinfeld and Asaf Shapira for many stimulating discussions during the development of this work, especially for bringing the papers \cite{alon2008testing} and \cite{timmons2015counterexample} to his attention.

\addcontentsline{toc}{section}{References}
\bibliographystyle{alpha}
\bibliography{reference}

\appendix
\section{Proof of Proposition~\ref{prop:downward_closed_equivalence}}\label{sec:proof_of_easy_equivalence}

\begin{proof}[Proof of Proposition~\ref{prop:downward_closed_equivalence}]
To obtain the first inequality, note that a sample $x$ from $\mu$ is equivalent to an $f$-labeled sample $(x,f(x))$ with $x\sim\mu$, if $f$ is the indicator function of $\supp(\mu)$. This obviously gives a reduction from the distribution testing problem in Definition~\ref{def:support_testing} to the function testing problem in Definition~\ref{def:distribution_free_testing}.

To obtain the second inequality, consider an algorithm $\cA$ testing whether $\supp(\mu)\in \cH$ using $m:=\dsam(\cH,\varepsilon)$ samples from any $\mu$. To test whether $f^{-1}(1)\in \cH$ with respect to $\mu$ in the sense of Definition~\ref{def:distribution_free_testing}, we run the following procedure:
\begin{enumerate}
\item Take samples $\big(x^{(i)},f(x^{(i)})\big)$ for $1\leq i\leq m'=\lceil 18m/\varepsilon\rceil$, where each $x^{(i)}$ is drawn from $\mu$.
\item If the number of $i\in [m']$ such that $f(x^{(i)})\neq 0$ is at most $9m\leq \varepsilon m'/2$, we accept.
\item Otherwise, we take the first $9m$ samples $x^{(i)}$ such that $f(x^{(i)})\neq 0$ and group them into 9 batches of size $m$. These are 9 batches of independent samples from $\mu$ conditioned on the set $f^{-1}(1)$. We then run $\cA$ on these 9 batches and take the majority vote to test whether the support of this conditional distribution (denoted by $\mu'$) belongs to $\cH$. 
\end{enumerate}

Completeness of the reduction: if $f\in \cH$, then since $\cH$ is downward-closed we have $\supp(\mu')\in \cH$, and thus step 3 accepts with probability at least $2/3$. 

Soundness of the reduction: If $f$ is $\varepsilon$-far from $\cH$ with respect to $\mu$, then since the identically-zero function belongs to $\cH$ we have $\Prs{x\sim \mu}{f(x)\neq 0}\geq \varepsilon$, and thus step 2 passes with probability at most $1/9$. We claim that $\|\mu'-\nu\|_{\mathrm{TV}}\geq \varepsilon$ for any distribution $\nu$ over $\Lambda$ such that $\supp(\nu)\in \cH$; in that case, step 3 passes with probability at most $1/6$ and the overall rejection probability of our procedure is at least $1-1/9-1/6=2/3$. Suppose that $\|\mu'-\nu\|_{\mathrm{TV}}<\varepsilon$ for some $\nu$ with $\supp(\nu)\in \cH$. Since $\supp(\mu')\subseteq f^{-1}(1)$ and $\cH$ is downward-closed, we can obviously assume $\supp(\nu)\subseteq f^{-1}(1)$. Let $g$ be the indicator function of $\supp(\nu)$ and we have
\[
\Pru{x\sim \mu}{f(x)\neq g(x)}\leq \Pru{x\sim \mu'}{f(x)\neq g(x)}=\Pru{x\sim \mu'}{x\not\in \supp(\nu)}\leq \|\mu'-\nu\|_{\mathrm{TV}}<\varepsilon,
\]
contradicting the assumption that $f$ is $\varepsilon$-far from $\cH$ with respect to $\mu$.
\end{proof}

\end{document}